\documentclass[pdflatex,sn-nature]{sn-jnl}

\usepackage{graphicx}%
\usepackage{multirow}%
\usepackage{amsmath,amssymb,amsfonts,mathtools,physics,physics}%
\usepackage{amsthm}%
\usepackage{mathrsfs}%
\usepackage[title]{appendix}%
\usepackage{xcolor}%
\usepackage{textcomp}%
\usepackage{manyfoot}%
\usepackage{booktabs}%
\usepackage{algorithm}%
\usepackage{algorithmicx}%
\usepackage{algpseudocode}%
\usepackage{listings}%
\usepackage{tabularx}
\usepackage{xltabular}

\usepackage{physics}      
\usepackage{braket}       
\usepackage{dsfont}       
\usepackage{upgreek}      
\usepackage{cancel}       
\usepackage{marvosym}     
\usepackage{tikz}         

\theoremstyle{thmstyleone}%
\newtheorem{theorem}{Theorem}[section]%
\newtheorem{corollary}{Corollary}[section]%
\newtheorem{lemma}{Lemma}[section]%

\theoremstyle{thmstyletwo}%
\newtheorem{remark}{Remark}[section]%

\theoremstyle{thmstylethree}%
\newtheorem{definition}{Definition}[section]%

\begin{document}

\title[QPI-DeepONet-MAC]{
QPI-DeepONet-MAC: A Scalable and Stable Hybrid Classical-Quantum Architecture for Physics-Informed Deep Operator Networks
}

\author*[1,2]{\fnm{Said} \sur{Lantigua}}\email{said.lantigua@acad.ufsm.br}

\author[1,3]{\fnm{José} \sur{Valencia}}\email{jcvaal@uel.br}

\author[1]{\fnm{Gilson} \sur{Giraldi}}\email{gilson@lncc.br}

\author[1]{\fnm{Renato} \sur{Portugal}}\email{portugal@lncc.br}

\author[2]{\fnm{Jonas} \sur{Maziero}}\email{jonas.maziero@ufsm.br}

\affil*[1]{\orgdiv{Departament of Physics, Center for Natural and Exact Sciences}, \orgname{Federal University of Santa Maria}, \orgaddress{\city{Santa Maria}, \postcode{97105-900}, \state{RS}, \country{Brazil}}}

\affil[2]{\orgname{National Laboratory for Scientific Computing (LNCC)}, \orgaddress{\street{Av. Getúlio Vargas, 333 - Quitandinha}, \city{Petrópolis}, \postcode{25651-075}, \state{RJ}, \country{Brazil}}}

\affil[3]{\orgdiv{Department of Mathematics}, \orgname{State University of Londrina (UEL)}, \orgaddress{\street{Celso Garcia Cid, PR-445, Km 380}, \city{Londrina}, \postcode{86057-970}, \state{PR}, \country{Brazil}}}

\abstract{
General operator learning for parametric partial differential equations (PDEs) is a fundamental challenge at the intersection of artificial intelligence and physics-based modeling. Physics-informed Deep Operator Networks (PI-DeepONets) incorporate governing equations and physical constraints into the learning process, but can exhibit optimization difficulties and limited expressivity in high-dimensional settings. Quantum extensions offer additional representational resources, yet their trainability can be hindered by barren plateaus, in which gradients may vanish rapidly with increasing circuit size or depth.
Here, we introduce QPI-DeepONet-MAC, a hybrid quantum-classical architecture that integrates parameterized quantum circuits into Physics-Informed DeepONets through a multiplicative-and-additive coupling (MAC) of classical neural representations and quantum expectation values. We establish analytically that the proposed architecture retains the universal approximation capability of operator networks under appropriate assumptions. We further derive bounds on the gradients with respect to the quantum parameters, providing conditions under which exponentially vanishing gradients can be avoided, together with a scaling criterion for trainability. In addition, we derive a lower bound on the number of sensors required to achieve a prescribed approximation accuracy, explicitly relating the sampling requirements to the regularity and dimensionality of the input functions and to the dynamical growth of the underlying PDE.
To regulate the quantum component during optimization, we introduce an informational regularization scheme based on the dynamics of quantum coherence and state fidelity. The resulting framework penalizes loss of coherence while controlling changes in the quantum-state trajectory, thereby providing a mechanism for preserving quantum-state structure during training. Together, these results provide a theoretical foundation for incorporating parameterized quantum circuits into physics-informed operator learning and establish QPI-DeepONet-MAC as a principled hybrid framework for the approximation of parametric PDE solution operators.
}

\maketitle

\section{Introduction}\label{sec:1}

In recent years, Physics-Informed Deep Operator Networks (PI-DeepONets) have emerged as a promising framework for learning solution operators of partial differential equations (PDEs) while incorporating the underlying physical laws into the training process \cite{Wang202100, Lu202100}. Unlike conventional numerical solvers, which generally require a new computational solution for each parameter configuration, operator-learning approaches aim to learn a mapping from problem parameters or input functions to the corresponding PDE solutions, thereby enabling efficient evaluation over multiple configurations once the model has been trained \cite{Raissi201900}. By incorporating automatic differentiation into the training process, PI-DeepONets can enforce the residuals of the governing equations together with initial and boundary conditions directly through the loss function \cite{Raissi201900, Wang202100}. This physics-informed formulation can substantially reduce the reliance on large collections of precomputed solution data required by purely data-driven operator-learning methods, while promoting consistency with the underlying physical constraints \cite{Karniadakis202100}.

Despite these advantages, PI-DeepONets inherit important challenges associated with both operator learning and physics-informed optimization. In particular, their expressive capacity and training efficiency can become limiting factors in high-dimensional problems, where optimization may become increasingly difficult and computationally demanding \cite{Wang202100, Lu202100, Giraldi202300, Huang202300, Sharma202300, Cheng202100, Wu201800, Chu202200, Bernardes-Vieira202100, Cai202100, Cai202101, Rudy201700}. These challenges are related to broader difficulties encountered by physics-informed neural networks (PINNs and CPINNs), which have nevertheless demonstrated the potential of physics-informed learning across applications such as fluid dynamics \cite{Jin202100}, materials science \cite{Kamalia202300}, and multiscale modeling \cite{Kashinath202100}. This motivates the exploration of alternative computational architectures capable of increasing representational flexibility while maintaining the physical structure imposed by the governing equations.


Quantum computing provides an alternative computational framework for machine learning and scientific computing, motivating quantum counterparts of several classical architectures, including quantum neural networks (QNNs), quantum convolutional neural networks (QCNNs), quantum physics-informed neural networks (QPINNs), quantum constrained physics-informed neural networks (QCPINNs), and quantum physics-informed DeepONets (QPI-DeepONets) \cite{Nielsen201000, Markidis202200, Trahan202400, Xiao202400}. Quantum algorithms have also been developed for the numerical treatment of differential equations \cite{Harrow200900, Sedykh202400, Norambuena202400, Clader201300, Costa201900, Berry201500, Cao201300, Wang201900, Arrazola201900, Childs202000}. In variational quantum models, parameterized quantum circuits (PQCs) provide flexible quantum feature maps and can represent nonlinear functions through measurements in high-dimensional Hilbert spaces \cite{Schuld201900, Perez-Salinas202000}. However, practical implementations remain constrained by the limited coherence, connectivity, gate fidelity, and circuit depth available on noisy intermediate-scale quantum (NISQ) hardware \cite{Scherer201700, Linden202200, Dong201000, Paine202300, Andrew202100, Babbush202300, Low202400, Broeckmann202400}. In addition, the optimization of PQCs can suffer from barren plateaus, characterized by gradients whose variance decreases exponentially with system size or circuit depth under broad classes of circuit and cost-function constructions, thereby making parameter optimization increasingly difficult \cite{McClean201800, Wang202101, Cerezo202100}. Several strategies have been proposed to alleviate barren plateaus, including architectural modifications and hybrid quantum-classical schemes \cite{Zhao202200}. Nevertheless, establishing rigorous trainability guarantees for quantum-enhanced physics-informed operator-learning architectures remains an open challenge. This motivates the development of architectures in which the quantum components are incorporated in a manner that preserves expressive power while providing analytical control over gradient scaling.


To address these challenges, we introduce Physics-Informed Deep Operator Networks with Multiplicative-Additive Coupling (QPI-DeepONet-MAC), a hybrid quantum-classical architecture that integrates parameterized quantum circuits into PI-DeepONets through structured multiplicative and additive couplings within the branch and trunk networks. Inspired by the QPINN-MAC architecture introduced in Ref. \cite{Lantigua202600}, the proposed framework combines classical operator-learning representations with quantum expectation values, while introducing a quantum-dependent additive coupling that enriches the resulting operator representation. We theoretically establish that QPI-DeepONet-MAC retains the universal approximation capability of the underlying DeepONet architecture under appropriate assumptions \cite{Gonon202500, Manzano202500}. Furthermore, we derive upper and lower bounds on the gradients with respect to the quantum parameters, yielding explicit scaling conditions under which the gradients remain detectable and providing theoretical conditions for mitigating barren-plateau effects. In addition, we derive a bound on the minimum number of sensors required to achieve a prescribed approximation accuracy for a class of well-posed parametric PDEs. Finally, we introduce a quantum-informational regularization scheme based on the dynamics of quantum coherence and fidelity, designed to discourage the loss of quantum-state structure during optimization and thereby preserve the expressive resources of the quantum component.

\section{Theoretical definition of the problem and the classical architecture of PI-DeepONets}\label{sec:2}
Let us begin by formally defining a system of partial differential equations (PDEs) on
\begin{equation}\label{eq:01}
    D \coloneq \Omega \times \mathbb{R}^{+},
\end{equation}
where \( D \) is an open set defined in terms of a spatial domain \( \Omega \subset \mathbb{R}^{n}\) with boundary \( \partial \Omega \), given the set
\begin{equation}\label{eq:02}
    \mbox{Diff} \coloneq \left\{ \hat{\mathcal{L}} : C_{\mathbb{R}^{M}}^{\infty} \left( D \right) \longrightarrow C_{\mathbb{R}^{M}}^{\infty} \left( D \right) \, | \, \hat{\mathcal{L}} \mbox{ is a differential operator} \right\},
\end{equation}
with
\begin{equation}\label{eq:03}
    C_{\mathbb{R}^{M}}^{\infty} \left( D \right) \coloneq \left\{ f : D \to \mathbb{R}^{M} \, | \, f \mbox{ is infinitely differentiable} \right\}.
\end{equation}
Given the above settings, a PDE can be defined as:
\begin{equation}\label{eq:04} 
    \mathcal{F} \left( \mathbf{y}, \mathbf{u}, \hat{\mathcal{L}}_{1} \left[ \cdot \right], \dots, \hat{\mathcal{L}}_{M} \left[ \cdot \right] \right) = \boldsymbol{\mathfrak{v}} \left( \tilde{\mathbf{y}} \right),
\end{equation}
where \( \mathbf{y} \in D \) is a \( (n + 1) \)-tuple of spatial and temporal coordinates, \( \tilde{\mathbf{y}} \in D \) is a \( (n + 1) \)-tuple of spatial coordinates whose temporal coordinate is zero (\(t = 0\)), \( \hat{\mathcal{L}}_{1} \left[ \cdot \right], \dots, \hat{\mathcal{L}}_{M} \left[ \cdot \right] \in \mbox{Diff} \), with vector functions \( \mathbf{u} = \left( u_{1}, \dots, u_{M} \right) \in C_{\mathbb{R}^{M}}^{\infty} \left( D \right) \), \( \boldsymbol{\mathfrak{v}} = \left( v_{1},\dots,v_{M} \right) \in C_{\mathbb{R}^{M}} \left( D \right) \) \cite{Rudin197600, Rudin198700, Tao201600, Brezis201100, Kiseliov198400, Piskunov197700, Piskunov197701, Demidovich197000, Spivak199400, Sunder200000, Pinkus201200}. And where
\begin{equation}\label{eq:05} 
    \mathcal{F} : D \times C_{\mathbb{R}^{M}}^{\infty} \left( D \right) \times \mbox{Diff}^{M} \longrightarrow C_{\mathbb{R}^{M}} \left( D \right),
\end{equation}
with \(\mbox{Diff}^{M} = \mbox{Diff} \times \mbox{Diff} \times \dots \times \mbox{Diff}\) and subject to the conditions:
\begin{equation}\label{eq:06} 
    \mathbf{u} \left( \mathbf{x},t_{0} \right) = \mathbf{u}_{0} \left( \mathbf{x} \right), \quad \left( \mathbf{x},t_{0} \right) \in D,
\end{equation}
and
\begin{equation}\label{eq:07} 
    \mathbf{u} \left( \mathbf{x}_{B},t \right) = \mathbf{u}_{B} \left( t \right), \quad \mathbf{x}_{B} \in \partial \Omega.
\end{equation}
In this regard, as demonstrated in Refs. \cite{Karniadakis202100, Lu202100}, a problem defined by equations (\ref{eq:04})–(\ref{eq:07}) can be solved using stacked or unstacked DeepONet architectures; see Figures \ref{fig:Fig1} and \ref{fig:Fig2}.
\begin{figure}[ht!]
    \centering
    \includegraphics[width=1.0\linewidth]{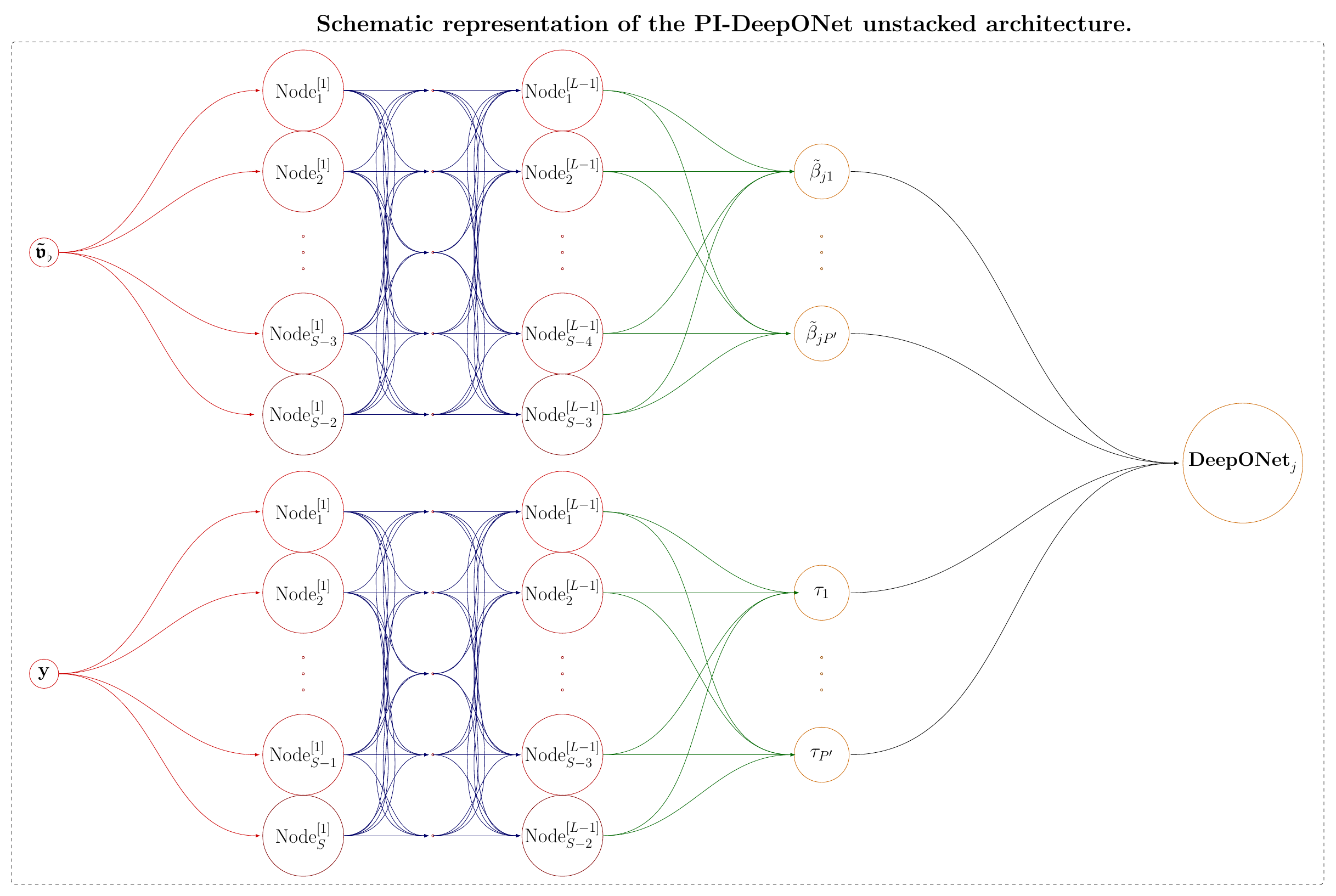}
    \caption{Schematic representation of the unstacked Physics-Informed DeepONet (PI-DeepONet) architecture. The diagram details the distinct input pathways for the branch (\(\boldsymbol{\tilde{\mathfrak{v}}}_{\flat}\)) and trunk (\(\mathbf{y}\)) variables, the hidden layers, the intermediate branch and trunk outputs (\(\boldsymbol{\hat{\beta}}_{j}\) and \(\boldsymbol{\tau}\)), and their final integration via the \(\mathbf{DeepONet}_{j}\) operator.}
    \label{fig:Fig1}
\end{figure}
\begin{figure}[ht!]
    \centering
    \includegraphics[width=1.0\linewidth]{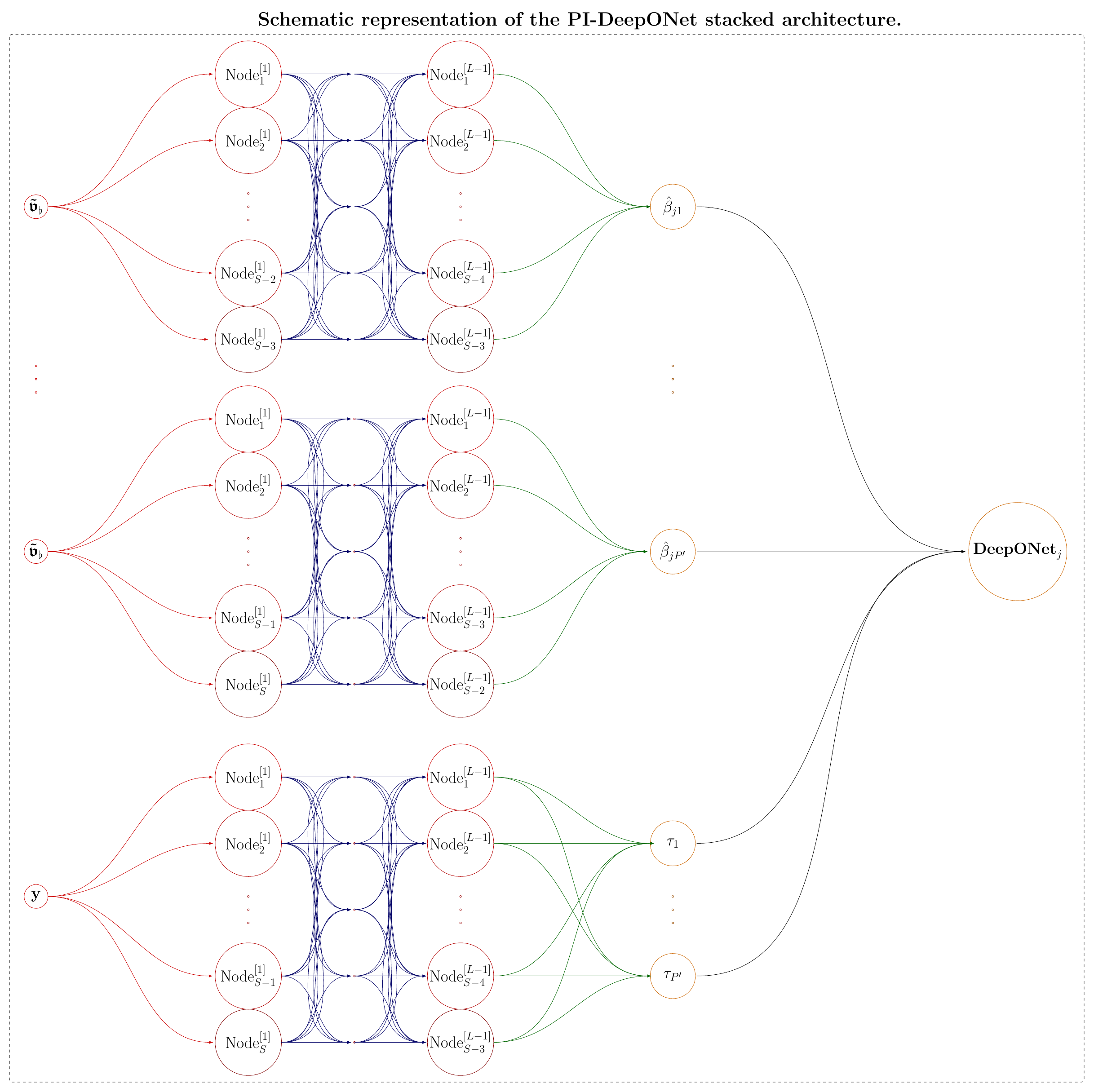}
    \caption{Schematic representation of the stacked Physics-Informed DeepONet (PI-DeepONet) architecture. This layout illustrates the multi-branch input configuration coupled with the trunk network, designed to process complex multi-field representations effectively.}
    \label{fig:Fig2}
\end{figure}
We would like to highlight that, for the sake of simplicity, this article focuses on the unstacked version of this classical architecture. However, as shown in Appendix \ref{sec:Ap-A}, these definitions are sufficiently general and also apply to the analogous stacked version of the architecture \cite{Hornik199100, Cybenko198900, Barron200200, Stone194800, Haykin199900, Bishop200600, Mohri201800, Anthony199900, Hastie198600}. That said, DeepONets are mathematically defined by the map:
\begin{equation}\label{eq:08}
    \mathbf{DeepONet} \left( \cdot, \boldsymbol{\mathcal{W}}_{\flat}^{\mathcal{B}} \right) \left( \cdot, \boldsymbol{\mathcal{W}}^{\mathcal{T}} \right): \mathbb{R}^{m M \times 1} \longrightarrow C_{\mathbb{R}^{M}} \left( D \right),
\end{equation}
with the correspondence rule
\begin{equation}\label{eq:09}
    \begin{split}
        \mathbf{DeepONet} \left( \boldsymbol{\tilde{\mathfrak{v}}}_{\flat}, \boldsymbol{\mathcal{W}}_{\flat}^{\mathcal{B}} \right) \left( \cdot, \boldsymbol{\mathcal{W}}^{\mathcal{T}} \right) &= \left( \mathbf{DeepONet}_{1} \left( \boldsymbol{\tilde{\mathfrak{v}}}_{\flat}, \boldsymbol{\mathcal{W}}_{\flat}^{\mathcal{B}} \right) \left( \cdot, \boldsymbol{\mathcal{W}}^{\mathcal{T}} \right), \dots \right. \\
        &\left. \hspace{1.0cm} \dots, \mathbf{DeepONet}_{M} \left( \boldsymbol{\tilde{\mathfrak{v}}}_{\flat}, \boldsymbol{\mathcal{W}}_{\flat}^{\mathcal{B}} \right) \left( \cdot, \boldsymbol{\mathcal{W}}^{\mathcal{T}} \right) \right).
    \end{split}
\end{equation}
Here, each component function \(\mathbf{DeepONet}_{j}\) is defined by the map
\begin{equation}\label{eq:10}
    \mathbf{DeepONet}_{j} \left( \cdot, \boldsymbol{\mathcal{W}}_{\flat}^{\mathcal{B}} \right) \left( \cdot, \boldsymbol{\mathcal{W}}^{\mathcal{T}} \right): \mathbb{R}^{m M \times 1} \longrightarrow C_{\mathbb{R}} \left( D \right),
\end{equation}
whose correspondence rule,
\begin{equation}\label{eq:11}
    \mathbf{DeepONet}_{j} \left( \cdot, \boldsymbol{\mathcal{W}}_{\flat}^{\mathcal{B}} \right) \left( \cdot, \boldsymbol{\mathcal{W}}^{\mathcal{T}} \right) = \boldsymbol{\hat{\beta}}_{j} \left( \cdot, \boldsymbol{\mathcal{W}}_{\flat}^{\mathcal{B}} \right) \cdot \boldsymbol{\tau} \left( \cdot, \boldsymbol{\mathcal{W}}^{\mathcal{T}} \right),
\end{equation}
is defined in terms of the branch and trunk sub-networks. These are defined by expressions
\begin{equation}\label{eq:12}
    \begin{split}
        \boldsymbol{\hat{\beta}} \left( \cdot, \boldsymbol{\mathcal{W}}_{\flat}^{\mathcal{B}} \right) : \mathbb{R}^{mM \times 1} & \longrightarrow \mathbb{R}^{N_{l}M \times 1} \\
        \boldsymbol{\tilde{\mathfrak{v}}}_{\flat} \hspace{0.25cm} & \longmapsto \boldsymbol{\hat{\beta}} \left( \boldsymbol{\tilde{\mathfrak{v}}}_{\flat}, \boldsymbol{\mathcal{W}}_{\flat}^{\mathcal{B}} \right) = \sigma \left( \mathbf{W}_{\flat}^{\mathcal{B} \, \left[ L \right]} \mathbf{x}^{\mathcal{B} \, \left[ L - 1 \right]} \left( \boldsymbol{\tilde{\mathfrak{v}}}_{\flat} \right) +  \mathbf{b}_{\flat}^{\mathcal{B} \, \left[ L \right]} \right),
    \end{split}
\end{equation}
and
\begin{equation}\label{eq:13}
    \begin{split}
        \boldsymbol{\tau} \left( \cdot, \boldsymbol{\mathcal{W}}^{\mathcal{T}} \right) : D & \longrightarrow \mathbb{R}^{N_{l} \times 1} \\
        \mathbf{y} & \longmapsto \boldsymbol{\tau} \left( \mathbf{y}, \boldsymbol{\mathcal{W}}^{\mathcal{T}} \right) = \sigma \left( \mathbf{W}^{\mathcal{T} \, \left[ \hat{L} \right]} \mathbf{x}^{\mathcal{T} \, \left[ \hat{L} - 1 \right]} \left( \mathbf{y} \right) +  \mathbf{b}^{\mathcal{T} \, \left[ \hat{L} \right]} \right).
    \end{split}
\end{equation}
where \(\boldsymbol{\tilde{\mathfrak{v}}}_{\flat}\) is given by the isomorphism
\begin{equation}\label{eq:14}
    \begin{split}
        \textbf{Vec} : \left(\mathbb{R}^{m\times 1}\right)^{M} & \longrightarrow \mathbb{R}^{mM \times 1} \\
        \boldsymbol{\tilde{\mathfrak{v}}} \hspace{0.75cm} & \longmapsto \boldsymbol{\tilde{\mathfrak{v}}}_{\flat} = \textbf{Vec} \left( \boldsymbol{\tilde{\mathfrak{v}}} \right) \coloneq \left[\boldsymbol{\tilde{v}}_{1}^{\intercal},\dots,\boldsymbol{\tilde{v}}_{M}^{\intercal}\right]^{\intercal},
    \end{split}
\end{equation}
for a matrix \( \boldsymbol{\tilde{\mathfrak{v}}} \in \left(\mathbb{R}^{m \times 1}\right)^{M} \), whose component \( \boldsymbol{\tilde{v}}_{j} \in \mathbb{R}^{m \times 1} \) is simply the discretization of the component \( v_{j} \) of \( \boldsymbol{\mathfrak{v}} \coloneq \left( v_{1},\dots,v_{M} \right) \in \mathbb{V}^{M} \) obtained via \( m \) sensors. Furthermore, it is crucial to highlight that in this work, each \( \mathbb{V} \) is a compact subset of \( C_{\mathbb{R}}\left(\mathbb{K}\right) \), \( \mathbb{K} \) is a compact subset of a Banach space \( \mathbb{B} \), and we are considering the uniform norm (supremum norm) on \( C_{\mathbb{R}}\left(\mathbb{K}\right) \). Furthermore, the existence of \( \mathbb{V} \) is guaranteed by the Arzelà-Ascoli theorem \cite{Rudin197600}, or, more precisely, \( \mathbb{V} \) is a compact subset of \( C_{\mathbb{R}}\left(\mathbb{K}\right) \) consisting of uniformly bounded and equicontinuous functions. Moreover, in the definitions presented above
\begin{equation}\label{eq:15} 
    \boldsymbol{\mathcal{W}}_{\flat}^{\mathcal{B}} \coloneq \bigcup_{l \in \mathcal{I}_{L_{\flat}}} \left\{ \mathbf{W}_{\flat}^{\mathcal{B} \, \left[ l \right]}, \mathbf{b}_{\flat}^{\mathcal{B} \, \left[ l \right]} \right\},
\end{equation}
is the set of weights and biases of the branch subnet with \(\mathbf{W}_{\flat}^{\mathcal{B} \, \left[ l \right]} \in \mathbb{R}^{N_{l}M \times N_{l-1}M}\), \(\mathbf{b}_{\flat}^{\mathcal{B} \, \left[ l \right]} \in \mathbb{R}^{N_{l}M \times 1}\), for \( N_{l-1},N_{l} \in \mathbb{N}^{\ast}\) and \( \mathcal{I}_{L_{\flat}} = \{0,1,\dots,L_{\flat}\}\) with \( L_{\flat} \in \mathbb{N}^{\ast}\)
of classic branch-net parameters. In a similar way,
\begin{equation}\label{eq:16} 
    \boldsymbol{\mathcal{W}}^{\mathcal{T}} \coloneq \bigcup_{l \in \mathcal{I}_{L}} \left\{ \mathbf{W}^{\mathcal{T} \, \left[ l \right]}, \mathbf{b}^{\mathcal{T} \, \left[ l \right]} \right\},
\end{equation}
is the set of weights and biases of the trunk subnet with \(\mathbf{W}^{\mathcal{T} \, \left[ l \right]} \in \mathbb{R}^{N_{l} \times N_{l-1}}\), \(\mathbf{b}^{\mathcal{T} \, \left[ l \right]} \in \mathbb{R}^{N_{l} \times 1}\), for \( N_{l-1},N_{l} \in \mathbb{N}^{\ast}\) and \( \mathcal{I}_{\hat{L}} = \{0,1,\dots,\hat{L}\}\) with \( \hat{L} \in \mathbb{N}^{\ast}\). Furthermore, \( \sigma: \mathbb{R} \to \mathbb{R} \) is an injective activation function satisfying the Tauber-Wiener (TW) condition \cite{Lu202100}, as required in Lemma \ref{lemma:Lem-B.1}. However, it is important to note that the extension to general activation functions will be addressed in future work.

Thus, for any \( \boldsymbol{\mathfrak{v}} \in \mathbb{V}^{M} \) with discretization \( \boldsymbol{\tilde{\mathfrak{v}}}_{\flat} \in \mathbb{R}^{m M \times 1} \) and coordinates \( \mathbf{y} \in D \), there is a set of weights and biases \( \boldsymbol{\mathcal{W}}_{\flat}^{\mathcal{B} \ast} \), \( \boldsymbol{\mathcal{W}}^{\mathcal{T} \ast} \) that allow \(\mathbf{DeepONet}_{j}\) to approximate the operator defined by:
\begin{equation}\label{eq:17}
    \begin{split}
        \mathcal{G}: \mathbb{V} &\longrightarrow C_{\mathbb{R}} \left( D \right) \\
        v \, &\longmapsto \mathcal{G} \left( v \right),
    \end{split}
\end{equation}
with arbitrary precision. That is, for any \( \epsilon_{j} > 0\) and \(j \in \mathcal{I}_{M} = \{1,2,\dots,M\}\), the following inequality is satisfied
\begin{equation}\label{eq:18}
    \left\Vert \mathbf{DeepONet}_{j} \left( \boldsymbol{\tilde{\mathfrak{v}}}_{\flat}, \boldsymbol{\mathcal{W}}_{\flat}^{\mathcal{B} \ast} \right) \left( \mathbf{y}, \boldsymbol{\mathcal{W}}^{\mathcal{T} \ast} \right) - \mathcal{G}_{j} \left( \boldsymbol{\mathfrak{v}} \right) \left( \mathbf{y} \right) \right\Vert_{2} \leq \epsilon_{j}.
\end{equation}
Consequently, using the properties of the triangular inequality and considering \( \epsilon = \epsilon_{1} + \dots + \epsilon_{M} \), \(\mathbf{DeepONet}\) can approximate the operator defined by:
\begin{equation}\label{eq:19}
    \begin{split}
        \boldsymbol{\mathcal{G}}: \mathbb{V}^{M} &\longrightarrow C_{\mathbb{R}^{M}} \left( D \right) \\
        \boldsymbol{\mathfrak{v}} \hspace{0.225cm} \, &\longmapsto \boldsymbol{\mathcal{G}} \left( \boldsymbol{\mathfrak{v}} \right) \left( \cdot \right) = \left(\mathcal{G}_{1} \left( \boldsymbol{\mathfrak{v}} \right) \left( \cdot \right), \dots, \mathcal{G}_{M} \left( \boldsymbol{\mathfrak{v}} \right) \left( \cdot \right) \right),
    \end{split}
\end{equation}
with arbitrary precision. That is, for any \( \epsilon > 0\), the following inequality is satisfied
\begin{equation}\label{eq:20}
    \left\Vert \mathbf{DeepONet} \left( \boldsymbol{\tilde{\mathfrak{v}}}_{\flat}, \boldsymbol{\mathcal{W}}_{\flat}^{\mathcal{B} \ast} \right) \left( \mathbf{y}, \boldsymbol{\mathcal{W}}^{\mathcal{T} \ast} \right) - \boldsymbol{\mathcal{G}} \left( \boldsymbol{\mathfrak{v}} \right) \left( \mathbf{y} \right) \right\Vert_{2,M} \leq \epsilon.
\end{equation}

Thus, concluding this brief yet necessary review of PI-DeepONets—whose formal and rigorous definition is presented in greater depth in Appendix \ref{sec:Ap-A}—we must define the  ``Physics-Informed'' component. Consequently, it is fair to say that for the architecture to approximate the operator \( \boldsymbol{\mathcal{G}} \)—which effectively maps to the function \( \mathbf{u} \) satisfying the problem defined by expressions (\ref{eq:04})–(\ref{eq:07})—it is necessary to find the optimal sets of weights and biases, namely: \( \boldsymbol{\mathcal{W}}_{\flat}^{\mathcal{B} \ast} \) and \( \boldsymbol{\mathcal{W}}^{\mathcal{T} \ast} \). This can be achieved during the optimization (minimization) process of a loss function composed of equations (\ref{eq:04}), (\ref{eq:06}), and (\ref{eq:07}), as well as, potentially, some known solutions \( \mathbf{u} \left( x_{\ell_{1}}, t_{\ell_{2}} \right) \), where \( \ell_{1} \in \mathcal{I}_{N_{1}} = \{1, 2, \dots, N_{1}\}\) and \( \ell_{2} \in \mathcal{I}_{N_{2}} = \{1, 2, \dots, N_{2}\} \) \cite{Feurer201500, Balaji201800, Snoek201200, Kotthoff201700, Elshawi201900, Pascanu201300, Knuth199700, Larkin197000}. This function is mathematically defined as:
\begin{equation}\label{eq:21}
    \mathcal{L} = \lambda_{IC}\mathcal{L}_{IC} + \lambda_{BC}\mathcal{L}_{BC} + \lambda_{PDE}\mathcal{L}_{PDE} + \lambda_{SOL}\mathcal{L}_{SOL},
\end{equation}
where
\begin{align}
    &\mathcal{L}_{IC} = \frac{1}{M N_{1}} \sum_{\substack{\ell_{1} \in \mathcal{I}_{N_{1}}}} \left\Vert \mathbf{DeepONet} \left( \boldsymbol{\tilde{\mathfrak{v}}}_{\flat}, \boldsymbol{\mathcal{W}}_{\flat}^{\mathcal{B}} \right) \left( \mathbf{x}_{\ell_{1}}, t_{0}, \boldsymbol{\mathcal{W}}^{\mathcal{T}} \right) - \mathbf{u}_{0} \left( \mathbf{x}_{\ell_{1}} \right) \right\Vert_{2,M}^{2}, \label{eq:22} \\
    &\mathcal{L}_{BC} = \frac{1}{M N_{2}} \sum_{\substack{\ell_{2} \in \mathcal{I}_{N_{2}}}} \left\Vert \mathbf{DeepONet} \left( \boldsymbol{\tilde{\mathfrak{v}}}_{\flat}, \boldsymbol{\mathcal{W}}_{\flat}^{\mathcal{B}} \right) \left( \mathbf{x}_{\mathcal{B}}, t_{\ell_{2}}, \boldsymbol{\mathcal{W}}^{\mathcal{T}} \right) - \mathbf{u}_{B} \left( t_{\ell_{2}} \right) \right\Vert_{2,M}^{2}, \label{eq:23} \\
    &\mathcal{L}_{PDE} = \frac{1}{M N_{1} N_{2}} \sum_{\substack{\ell_{1} \in \mathcal{I}_{N_{1}} \\ \ell_{2} \in \mathcal{I}_{N_{2}} }} \left\Vert \mathcal{F}\left( \mathbf{x}_{\ell_{1}}, t_{\ell_{2}}, \mathbf{DeepONet} \left( \boldsymbol{\tilde{\mathfrak{v}}}_{\flat}, \boldsymbol{\mathcal{W}}_{\flat}^{\mathcal{B}} \right) \left( \mathbf{x}_{\ell_{1}}, t_{\ell_{2}}, \boldsymbol{\mathcal{W}}^{\mathcal{T}} \right), \right. \right. \label{eq:24} \\
    &\left. \left. \hat{\mathcal{L}}_{1} \left[ \cdot \right], \dots, \hat{\mathcal{L}}_{M} \left[ \cdot \right] \right) - \mathbf{s} \left( \mathbf{x}_{\ell_{1}}, t_{\ell_{2}} \right) \right\Vert_{2,M}^{2},\nonumber \\
    &\mathcal{L}_{SOL} = \frac{1}{M N_{1} N_{2}} \sum_{\substack{\ell_{1} \in \mathcal{I}_{N_{1}} \\ \ell_{2} \in \mathcal{I}_{N_{2}} }} \left\Vert \mathbf{DeepONet} \left( \boldsymbol{\tilde{\mathfrak{v}}}_{\flat}, \boldsymbol{\mathcal{W}}_{\flat}^{\mathcal{B}} \right) \left( \mathbf{x}_{\ell_{1}}, t_{\ell_{2}}, \boldsymbol{\mathcal{W}}^{\mathcal{T}} \right) - \mathbf{u} \left( \mathbf{x}_{\ell_{1}}, t_{\ell_{2}} \right) \right\Vert_{2,M}^{2}. \label{eq:25}
\end{align}

\section{QPI-DeepONet-MAC Architecture}\label{sec:3}
First of all, it is important to highlight that, in the formulation presented below, we will use Einstein notation whenever possible, along with the index sets \(\mathcal{I}_{P} = \{1,\dots,P\}\) and \(\mathcal{I}_{N_{\nu}} = \{1,\dots,N_{\nu}\}\), where \(P, N_{\nu} \in \mathbb{N}^{\ast}\), and the sets of observables \(\mathbb{O}^{\nu}\), for \(\nu \in \{\mathcal{B}, \mathcal{T}\}\), which contain the Hermitian operators \(\hat{O}^{\nu}\) acting on \(\mathcal{H}_{\mathcal{Q}}^{\nu}\).

Analogous to the introduction of the classical PI-DeepONet model in Sec. \ref{sec:2}, we formally introduce our classical-quantum hybrid model, which we have named QPI-DeepONet-MAC and whose schematic representation is shown in Fig. \ref{fig:Fig3}.
\begin{figure}[ht!]
    \centering
    \includegraphics[width=1.0\linewidth]{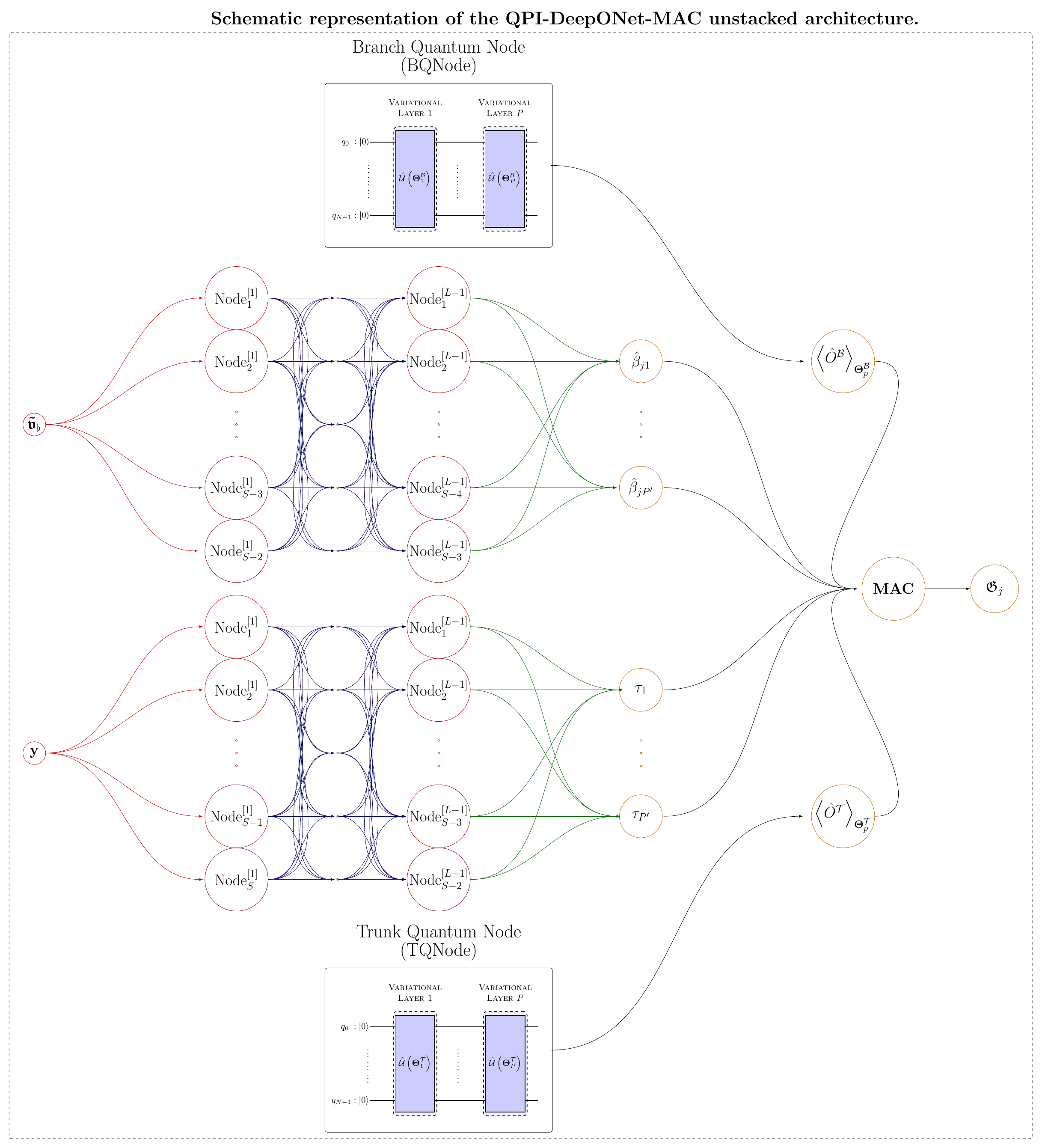}
    \caption{Schematic representation of the unstacked Quantum Physics-Informed DeepONet with Multiplication-And-Concatenation (QPI-DeepONet-MAC) architecture. The schematic highlights the integration of Branch Quantum Nodes (BQNodes) and Trunk Quantum Nodes (TQNodes) containing parameterized variational quantum layers, the measurement expectation value operators (\( \expval{\hat{O}^{\mathcal{B}}}_{\boldsymbol{\Theta}_{p}^{\mathcal{B}}} \) and \( \expval{\hat{O}^{\mathcal{T}}}_{\boldsymbol{\Theta}_{p}^{\mathcal{T}}} \)), and the central MAC combination block. Crucially, this unstacked framework represents the viable approach for practical implementation in near-term hybrid quantum-classical hardware, mitigating unintended entanglement and noise propagation. However, we emphasize that the above scheme has an analogue written in terms of a stacked DeepONet.}
    \label{fig:Fig3}
\end{figure} 
This model is constructed using two components. The first one is the density matrix \( \rho_{T} \), defined by:
\begin{equation}\label{eq:26}
    \rho_{T} =  \displaystyle{\bigotimes_{\substack{\mu \in \{ \mathcal{C}, \mathcal{Q} \} \\ \nu \in \{ \mathcal{B}, \mathcal{T} \}}} \rho_{\mu}^{\nu}}.
\end{equation}
That is, through the tensor product of the density sub-matrices 
\begin{equation}\label{eq:26-1}
    \rho_{\mathcal{C}}^{\mathcal{B}} \left( \boldsymbol{\tilde{\mathfrak{v}}}_{\flat}, \boldsymbol{\mathcal{W}}_{\flat}^{\mathcal{B}} \right) \coloneq \mathcal{C}_{k^{\mathcal{B}}l^{\mathcal{B}}} \left( \boldsymbol{\tilde{\mathfrak{v}}}_{\flat}, \boldsymbol{\mathcal{W}}_{\flat}^{\mathcal{B}} \right) \ketbra{k^{\mathcal{B}}}{l^{\mathcal{B}}},
\end{equation}
\begin{equation}\label{eq:26-2}
    \rho_{\mathcal{C}}^{\mathcal{T}} \left( \mathbf{y}, \boldsymbol{\mathcal{W}}^{\mathcal{T}} \right) \coloneq \mathcal{C}_{k^{\mathcal{T}}l^{\mathcal{T}}} \left( \mathbf{y}, \boldsymbol{\mathcal{W}}^{\mathcal{T}} \right) \ketbra{k^{\mathcal{T}}}{l^{\mathcal{T}}},
\end{equation}
\begin{equation}\label{eq:26-3}
    \rho_{\mathcal{Q}}^{\mathcal{B}} \left( \hat{O}^{\mathcal{B}}, \boldsymbol{\Theta}_{p}^{\mathcal{B}} \right) \coloneq \mathcal{P}_{m^{\mathcal{B}}n^{\mathcal{B}}} \left( \hat{O}^{\mathcal{B}}, \boldsymbol{\Theta}_{p}^{\mathcal{B}} \right) \ketbra{\Psi_{m^{\mathcal{B}}} \left( \boldsymbol{\Theta}_{p}^{\mathcal{B}} \right)}{\Psi_{n^{\mathcal{B}}} \left( \boldsymbol{\Theta}_{p}^{\mathcal{B}} \right)}
\end{equation}
and
\begin{equation}\label{eq:26-4}
    \rho_{\mathcal{Q}}^{\mathcal{T}} \left( \hat{O}^{\mathcal{T}}, \boldsymbol{\Theta}_{p}^{\mathcal{T}} \right) \coloneq \mathcal{P}_{m^{\mathcal{T}}n^{\mathcal{T}}} \left( \hat{O}^{\mathcal{T}}, \boldsymbol{\Theta}_{p}^{\mathcal{T}} \right) \ketbra{\Psi_{m^{\mathcal{T}}} \left( \boldsymbol{\Theta}_{p}^{\mathcal{T}} \right)}{\Psi_{n^{\mathcal{T}}} \left( \boldsymbol{\Theta}_{p}^{\mathcal{T}} \right)},
\end{equation}
which act in the composite Hilbert space \( \mathcal{H} = \mathcal{H}_{\mathcal{C}}^{\mathcal{B}} \otimes \mathcal{H}_{\mathcal{Q}}^{\mathcal{B}} \otimes \mathcal{H}_{\mathcal{C}}^{\mathcal{T}} \otimes \mathcal{H}_{\mathcal{Q}}^{\mathcal{T}} \) with \( \mathcal{H}_{\nu \in \{ \mathcal{C}, \mathcal{Q} \}}^{\nu \in \{ \mathcal{B}, \mathcal{T} \}} \) being \(\mathcal{H}_{\nu}^{\mu} \coloneq \mbox{span} \left\{ \ket{i_{k_{\mu}^{\nu}}}^{\otimes k_{\mu}^{\nu}} \right\}_{i_{k_{\mu}^{\nu}} \in \{0,1\} \, \forall \, k_{\mu}^{\nu} \in \{ 1, \dots, N \}}\) and \( N \in \mathbb{N}^{\ast} \) \cite{ReddyVadyala202300, Trahan202400, Xiao202400, Markidis202200, Siegl202400, Criado202300, Yadav202300, Sargsyana202400, Sedykh202400, Batra202400, Panichi202500}. Which encode classical and quantum information through the functions \(\mathcal{C}_{k^{\mathcal{B}}l^{\mathcal{B}}} \left( \cdot, \boldsymbol{\mathcal{W}}_{\flat}^{\mathcal{B}} \right): \mathbb{R}^{m M \times 1} \longrightarrow \mathbb{R}, \, \forall \, k^{\mathcal{B}},l^{\mathcal{B}} \in \mathcal{I}_{N_{\mathcal{B}}}\), \(\mathcal{C}_{k^{\mathcal{T}}l^{\mathcal{T}}} \left( \cdot, \boldsymbol{\mathcal{W}}^{\mathcal{T}} \right): D \longrightarrow \mathbb{R}, \, \forall \, k^{\mathcal{T}},l^{\mathcal{T}} \in \mathcal{I}_{N_{\mathcal{T}}}\), \(\mathcal{P}_{m^{\mathcal{B}}n^{\mathcal{B}}} \left( \cdot, \boldsymbol{\Theta}_{p}^{\mathcal{B}} \right): \mathbb{O}^{\mathcal{B}} \longrightarrow \mathbb{R}, \, \forall \, m^{\mathcal{B}},n^{\mathcal{B}} \in \mathcal{I}_{N_{\mathcal{B}}}, p \in \mathcal{I}_{P}\) and \(\mathcal{P}_{m^{\mathcal{T}}n^{\mathcal{T}}} \left( \cdot, \boldsymbol{\Theta}_{p}^{\mathcal{T}} \right): \mathbb{O}^{\mathcal{T}} \longrightarrow \mathbb{R}, \, \forall \, m^{\mathcal{T}},n^{\mathcal{T}} \in \mathcal{I}_{N_{\mathcal{T}}}, p \in \mathcal{I}_{P}\).

The second of these ingredients is the Hermitian coupling operator \( \hat{\mathfrak{O}} \) defined by:
\begin{equation}\label{eq:27} 
    \hat{\mathfrak{O}} = [\hat{A}_{\mathcal{C}}^{\mathcal{B}} + \mathds{1}] \otimes \hat{O}^{\mathcal{B}} \otimes [\hat{A}_{\mathcal{C}}^{\mathcal{T}} + \mathds{1}] \otimes \hat{O}^{\mathcal{T}} + \mathds{1} \otimes \hat{O}^{\mathcal{B}} \otimes \mathds{1} \otimes \hat{O}^{\mathcal{T}},
\end{equation}
where
\begin{equation}\label{eq:27-1}
    \hat{A}_{\mathcal{C}}^{\mathcal{B}} \left( \boldsymbol{\tilde{\mathfrak{v}}}_{\flat}, \boldsymbol{\mathcal{W}}_{\flat}^{\mathcal{B}} \right) = \mathcal{C}_{k^{\mathcal{B}}l^{\mathcal{B}}}^{\prime} \left( \boldsymbol{\tilde{\mathfrak{v}}}_{\flat}, \boldsymbol{\mathcal{W}}_{\flat}^{\mathcal{B}} \right) \ketbra{k^{\mathcal{B}}}{l^{\mathcal{B}}}
\end{equation}
and
\begin{equation}\label{eq:27-2}
    \hat{A}_{\mathcal{C}}^{\mathcal{T}} \left( \mathbf{y}, \boldsymbol{\mathcal{W}}^{\mathcal{T}} \right) = \mathcal{C}_{k^{\mathcal{T}}l^{\mathcal{T}}}^{\prime} \left( \mathbf{y}, \boldsymbol{\mathcal{W}}^{\mathcal{T}} \right) \ketbra{k^{\mathcal{T}}}{l^{\mathcal{T}}}
\end{equation}
are operators that also encode part of the classical information in the coefficient functions \(\mathcal{C}_{k^{\mathcal{B}}l^{\mathcal{B}}}^{\prime} \left( \cdot, \boldsymbol{\mathcal{W}}_{\flat}^{\mathcal{B}} \right): \mathbb{R}^{m M \times 1} \longrightarrow \mathbb{R}, \quad k^{\mathcal{B}},l^{\mathcal{B}} \in \mathcal{I}_{N}\) and \(\mathcal{C}_{k^{\mathcal{T}}l^{\mathcal{T}}}^{\prime} \left( \cdot, \boldsymbol{\mathcal{W}}^{\mathcal{T}} \right): D \longrightarrow \mathbb{R}, \quad k^{\mathcal{T}},l^{\mathcal{T}} \in \mathcal{I}_{N}\).

On the other hand, the parameterized quantum state \(\ket{\Psi(\boldsymbol{\Theta}_{p}^{\nu})}\), corresponding to the \(p\)-th computational node of the branch (\(\nu=\mathcal{B}\)) or trunk (\(\nu=\mathcal{T}\)) sub-networks, is constructed using a generic universal parameterized quantum circuit (PQC) \cite{Farhi201800, Grant201800, Ding202400, Lloyd202000, Beer202200, Schuld202100, Peruzzo201400}. For a circuit with \(N_{\nu}\) qubits and depth \(\mathcal{N}_{\nu}\), the local quantum parameters are encoded in the matrix \(\boldsymbol{\Theta}_{p}^{\nu} \in \mathbb{R}^{N_{\nu} \times \mathcal{N}_{\nu}}\). Consequently, the global set of all trainable quantum angles for each sub-network is fully defined by the parameter tensor \(\boldsymbol{\Xi}^{\nu} \in \mathbb{R}^{P \times N_{\nu} \times \mathcal{N}_{\nu}}\). Quantum evolution is driven by the alternation between continuous parametric layers \(\widehat{\mathcal{VL}}_{\text{Param}}\)—containing standard rotation and controlled-phase gates—and discrete layers \(\widehat{\mathcal{VL}}_{\text{Discr}, p,k}\), constructed from elements of the Clifford group augmented with Toffoli gates. This structure allows us to define the global unitary operator for the \(p\)-th node as
\begin{equation}\label{eq:28}
    \widehat{U}\left(\boldsymbol{\Theta}_{p}^{\nu}\right) = \prod_{k \in \mathcal{I}_{\mathcal{N}_{\nu}}} \left[ \widehat{\mathcal{VL}}_{\text{Param}} \left( \theta_{p,k}^{\nu} \right) \right] \left[ \widehat{\mathcal{VL}}_{\text{Discr}, p,k} \right],
\end{equation}
which remains entirely agnostic regarding the specific choice of quantum hardware or circuit analysis. Finally, by applying this general unitary operator to the reference vacuum state \(\ket{0}^{\otimes N_\nu}\), we obtain the complete set of parameterized quantum states for our architecture:
\begin{equation}\label{eq:29}
    \left\{ \ket{\Psi \left( \boldsymbol{\Theta}_{p}^{\nu} \right)} = \widehat{U} \left( \boldsymbol{\Theta}_{p}^{\nu} \right) \ket{0}^{\otimes N_\nu} \right\}_{\nu \in \{ \mathcal{B}, \mathcal{T} \}, p \in \mathcal{I}_{P}}.
\end{equation}

This way, when taking the dash \(\Tr{ \rho_{T} \hat{\mathcal{D}}}\), we obtain the map:
\begin{equation}\label{eq:30}
    \mathfrak{G}_{j} \left( \cdot, \hat{O}^{\mathcal{B}}, \boldsymbol{\mathcal{W}}_{\flat}^{\mathcal{B}}, \boldsymbol{\Theta}^{\mathcal{B}}_{p} \right) \left( \cdot, \hat{O}^{\mathcal{T}}, \boldsymbol{\mathcal{W}}^{\mathcal{T}}, \boldsymbol{\Theta}^{\mathcal{T}}_{p} \right): \mathbb{R}^{m M \times 1} \longrightarrow C_{\mathbb{R}} \left( D \right),
\end{equation}
whose matching rule,
\begin{equation}\label{eq:31}
    \begin{split}
        &\mathfrak{G}_{j} \left( \cdot, \hat{O}^{\mathcal{B}}, \boldsymbol{\mathcal{W}}_{\flat}^{\mathcal{B}}, \boldsymbol{\Theta}^{\mathcal{B}}_{p} \right) \left( \cdot, \hat{O}^{\mathcal{T}}, \boldsymbol{\mathcal{W}}^{\mathcal{T}}, \boldsymbol{\Theta}^{\mathcal{T}}_{p} \right) = \mathfrak{B}_{jp^{\prime}} \left( \cdot, \hat{O}^{\mathcal{B}}, \boldsymbol{\mathcal{W}}_{\flat}^{\mathcal{B}}, \boldsymbol{\Theta}^{\mathcal{B}}_{p^{\prime}} \right) \\
        &\hspace{3.5cm} \times \mathfrak{T}_{p^{\prime}}\left( \cdot, \hat{O}^{\mathcal{T}}, \boldsymbol{\mathcal{W}}^{\mathcal{T}}, \boldsymbol{\Theta}^{\mathcal{T}}_{p} \right) + \mathfrak{b}_{j} \left( \hat{O}^{\mathcal{B}}, \hat{O}^{\mathcal{T}}, \boldsymbol{\Theta}^{\mathcal{B}}_{p}, \boldsymbol{\Theta}^{\mathcal{T}}_{p} \right),
    \end{split}
\end{equation}
is written in terms of functions
\begin{equation}\label{eq:32}
    \begin{split}
        &\mathfrak{B}_{jp^{\prime}} \left( \cdot, \hat{O}^{\mathcal{B}}, \boldsymbol{\mathcal{W}}_{\flat}^{\mathcal{B}}, \boldsymbol{\Theta}^{\mathcal{B}}_{p} \right) = \left[\hat{\beta}_{jp^{\prime}} \left( \cdot, \boldsymbol{\mathcal{W}}_{\flat}^{\mathcal{B}} \right) + 1\right] \\
        &\hspace{3.5cm} \times \mathcal{P}_{m^{\mathcal{B}}n^{\mathcal{B}}} \bra{\Psi_{m^{\mathcal{B}}} \left( \boldsymbol{\Theta}^{\mathcal{B}}_{p} \right) } \hat{O}^{\mathcal{B}} \ket{\Psi_{n^{\mathcal{B}}} \left( \boldsymbol{\Theta}^{\mathcal{B}}_{p} \right) },
    \end{split}
\end{equation}
\begin{equation}\label{eq:33}
    \begin{split}
        &\mathfrak{T}_{p^{\prime}}\left( \cdot, \hat{O}^{\mathcal{T}}, \boldsymbol{\mathcal{W}}^{\mathcal{T}}, \boldsymbol{\Theta}^{\mathcal{T}}_{p} \right) = \left[\tau_{p^{\prime}} \left( \cdot, \boldsymbol{\mathcal{W}}^{\mathcal{T}} \right) + 1\right] \\
        &\hspace{3.5cm} \times \mathcal{P}_{m^{\mathcal{T}}n^{\mathcal{T}}} \bra{\Psi_{m^{\mathcal{T}} } \left( \boldsymbol{\Theta}^{\mathcal{T}}_{p} \right) } \hat{O}^{\mathcal{T}} \ket{\Psi_{n^{\mathcal{T}}} \left( \boldsymbol{\Theta}^{\mathcal{T}}_{p} \right)},
    \end{split}
\end{equation}
and
\begin{equation}\label{eq:34}
    \begin{split}
        &\mathfrak{b}_{j} \left( \hat{O}^{\mathcal{B}}, \hat{O}^{\mathcal{T}}, \boldsymbol{\Theta}^{\mathcal{B}}_{p}, \boldsymbol{\Theta}^{\mathcal{T}}_{p} \right) = \mathcal{P}_{m^{\mathcal{B}}n^{\mathcal{B}}} \bra{\Psi_{m^{\mathcal{T}} } \left( \boldsymbol{\Theta}^{\mathcal{B}}_{p} \right) } \hat{O}^{\mathcal{B}} \ket{\Psi_{n^{\mathcal{B}}} \left( \boldsymbol{\Theta}^{\mathcal{B}}_{p} \right) } \\
        &\hspace{3.5cm} \times \mathcal{P}_{m^{\mathcal{T}}n^{\mathcal{T}}}\bra{\Psi_{m^{\mathcal{T}} } \left( \boldsymbol{\Theta}^{\mathcal{T}}_{p} \right) } \hat{O}^{\mathcal{T}} \ket{\Psi_{n^{\mathcal{T}}} \left( \boldsymbol{\Theta}^{\mathcal{T}}_{p} \right) }.
    \end{split}
\end{equation}
We can approximate the \( j \)-th component of \( \boldsymbol{\mathcal{G}} \) with arbitrary precision. That is, for any \( \epsilon_{j} > 0 \) with \( j \in \mathcal{I}_{M} = \left\{1, \dots, M \right\} \), \( \boldsymbol{\mathfrak{v}} \in \mathbb{V}^{M} \), \( \boldsymbol{\tilde{\mathfrak{v}}}_{\flat} \in \mathbb{R}^{m M \times 1} \), and \( \mathbf{y} \in D \) exists \( \hat{O}^{\mathcal{B}} \in \mathbb{O}^{\mathcal{B}} \), \( \hat{O}^{\mathcal{T}} \in \mathbb{O}^{\mathcal{T}} \) and a set of weights and biases \( \boldsymbol{\mathcal{W}}_{\flat}^{\mathcal{B} \ast} \), \( \boldsymbol{\mathcal{W}}^{\mathcal{T} \ast} \), \( \boldsymbol{\Theta}_{p}^{\mathcal{B} \ast} \), \( \boldsymbol{\Theta}_{p}^{\mathcal{T} \ast} \) such that
\begin{equation}\label{eq:35}
    \left\Vert \mathfrak{G}_{j} \left( \boldsymbol{\tilde{\mathfrak{v}}}_{\flat}, \hat{O}^{\mathcal{B}}, \boldsymbol{\mathcal{W}}_{\flat}^{\mathcal{B} \ast}, \boldsymbol{\Theta}^{\mathcal{B} \ast}_{p} \right) \left( \cdot, \hat{O}^{\mathcal{T}}, \boldsymbol{\mathcal{W}}^{\mathcal{T} \ast}, \boldsymbol{\Theta}^{\mathcal{T} \ast}_{p} \right) - \mathcal{G}_{j} \left( \boldsymbol{\mathfrak{v}} \right) \left( \cdot \right) \right\Vert_{\infty} \leq \epsilon_{j}.
\end{equation}
Therefore, when using the properties of the triangular inequality and considering \( \epsilon = \epsilon_{1} + \dots + \epsilon_{M} \), we finally obtain
\begin{equation}\label{eq:36}
    \left\Vert \boldsymbol{\mathfrak{G}} \left( \boldsymbol{\tilde{\mathfrak{v}}}_{\flat}, \hat{O}^{\mathcal{B}}, \boldsymbol{\mathcal{W}}_{\flat}^{\mathcal{B} \ast}, \boldsymbol{\Xi}^{\mathcal{B} \ast} \right) \left( \cdot, \hat{O}^{\mathcal{T}}, \boldsymbol{\mathcal{W}}^{\mathcal{T} \ast}, \boldsymbol{\Xi}^{\mathcal{T} \ast} \right) - \boldsymbol{\mathcal{G}} \left( \boldsymbol{\mathfrak{v}} \right) \left( \cdot \right) \right\Vert_{\infty,M} \leq \epsilon,
\end{equation}
and thus, \( \boldsymbol{\mathfrak{G}} \) is defined by the map
\begin{equation}\label{eq:37}
    \boldsymbol{\mathfrak{G}} \left( \cdot, \hat{O}^{\mathcal{B}}, \boldsymbol{\mathcal{W}}_{\flat}^{\mathcal{B}}, \boldsymbol{\Xi}^{\mathcal{B}} \right) \left( \cdot, \hat{O}^{\mathcal{T}}, \boldsymbol{\mathcal{W}}^{\mathcal{T}}, \boldsymbol{\Xi}^{\mathcal{T}} \right): \mathbb{R}^{m M \times 1} \longrightarrow C_{\mathbb{R}^{M}} \left( D \right),
\end{equation}
with matching rule
\begin{equation}\label{eq:38}
    \begin{split}
        &\boldsymbol{\mathfrak{G}} \left( \boldsymbol{\tilde{\mathfrak{v}}}_{\flat}, \hat{O}^{\mathcal{B}}, \boldsymbol{\mathcal{W}}_{\flat}^{\mathcal{B} \ast}, \boldsymbol{\Xi}^{\mathcal{B} \ast} \right) \left( \cdot, \hat{O}^{\mathcal{T}}, \boldsymbol{\mathcal{W}}^{\mathcal{T} \ast},\boldsymbol{\Xi}^{\mathcal{T} \ast} \right) = \\
        &\hspace{1.5cm} \left( \mathfrak{G}_{1} \left( \boldsymbol{\tilde{\mathfrak{v}}}_{\flat}, \hat{O}^{\mathcal{B}}, \boldsymbol{\mathcal{W}}_{\flat}^{\mathcal{B} \ast}, \boldsymbol{\Theta}^{\mathcal{B} \ast}_{1} \right) \left( \cdot, \hat{O}^{\mathcal{T}}, \boldsymbol{\mathcal{W}}^{\mathcal{T} \ast}, \boldsymbol{\Theta}^{\mathcal{T} \ast}_{1} \right), \dots \right.\\.
        &\left. \hspace{3.5cm} \dots, \mathfrak{G}_{M} \left( \boldsymbol{\tilde{\mathfrak{v}}}_{\flat}, \hat{O}^{\mathcal{B}}, \boldsymbol{\mathcal{W}}_{\flat}^{\mathcal{B} \ast}, \boldsymbol{\Theta}^{\mathcal{B} \ast}_{P} \right) \left( \cdot, \hat{O}^{\mathcal{T}}, \boldsymbol{\mathcal{W}}^{\mathcal{T}\ast}, \boldsymbol{\Theta}^{\mathcal{T} \ast}_{P} \right) \right),
    \end{split}
\end{equation}
Consequently, in a manner analogous to the previously presented classical DeepONet model, we define the ``Physics-Informed'' component that enables our hybrid architecture to approximate \( \boldsymbol{\mathcal{G}} \)—which maps the solution function \( \mathbf{u} \) of the problem defined by expressions (\ref{eq:04})–(\ref{eq:07})—via the loss function
\begin{equation}\label{eq:39}
    \tilde{\mathcal{L}} = \lambda_{IC} \tilde{\mathcal{L}}_{IC} + \lambda_{BC} \tilde{\mathcal{L}}_{BC} + \lambda_{PDE} \tilde{\mathcal{L}}_{PDE} + \lambda_{SOL} \hat{\mathcal{L}}_{SOL},
\end{equation}
where
\begin{align}
    &\tilde{\mathcal{L}}_{IC} = \frac{1}{M N_{1}} \sum_{\substack{\ell_{1} \in \mathcal{I}_{N_{1}}}} \left\Vert \boldsymbol{\mathfrak{G}} \left( \boldsymbol{\tilde{\mathfrak{v}}}_{\flat}, \hat{O}^{\mathcal{B}}, \boldsymbol{\mathcal{W}}_{\flat}^{\mathcal{B}}, \boldsymbol{\Xi}^{\mathcal{B}} \right) \left( \mathbf{x}_{\ell_{1}}, t_{0}, \hat{O}^{\mathcal{T}}, \boldsymbol{\mathcal{W}}^{\mathcal{T}}, \boldsymbol{\Xi}^{\mathcal{T}} \right) - \mathbf{u}_{0} \left( \mathbf{x}_{\ell_{1}} \right) \right\Vert_{2,M}^{2}, \label{eq:40} \\
    &\tilde{\mathcal{L}}_{BC} = \frac{1}{M N_{2}} \sum_{\substack{\ell_{2} \in \mathcal{I}_{N_{2}}}} \left\Vert \boldsymbol{\mathfrak{G}} \left( \boldsymbol{\tilde{\mathfrak{v}}}_{\flat}, \hat{O}^{\mathcal{B}}, \boldsymbol{\mathcal{W}}_{\flat}^{\mathcal{B}}, \boldsymbol{\Xi}^{\mathcal{B}} \right) \left( \mathbf{x}_{\mathcal{B}}, t_{\ell_{2}}, \hat{O}^{\mathcal{T}}, \boldsymbol{\mathcal{W}}^{\mathcal{T}}, \boldsymbol{\Xi}^{\mathcal{T}} \right) - \mathbf{u}_{B} \left( t_{\ell_{2}} \right) \right\Vert_{2,M}^{2}, \label{eq:41} \\
    &\tilde{\mathcal{L}}_{PDE} = \frac{1}{M N_{1} N_{2}} \sum_{\substack{\ell_{1} \in \mathcal{I}_{N_{1}} \\ \ell_{2} \in \mathcal{I}_{N_{2}} }} \bigg\Vert \mathcal{F}\left( \mathbf{x}_{\ell_{1}}, t_{\ell_{2}}, \boldsymbol{\mathfrak{G}} \left( \boldsymbol{\tilde{\mathfrak{v}}}_{\flat}, \hat{O}^{\mathcal{B}}, \boldsymbol{\mathcal{W}}_{\flat}^{\mathcal{B}}, \boldsymbol{\Xi}^{\mathcal{B}} \right) \left( \mathbf{x}_{\ell_{1}}, t_{\ell_{2}}, \hat{O}^{\mathcal{T}}, \boldsymbol{\mathcal{W}}^{\mathcal{T}}, \boldsymbol{\Xi}^{\mathcal{T}} \right),\right. \nonumber \\
    &\left. \hspace{7.125cm}  \hat{\mathcal{L}}_{1} \left[ \cdot \right], \dots, \hat{\mathcal{L}}_{M} \left[ \cdot \right] \right) - \mathbf{s} \left( \mathbf{x}_{\ell_{1}}, t_{\ell_{2}} \right) \bigg\Vert_{2,M}^{2}, \label{eq:42} \\
    &\tilde{\mathcal{L}}_{SOL} = \frac{1}{M N_{1} N_{2}} \sum_{\substack{\ell_{1} \in \mathcal{I}_{N_{1}} \\ \ell_{2} \in \mathcal{I}_{N_{2}} }} \left\Vert \boldsymbol{\mathfrak{G}} \left( \boldsymbol{\tilde{\mathfrak{v}}}_{\flat}, \hat{O}^{\mathcal{B}}, \boldsymbol{\mathcal{W}}_{\flat}^{\mathcal{B}}, \boldsymbol{\Xi}^{\mathcal{B}} \right) \left( \mathbf{x}_{\ell_{1}}, t_{\ell_{2}}, \hat{O}^{\mathcal{T}}, \boldsymbol{\mathcal{W}}^{\mathcal{T}}, \boldsymbol{\Xi}^{\mathcal{T}} \right) \right. \nonumber \\
    &\left. \hspace{9.85cm}- \mathbf{u} \left( \mathbf{x}_{\ell_{1}}, t_{\ell_{2}} \right) \right\Vert_{2,M}^{2}, \label{eq:43}
\end{align}
and where again the hyperparameters \( \lambda_{IC} \), \( \lambda_{BC} \), \( \lambda_{PDE} \), and \( \lambda_{SOL} \) weight the physical constraints imposed by the system of interest.

Therefore, the feasibility and effectiveness of the model depend on determining the optimal set of classical and quantum parameters, denoted by \(\boldsymbol{\mathcal{W}}_{\flat}^{\mathcal{B} \ast}\), \(\boldsymbol{\mathcal{W}}^{\mathcal{T} \ast}\), \(\boldsymbol{\Xi}^{\mathcal{B} \ast}\), and \(\boldsymbol{\Xi}^{\mathcal{T} \ast}\). As in the classical approach, convergence toward these ideal parameters is achieved by minimizing the total loss function. During this optimization process, the model is trained to simultaneously satisfy the governing equation (\ref{eq:04}), the initial and boundary conditions (\ref{eq:06}) and (\ref{eq:07}), and—when available—data from known empirical or analytical solutions \(\mathbf{u} \left( \mathbf{x}_{\ell_{1}}, t_{\ell_{2}} \right)\). To ensure stable convergence, gradient propagation and the exact calculation of derivatives associated with the physical constraints are efficiently performed using Automatic Differentiation (AD) \cite{Raissi201900, Wang202100, Karniadakis202100}.


\section{Theoretical results}\label{sec:4}
The classical-quantum hybrid model QPI-DeepONet-MAC, previously presented, is formally underpinned by a rigorous theoretical framework that ensures its expressivity, trainability, and information efficiency. To properly establish these mathematical foundations, we first clarify the regularity assumptions for functions defined on compact domains.
\begin{remark}\label{remark:3.1}
    Although the definition of class \(C^{k}\) is based on open, consequently not necessarily compact sets, in this work we adapt this definition to compact sets as follows: we say that a function \(f: K \to \mathbb{R}^{M}\) defined on a generic compact set \(K\) is of class \(C^{k}\), written herein as \(f\in C_{\mathbb{R}^{M}}^{k}\left(K\right)\), if it admits a regular extension beyond its structural boundaries. Formally, this means that there exists an open set \(U\) containing \(K\) (\(K \subset U\)) and a function \(F: U \to \mathbb{R}^{M}\) of class \(C_{\mathbb{R}^{M}}^{k}\left(U\right)\) in the traditional sense on the open set \(U\), such that \(F(x) = f(x)\) for all \(x \in K\). This classical approach bypasses the geometric constraint of the boundary of \(K\), inheriting the regularity and the computation of derivatives from the open neighborhood \(U\).
\end{remark}
Supported by these topological considerations, we can now formalize the core representation capacity of our architecture \cite{Ismailov202600, Chen199500}. The following theorem establishes that the QPI-DeepONet-MAC retains the universal approximation property, ensuring its theoretical capability to map the continuous solutions of differential equations, analogous to what is shown in Refs. \cite{Lu202100, Lanthaler202200, Jiao202600, Ouyang202600, Chen202600, Zhang202500, Amin202600, Wang202100, Lantigua202600, Afrah202500, Perez-Salinas202500, Goto202100, Cai202102, Lin202300}.
\begin{theorem}\label{theorem:Teo-3.1}
    Let \( \mathcal{K} \subset D \) be a compact Hausdorff space and \( \mathbb{V} \subset C_{\mathbb{R}} \left( \mathbb{K} \right) \) a compact subset of functions. Let \( \mathfrak{G}_{j} \) be functions as in Lemma \ref{lemma:Lem-B.1}. Given an operator \( \boldsymbol{\mathcal{G}}: \mathbb{V}^{M} \longrightarrow C_{\mathbb{R}^{M}} \left( \mathcal{K} \right) \) capable of mapping component-wise a function \( \mathbf{u}: \mathcal{K} \longrightarrow \mathbb{R}^{M} \), such that \(\mathbf{u} \in C_{\mathbb{R}^{M}}^{\infty} \left( \mathcal{K} \right)\) and it is a solution of (\ref{eq:04}). Then, for any \( \epsilon_{j} > 0 \), \( \boldsymbol{\mathfrak{v}} \in \mathbb{V}^{M} \), \( \boldsymbol{\tilde{\mathfrak{v}}}_{\flat} \in \mathbb{R}^{m M \times 1} \), \( \mathbf{y} \in \mathcal{K} \), there exist \( \hat{O}^{\mathcal{B}} \in \mathbb{O}^{\mathcal{B}} \), \( \hat{O}^{\mathcal{T}} \in \mathbb{O}^{\mathcal{T}} \) and a set of weights and biases \( \boldsymbol{\mathcal{W}}_{\flat}^{\mathcal{B} \ast} \), \( \boldsymbol{\mathcal{W}}^{\mathcal{T} \ast} \), \( \boldsymbol{\Theta}_{p}^{\mathcal{B} \ast} \), \( \boldsymbol{\Theta}_{p}^{\mathcal{T} \ast} \), such that (\ref{eq:35}) is satisfied. Consequently, for any \( \epsilon > 0 \), the operator \( \boldsymbol{\mathfrak{G}} \) approximates \( \boldsymbol{\mathcal{G}} \) with the total error (\ref{eq:36}) in networks whose total number of parameters satisfies \( \mathcal{O} \left( \mathcal{N_{P_{C}}^{\mathcal{B}}} \right) \approx \mathcal{O} \left( \mathcal{N}^{\mathcal{B}} N^{\mathcal{B}} \right) \) and \( \mathcal{O} \left( \mathcal{N_{P_{C}}^{\mathcal{T}}} \right) \approx \mathcal{O} \left( \mathcal{N}^{\mathcal{T}} N^{\mathcal{T}} \right) \).
\end{theorem}
A fundamental consequence of this approximation capability—and of our specific hybrid coupling scheme—is its stabilizing effect on the optimization landscape. As demonstrated in the following corollary, the architecture actively mitigates the barren plateau phenomenon, maintaining detectable gradients and ensuring trainability even as the quantum components scale \cite{McClean201800, Cerezo202101, Cerezo202100}.
\begin{corollary}\label{corollary:Col-3.1}
    Under the hypotheses of Theorem \ref{theorem:Teo-3.1}, the QPI-DeepONet-MAC architecture mitigates the barren plateaus problem. That is, in the supremum norm, the gradient of the loss function
    \begin{equation}\label{eq:44}
        \begin{split}
            &\mathfrak{L} \left( \boldsymbol{\tilde{\mathfrak{v}}}_{\flat}, \mathbf{y}, \hat{O}^{\mathcal{B}}, \hat{O}^{\mathcal{T}}, \boldsymbol{\mathcal{W}}_{\flat}^{\mathcal{B} \ast}, \boldsymbol{\mathcal{W}}^{\mathcal{T} \ast}, \boldsymbol{\Xi}^{\mathcal{B} \ast}, \boldsymbol{\Xi}^{\mathcal{T} \ast} \right) \coloneq \\
            &\hspace{4.5cm} \displaystyle{\sum_{\substack{j \in \mathcal{I}_{M}}} \mathfrak{L}_{j} \left( \boldsymbol{\tilde{\mathfrak{v}}}_{\flat}, \mathbf{y}, \hat{O}^{\mathcal{B}}, \hat{O}^{\mathcal{T}}, \boldsymbol{\mathcal{W}}_{\flat}^{\mathcal{B} \ast}, \boldsymbol{\mathcal{W}}^{\mathcal{T} \ast}, \boldsymbol{\Xi}^{\mathcal{B} \ast}, \boldsymbol{\Xi}^{\mathcal{T} \ast} \right)},
        \end{split}
    \end{equation}
    where
    \begin{equation}\label{eq:45}
        \begin{split}
            &\mathfrak{L}_{j} \left( \boldsymbol{\tilde{\mathfrak{v}}}_{\flat}, \mathbf{y}, \hat{O}^{\mathcal{B}}, \hat{O}^{\mathcal{T}}, \boldsymbol{\mathcal{W}}_{\flat}^{\mathcal{B} \ast}, \boldsymbol{\mathcal{W}}^{\mathcal{T} \ast}, \boldsymbol{\Xi}^{\mathcal{B} \ast}, \boldsymbol{\Xi}^{\mathcal{T} \ast} \right) = \\
            &\hspace{3.0cm} \left\Vert \mathfrak{G}_{j} \left( \boldsymbol{\tilde{\mathfrak{v}}}_{\flat}, \hat{O}^{\mathcal{B}}, \boldsymbol{\mathcal{W}}_{\flat}^{\mathcal{B} \ast}, \boldsymbol{\Xi}^{\mathcal{B} \ast} \right) \left( \mathbf{y}, \hat{O}^{\mathcal{T}}, \boldsymbol{\mathcal{W}}^{\mathcal{T} \ast}, \boldsymbol{\Xi}^{\mathcal{T} \ast} \right) - u_{j} \left( \mathbf{y} \right) \right\Vert_{2}^{2},
        \end{split}
    \end{equation}
    with respect to the quantum parameters \( \boldsymbol{\Xi}^{\mathcal{B} \ast}, \boldsymbol{\Xi}^{\mathcal{T} \ast} \) is bounded from above by:
    \begin{equation}\label{eq:46}
        \begin{split}
            &\left\Vert \nabla_{\left( \boldsymbol{\Xi}^{\mathcal{B} \ast}, \boldsymbol{\Xi}^{\mathcal{T} \ast} \right)} \mathfrak{L} \left( \boldsymbol{\tilde{\mathfrak{v}}}_{\flat}, \mathbf{y}, \hat{O}^{\mathcal{B}}, \hat{O}^{\mathcal{T}}, \boldsymbol{\mathcal{W}}_{\flat}^{\mathcal{B} \ast}, \boldsymbol{\mathcal{W}}^{\mathcal{T} \ast}, \boldsymbol{\Xi}^{\mathcal{B} \ast}, \boldsymbol{\Xi}^{\mathcal{T} \ast} \right) \right\Vert_{\infty, M} \leq \\
            &\hspace{5.0cm} \mathcal{O} \left( \frac{1}{ \mathcal{N}^{\mathcal{T}} N^{\mathcal{T}} \sqrt{ \mathcal{N}^{\mathcal{B}} N^{\mathcal{B}} } } + \frac{1}{ \mathcal{N}^{\mathcal{B}} N^{\mathcal{B}} \sqrt{ \mathcal{N}^{\mathcal{T}} N^{\mathcal{T}} } } \right),
        \end{split}
    \end{equation}
    where \( \mathcal{N}^{\mathcal{B}} \) is the depth of the quantum component of the branch-net, \(\mathcal{N}^{\mathcal{T}}\) is the depth of the quantum component of the trunk-net, \( N^{\mathcal{B}} \) and \( N^{\mathcal{T}} \) are the number of qubits corresponding to the quantum components of the branch-net and trunk-net, respectively. Consequently, to keep the gradients detectable above a finite measurement precision \( \epsilon_{\mathrm{grad}} \), the problem size, defined by the expression \( \mathcal{O} \left( \left[ \mathcal{N}^{\mathcal{B}} \mathcal{N}^{\mathcal{T}} N^{\mathcal{B}} N^{\mathcal{T}} \sqrt{ \mathcal{N}^{\mathcal{B}} \mathcal{N}^{\mathcal{T}} N^{\mathcal{B}} N^{\mathcal{T}} } \right]/\left[ \mathcal{N}^{\mathcal{B}} N^{\mathcal{B}} \sqrt{ \mathcal{N}^{\mathcal{T}} N^{\mathcal{T}} } + \mathcal{N}^{\mathcal{T}} N^{\mathcal{T}} \sqrt{ \mathcal{N}^{\mathcal{B}} N^{\mathcal{B}} } \right] \right) \), must satisfy the following scaling condition for trainability:
    \begin{equation}\label{eq:47}
        \frac{\mathcal{N}^{\mathcal{B}} \mathcal{N}^{\mathcal{T}} N^{\mathcal{B}} N^{\mathcal{T}} \sqrt{ \mathcal{N}^{\mathcal{B}} \mathcal{N}^{\mathcal{T}} N^{\mathcal{B}} N^{\mathcal{T}} }}{\mathcal{N}^{\mathcal{B}} N^{\mathcal{B}} \sqrt{ \mathcal{N}^{\mathcal{T}} N^{\mathcal{T}} } + \mathcal{N}^{\mathcal{T}} N^{\mathcal{T}} \sqrt{ \mathcal{N}^{\mathcal{B}} N^{\mathcal{B}} }} \lesssim \mathcal{O}\left(\frac{1}{\epsilon_{\mathrm{grad}}}\right).
    \end{equation}
\end{corollary}
However, there is another interesting theoretical result regarding the behavior of the gradient, which is stated in the following Corollary.
\begin{corollary}\label{corollary:Col-3.2} 
    Under the assumptions of Theorem \ref{theorem:Teo-3.1}, given the loss function \(\mathfrak{L}\) defined in (\ref{eq:44}), \( \boldsymbol{\mathfrak{v}} \in \mathbb{V}^{M} \), \( \boldsymbol{\tilde{\mathfrak{v}}}_{\flat} \in \mathbb{R}^{m M \times 1} \), \( \mathbf{y} \in \mathcal{K} \), there are sets of weights and biases \( \boldsymbol{\mathcal{W}}_{\flat}^{\mathcal{B} \ast} \), \( \boldsymbol{\mathcal{W}}^{\mathcal{T} \ast} \), \( \boldsymbol{\Xi}^{\mathcal{B} \ast} \), \( \boldsymbol{\Xi}^{\mathcal{T} \ast} \) and local operators \( \hat{O}^{\mathcal{B}} \in \mathbb{O}^{\mathcal{B}} \), \( \hat{O}^{\mathcal{T}} \in \mathbb{O}^{\mathcal{T}} \), such that:
    \begin{equation}\label{eq:48}
        \begin{split}
            &\mathbb{E}_{\left( \boldsymbol{\Xi}^{\mathcal{B} \ast}, \boldsymbol{\Xi}^{\mathcal{T} \ast} \right)} \left[ \left\Vert \nabla_{\left( \boldsymbol{\Xi}^{\mathcal{B} \ast}, \boldsymbol{\Xi}^{\mathcal{T} \ast} \right)} \mathfrak{L} \left( \boldsymbol{\tilde{\mathfrak{v}}}_{\flat}, \mathbf{y}, \hat{O}^{\mathcal{B}}, \hat{O}^{\mathcal{T}}, \boldsymbol{\mathcal{W}}_{\flat}^{\mathcal{B} \ast}, \boldsymbol{\mathcal{W}}^{\mathcal{T} \ast}, \boldsymbol{\Xi}^{\mathcal{B} \ast}, \boldsymbol{\Xi}^{\mathcal{T} \ast} \right) \right\Vert_{\infty, M}^{2} \right] \geq \\
            &\hspace{7.75cm} \mathcal{O} \left( \frac{1}{ \mbox{poly} \left(N^{\mathcal{B}}\right)} + \frac{1}{ \mbox{poly} \left(N^{\mathcal{T}}\right)} \right).
        \end{split}
    \end{equation}
\end{corollary}
In conclusion, having mathematically guaranteed both the expressivity and the optimization stability of the model, it is necessary to determine its informational lower bound. The subsequent theorem defines the minimum number of sensors required to accurately observe and reconstruct the dynamic evolution of the underlying well-posed physical system.
\begin{theorem}\label{theorem:Teo-3.2}
    Let \( \mathcal{K} \subset D \) be a compact subset and \( \mathbb{V} \subset C_{\mathbb{R}} \left( \mathbb{K} \right) \) also compact. Suppose that \(\mathbb{V}^{M}\) is a compact subset of \(C_{\mathbb{R}^{M}}^{k} \left(\mathbb{K}\right)\) and that \(\mathcal{F}\) takes the form
    \begin{equation}\label{eq:49}
        \mathcal{F} \left(\mathbf{y},\mathbf{u},\hat{\mathcal{L}}_{1}\left[\cdot\right],\ldots,\hat{\mathcal{L}}_{M}\left[\cdot\right]\right) = \frac{\partial\mathbf{u}}{\partial t} - A \mathbf{u}-F \left(\mathbf{y},\mathbf{u}\right),
    \end{equation}
    where \(A : D \left(A\right) \subset C_{\mathbb{R}^{M}}^{\infty} \left(D\right) \to C_{\mathbb{R}^{M}}^{\infty} \left(D\right)\) is a linear operator (representing any eventual differential operators with respect to spatial coordinates), densely defined and the infinitesimal generator of a \(C_{0}\)-semigroup, and \(F : D \times C_{\mathbb{R}^{M}}^{\infty} \left(D\right) \to C_{\mathbb{R}^{M}}^{\infty} \left(D\right)\) is a Lipschitz mapping with respect to the second variable. We also suppose that (\ref{eq:04})--(\ref{eq:07}) in this context defines an initial-boundary value problem well-posed in the Hadamard sense. Consider a QPI-DeepONet-MAC architecture possessing universal approximation capability over \(\mathbb{V}^{M}\). Then, given a tolerance error \(\tilde{\epsilon} > \epsilon > 0\), the minimum number of sensors \(m\) required to guarantee that \(\left\Vert \boldsymbol{\mathfrak{G}} \left( \boldsymbol{\tilde{\mathfrak{v}}}_{\flat}, \hat{O}^{\mathcal{B}}, \boldsymbol{\mathcal{W}}_{\flat}^{\mathcal{B} \ast}, \boldsymbol{\Xi}^{\mathcal{B} \ast} \right) \left( \cdot, \hat{O}^{\mathcal{T}}, \boldsymbol{\mathcal{W}}^{\mathcal{T} \ast}, \boldsymbol{\Xi}^{\mathcal{T} \ast} \right) - \boldsymbol{\mathcal{G}} \left( \boldsymbol{\mathfrak{v}} \right) \left( \cdot \right) \right\Vert_{\infty,M} \leq \tilde{\epsilon}\) is given by:
    \begin{equation}\label{eq:50}
        m \geq \left[ \frac{\mathfrak{C}_{L} C_{\text{Int}}}{\tilde{\epsilon} - \epsilon} \right]^{\frac{d}{k}} \exp{\frac{\omega t d}{k}}.
    \end{equation}
\end{theorem}
The detailed step-by-step mathematical proofs for Theorem \ref{theorem:Teo-3.1}, Corollary \ref{corollary:Col-3.1}, and Theorem \ref{theorem:Teo-3.2} are provided in the Appendix.

To conclude this section, we will carry out a series of deductions and analyses aimed at determining fundamental physical quantities and understanding how these quantities can provide the quantum information derived from our model. Specifically, we aim to establish a solid foundation for understanding the impact of these quantum resources and their role as learning regulators within our hybrid scheme. For this reason, we begin by obtaining a density function that contains only the quantum resource of our architecture, and to achieve this, we take the partial trace of (\ref{eq:26}) with respect to the classical components \( \rho_{\mathcal{C}}^{\nu \in \{\mathcal{B}, \mathcal{T}\}} \), obtaining
\begin{equation}\label{eq:51}
    \begin{split}
        \rho_{\mathcal{Q}} &= \operatorname{Tr}_{\mathcal{C}^{\mathcal{B}}, \mathcal{C}^{\mathcal{T}}} \left\{ \bigotimes_{\substack{\mu \in \{ \mathcal{C}, \mathcal{Q} \} \\ \nu \in \{ \mathcal{B}, \mathcal{T} \}}} \rho_{\mu}^{\nu} \right\} \\
        &= \operatorname{Tr}_{\mathcal{C}^{\mathcal{B}}} \left\{ \rho_{\mathcal{C}}^{\mathcal{B}} \right\} \otimes \rho_{\mathcal{Q}}^{\mathcal{B}} \otimes \operatorname{Tr}_{\mathcal{C}^{\mathcal{T}}} \left\{ \rho_{\mathcal{C}}^{\mathcal{T}} \right\} \otimes \rho_{\mathcal{Q}}^{\mathcal{T}}..
    \end{split}
\end{equation}
And since \(\operatorname{Tr}_{\mathcal{C}^{\nu}} \left\{ \rho_{\mathcal{C}}^{\nu} \right\} = 1\) for \(\nu \in \{\mathcal{B},\mathcal{T}\}\), expression (\ref{eq:51}) reduces to
\begin{equation}\label{eq:53}
    \rho_{\mathcal{Q}} = \bigotimes_{\substack{\nu \in \{ \mathcal{B}, \mathcal{T} \}}} \rho_{\mathcal{Q}}^{\nu}.
\end{equation}
This reduction is essential, since unlike the global evolution, which can be described by pure states, the dynamics of the quantum part in our hybrid architecture, obtained through the partial trace, can manifest mixed states due to the entanglement induced with the classical subsystems. This renders the informational measures proposed herein non-trivial and necessary for the control of expressivity.

The foregoing allows us to determine the coherence, which is given by:
\begin{equation}\label{eq:112}
    \mathcal{C}_{p} (\rho_{\mathcal{Q}}) = \frac{1}{N_{\mathcal{B}}N_{\mathcal{T}}} \sum_{(m^{\mathcal{B}}, m^{\mathcal{T}}) \neq (n^{\mathcal{B}}, n^{\mathcal{T}})} \left\Vert \mathcal{P}_{m^{\mathcal{B}}n^{\mathcal{B}}} \left( \hat{O}^{\mathcal{B}}, \boldsymbol{\Theta}_{p}^{\mathcal{B}} \right) \otimes \mathcal{P}_{m^{\mathcal{T}}n^{\mathcal{T}}} \left( \hat{O}^{\mathcal{T}}, \boldsymbol{\Theta}_{p}^{\mathcal{T}} \right) \right\Vert_{1},
\end{equation}
where \( 0 \leq \mathcal{C}_{p} \left( \rho_{\mathcal{Q}} \right) \leq 1 \), \( \Vert \cdot \Vert_{1} \) is the \( l_{1} \) norm, and each term of the summations in (\ref{eq:112}) is nothing more than the norm of each of the off-diagonal entries of (\ref{eq:53}). Now, we can define the measure of the decoherence rate between two subsequent states, for example, \( \rho_{\mathcal{Q}}^{\left( p-1 \right)}\) and \( \rho_{\mathcal{Q}}^{\left( p \right)} \), by the expression
\begin{equation}\label{eq:54}
    \mathcal{D}_{p} \left( \rho_{\mathcal{Q}}^{\left( p-1 \right)}, \rho_{\mathcal{Q}}^{\left( p \right)} \right) \coloneq \mathcal{C}_{p-1} \left( \rho_{\mathcal{Q}}^{\left( p-1 \right)} \right) - \mathcal{C}_{p} \left( \rho_{\mathcal{Q}}^{\left( p \right)} \right),
\end{equation}
with \( 0 \leq \mathcal{D}_{p} \left( \rho_{\mathcal{Q}}^{\left( p-1 \right)}, \rho_{\mathcal{Q}}^{\left( p \right)} \right) \leq 1 \). This would allow monitoring whether the overall state is collapsing into a diagonal state, which implies a loss of expressivity.

Before proceeding with the definitions of interest in this section, we will calculate \( \varrho_{\mathcal{Q}} \). 
The density matrix \( \rho_{\mathcal{Q}} \) acts on \( \mathcal{H}_{\mathcal{Q}}^{\mathcal{B}} \otimes \mathcal{H}_{\mathcal{Q}}^{\mathcal{T}}\), \( \Tr{\rho_{\mathcal{Q}}} = 1 \), \( \rho_{\mathcal{Q}} = \rho_{\mathcal{Q}}^{\dagger} \to \expval{\rho_{\mathcal{Q}}} \geq 0 \), and $\varrho_{\mathcal{Q}}$ is normal, that is \( \left[\rho_{\mathcal{Q}}, \rho_{\mathcal{Q}}^{\dagger}\right] = 0 \). The spectral theorem states that any normal matrix can be diagonalized by a unitary transformation \( \hat{U} \), which in our case is defined by:
\begin{equation}\label{eq:55}
    \hat{U} \coloneq \displaystyle{\bigotimes_{\substack{\nu \in \{ \mathcal{B}, \mathcal{T} \}}} \hat{U}_{\mathcal{Q}}^{\nu}},
\end{equation}
which must satisfy \( \hat{U} = \hat{U}^{\dagger} \), whose transformation matrices \( \hat{U}_{\mu}^{\nu} \) act on \( \mathcal{H}_{\mathcal{Q}}^{\mathcal{B}} \otimes \mathcal{H}_{\mathcal{Q}}^{\mathcal{T}}\) and belong to the set
\begin{equation}\label{eq:56}
    \left\{ \hat{U}_{\mathcal{Q}}^{\nu} \right\} _{\substack{\nu \in \{ \mathcal{B}, \mathcal{T} \}}},
\end{equation}
which are defined in terms of the normalized eigenvectors
\begin{equation}\label{eq:57}
    \left\{ \ket{\lambda_{\mathcal{Q}}^{\nu}}_{j_{\nu}} \right\} _{\substack{j_{\nu} = \{1_{\nu}, \dots, N_{\nu}\}, \; \forall \; \nu \in \{ \mathcal{B}, \mathcal{T} \}}},
\end{equation}
obtained from the equation
\begin{equation}\label{eq:58}
     \left( \rho_{\mathcal{Q}} - \lambda_{\mathcal{Q}_{j_{\nu}}}^{\nu} \mathds{1} \right) \ket{\lambda_{\mathcal{Q}}^{\nu}}_{j_{\nu}} = 0,
\end{equation}
for each eigenvalue obtained from the equation
\begin{equation}\label{eq:59}
     \det\left( \rho_{\mathcal{Q}} - \lambda_{\mathcal{Q}}^{\nu} \mathds{1} \right) = 0.
\end{equation}
Thus, \( \varrho_{\mathcal{Q}} \) can be given by:
\begin{equation}\label{eq:60}
    \begin{split}
        \varrho_{\mathcal{Q}} &= \hat{U}^{\dagger} \rho_{\mathcal{Q}} \hat{U} \\
        &= \displaystyle{\bigotimes_{\substack{\nu \in \{ \mathcal{B}, \mathcal{T} \}}} \hat{U}_{\mathcal{Q}}^{\nu \,\dagger}} \left[ \displaystyle{\bigotimes_{\substack{\nu \in \{ \mathcal{B}, \mathcal{T} \}}} \rho_{\mathcal{Q}}^{\nu}} \right] \displaystyle{\bigotimes_{\substack{\nu \in \{ \mathcal{B}, \mathcal{T} \}}} \hat{U}_{\mathcal{Q}}^{\nu}} \\
        &= \displaystyle{\bigotimes_{\substack{\nu \in \{ \mathcal{B}, \mathcal{T} \}}} \hat{U}_{\mathcal{Q}}^{\nu \, \dagger} \rho_{\mathcal{Q}}^{\nu} \hat{U}_{\mathcal{Q}}^{\nu}} \\
        &= \displaystyle{\bigotimes_{\substack{\nu \in \{ \mathcal{B}, \mathcal{T} \}}} \varrho_{\mathcal{Q}}^{\nu}},
    \end{split}
\end{equation}
and which admits the spectral decomposition
\begin{equation}\label{eq:61}
    \varrho_{\mathcal{Q}} = \displaystyle{\sum_{\substack{k_{\mathcal{B}} \in \mathcal{I}_{N_{\mathcal{B}}} \\ k_{\mathcal{T}} \in \mathcal{I}_{N_{\mathcal{T}}}}} \Lambda_{k_{\mathcal{B}}  k_{\mathcal{T}}} \ketbra{\Lambda_{k_{\mathcal{B}} k_{\mathcal{T}}}}{\Lambda_{k_{\mathcal{B}} k_{\mathcal{T}}}}},
\end{equation}
where \(\Lambda_{k_{\mathcal{B}} k_{\mathcal{T}}} = \lambda_{\mathcal{Q}_{k_{\mathcal{B}}}}^{\mathcal{B}} \lambda_{\mathcal{Q}_{k_{\mathcal{T}}}}^{\mathcal{T}}\) represent the eigenvalues of the composite system and the kets represent the orthonormal basis of eigenvectors in \( \mathcal{H} \), defined by the tensor product \(\ket{\Lambda_{k_{\mathcal{B}} k_{\mathcal{T}}}} = \ket{\lambda_{\mathcal{Q}_{k_{\mathcal{B}}}}^{\mathcal{B}}} \otimes \ket{\lambda_{\mathcal{Q}_{k_{\mathcal{T}}}}^{\mathcal{T}}}\). Therefore, using (\ref{eq:61}), the Von Neumann entropy assumes the form
\begin{equation}\label{eq:62}
    \begin{split}
        S_{p} \left( \varrho_{\mathcal{Q}} \right) &= - \displaystyle{\sum_{\substack{\nu \in \{ \mathcal{B}, \mathcal{T} \} \\ k_{\nu} \in \mathcal{I}_{N_{\nu}}}} \lambda_{\mathcal{Q}_{k_{\nu}}}^{\nu} \ln{\left( \lambda_{\mathcal{Q}_{k_{\nu}}}^{\nu} \right)}} \\
        &= S_{p} \left( \varrho_{\mathcal{Q}}^{\mathcal{B}} \right) + S_{p} \left( \varrho_{\mathcal{Q}}^{\mathcal{T}} \right),
    \end{split}
\end{equation}
which is nothing more than the sum of the entropies of each of the subsystems, and emphasizes that there is no correlation between the subsystems. The aforementioned is expected by the construction of the density matrix \( \varrho_{\mathcal{Q}} \), since said entropy measures the uncertainty and the lack of information about the quantum state of a system, quantifying the degree of mixing of a density matrix.

Next, we will calculate the last of the quantities that we consider relevant for our hybrid model, which is the fidelity, which, for two subsequent states \( \varrho_{\mathcal{Q}}^{\left( p-1 \right)}\), \( \varrho_{\mathcal{Q}}^{\left( p \right)} \) acting on \( \mathcal{H}_{\mathcal{Q}}^{\mathcal{B}} \otimes \mathcal{H}_{\mathcal{Q}}^{\mathcal{T}}\), given its multiplicative nature, and using the property \( \sqrt{A^{\left( 1 \right)} \otimes A^{\left( 2 \right)}} = \sqrt{A^{\left( 1 \right)}} \otimes \sqrt{A^{\left( 2 \right)}} \), can be written as:
\begin{equation}\label{eq:63}
    \begin{split}
        F_{p} \left( \varrho_{\mathcal{Q}}^{\left( p-1 \right)}, \varrho_{\mathcal{Q}}^{\left( p \right)} \right) &= \Tr{\sqrt{\sqrt{\displaystyle{\bigotimes_{\substack{\nu \in \{ \mathcal{B}, \mathcal{T} \}}} \varrho_{\mathcal{Q}}^{\nu \, \left( p-1 \right)}}}  \left[ \displaystyle{\bigotimes_{\substack{\nu \in \{ \mathcal{B}, \mathcal{T} \}}} \varrho_{\mathcal{Q}}^{\nu \, \left( p \right)}} \right] \sqrt{\displaystyle{\bigotimes_{\substack{\nu \in \{ \mathcal{B}, \mathcal{T} \}}} \varrho_{\mathcal{Q}}^{\nu \left( p-1 \right)}}}}} \\
        &= \Tr{\displaystyle{\bigotimes_{\substack{\nu \in \{ \mathcal{B}, \mathcal{T} \}}} \sqrt{\sqrt{\varrho_{\mathcal{Q}}^{\nu \, \left( p-1 \right)}} \varrho_{\mathcal{Q}}^{\nu \, \left( p \right)} \sqrt{\varrho_{\mathcal{Q}}^{\nu \, \left( p-1 \right)}}}}}.
    \end{split}
\end{equation}
And upon using the property \( \Tr{A^{\left( 1 \right)} \otimes A^{\left( 2 \right)}} = \Tr{A^{\left( 1 \right)}} \otimes \Tr{A^{\left( 2 \right)}} \), we obtain:
\begin{equation}\label{eq:64}
    \begin{split}
        F_{p} \left( \varrho_{\mathcal{Q}}^{\left( p-1 \right)}, \varrho_{\mathcal{Q}}^{\left( p \right)} \right) &= \displaystyle{\bigotimes_{\substack{\nu \in \{ \mathcal{B}, \mathcal{T} \}}} \Tr{\sqrt{\sqrt{\varrho_{\mathcal{Q}}^{\nu \, \left( p-1 \right)}} \varrho_{\mathcal{Q}}^{\nu \, \left( p \right)} \sqrt{\varrho_{\mathcal{Q}}^{\nu \, \left( p-1 \right)}}}}} \\
        &= \displaystyle{\bigotimes_{\substack{\nu \in \{ \mathcal{B}, \mathcal{T} \}}} F \left( \varrho_{\mathcal{Q}}^{\nu \, \left( p-1 \right)}, \varrho_{\mathcal{Q}}^{\nu \, \left( p \right)} \right)}.
    \end{split}
\end{equation}
Due to the construction of \( \varrho_{\mathcal{Q}} \), each individual fidelity between the states \( \varrho_{\mathcal{Q}}^{\nu \, \left( p-1 \right)}\), \( \varrho_{\mathcal{Q}}^{\nu \, \left( p \right)} \) reduces to the classical Bhattacharyya fidelity between the eigenvalue distributions:
\begin{equation}\label{eq:65}
    F \left( \varrho_{\mathcal{Q}}^{\nu \, \left( p-1 \right)}, \varrho_{\mathcal{Q}}^{\nu \, \left( p \right)} \right) = \displaystyle{\sum_{\substack{k_{\nu} \in \mathcal{I}_{N_{\nu}}}} \sqrt{\lambda_{\mathcal{Q}_{k_{\nu}}}^{\nu \, \left( p-1 \right)} \lambda_{\mathcal{Q}_{k_{\nu}}}^{\nu \, \left( p \right)}}}.
\end{equation}
Finally, we must provide an application that justifies the calculations of expressions (\ref{eq:54}), (\ref{eq:62}), and (\ref{eq:65}), which in our case would be obtaining information regarding the changes in the quantum state as the training evolves, genuinely supporting our hybrid model. Because of this, it is fundamental to understand that if the quantum system collapses into a system where the entropy is low, we face a loss of information that leads us to less expressive states. For instance, in the case where the quantum components collapse into a state without superposition.

The foregoing can occur for states in which the decoherence rate is maximum, that is, \( \mathcal{D}_{p} \left( \rho_{T}^{\left( p-1 \right)}, \rho_{T}^{\left( p \right)} \right) \approx 1 \), which implies that \(S_{p} \left( \varrho_{\mathcal{Q}} \right) \approx S_{p} \left( \varrho_{\mathcal{Q}}^{\left( p-1 \right)} \right) \). Furthermore, in this regime of low entropy and high initial decoherence, there might exist a case where \( \lambda_{\mathcal{Q}_{k_{\nu}}}^{\nu \left( p-1; p \right)} \to \delta_{k_{\nu} k_{\nu}^{\prime}} \), where it is the Kronecker delta. Therefore, the Bhattacharyya fidelity simplifies to the overlap of the dominant populations:
\begin{equation}\label{eq:66}
        F \left( \varrho_{\mathcal{Q}}^{\nu \, \left( p-1 \right)}, \varrho_{\mathcal{Q}}^{\nu \, \left( p \right)} \right) = \displaystyle{\sum_{\substack{k_{\nu} \in \mathcal{I}_{N_{\nu}}}} \sqrt{\delta_{k_{\nu} k_{\nu}^{\prime}} \delta_{k_{\nu} k_{\nu}^{\prime \prime}}}} 
        = \sqrt{\delta_{k_{\nu}^{\prime \prime} k_{\nu}^{\prime \prime}}}.
\end{equation}
Consequently, from the fidelity, we observe that \( F_{p} \left( \varrho_{\mathcal{Q}}^{\left( p-1 \right)}, \varrho_{\mathcal{Q}}^{\left( p \right)} \right) \to 1 \) denotes that the network has reached a convergence regime, wherein the trajectory in the quantum state space has become stationary. The informational risk does not lie in the stability of this trajectory, but rather in the nature of the final state to which the system has converged. If this convergence occurs towards a state of residual entropy \( S_{p} \left( \varrho_{\mathcal{Q}} \right) \to 0 \), the model has reached a stationary solution that sacrificed quantum superposition in favor of a classical diagonal configuration.

Thus, we can see that in this convergence regime towards low-entropy states, the Bhattacharyya fidelity in (\ref{eq:65}) acts upon a classical probability distribution, losing its quantum superposition character. Consequently, the system is not only stable, but ``frozen'' in a trivial classical configuration. Hence, the expressivity of the hybrid model is reduced, since the network ceases to utilize entanglement and interference as feature space expanders.

For this reason, it becomes necessary to implement a regularization strategy that must guarantee that the stationary point of convergence (\(F_{p} \left( \varrho_{\mathcal{Q}}^{\left( p-1 \right)}, \varrho_{\mathcal{Q}}^{\left( p \right)} \right) \to 1\)) preserves an entropy value \(S_{p} \left( \varrho_{\mathcal{Q}}\right) > S_{\mathrm{threshold}}\) and a coherence \( \mathcal{C}_{p} \left(\varrho_{\mathcal{Q}}^{\left( p \right)} \right) > 0 \), ensuring that the final solution maintains the requisite quantum richness for the approximation of complex operators. Thus, to mitigate this risk and guarantee that the model's convergence occurs along a trajectory rich in quantum information, we propose the integration of regularization terms defined by:
\begin{equation}\label{eq:67}
    \mathcal{L}_{INF} = \lambda_{COE} \displaystyle{\sum_{\substack{p \in \mathcal{I}_{P}}} \left(1 - \mathcal{C}_{p}\left( \varrho_{\mathcal{Q}}^{\left( p \right)} \right)\right)} + \lambda_{FID} \displaystyle{\sum_{\substack{p \in \mathcal{I}_{P}}} \left(1 - F_{p} \left( \varrho_{\mathcal{Q}}^{\left( p-1 \right)}, \varrho_{\mathcal{Q}}^{\left( p \right)} \right)\right)},
\end{equation}
where \(\lambda_{COE}\) and \(\lambda_{FID}\) are hyperparameters that balance the precision of the PDE approximation with the integrity of the quantum state. The term \(\left(1 - \mathcal{C}_{p}\left( \varrho_{\mathcal{Q}}^{\left( p \right)} \right)\right)\) acts as a barrier that penalizes the loss of coherence, forcing the system to maintain active superpositions and preventing the state from collapsing into a classical basis. Complementarily, the term \(\left(1 - F_{p} \left( \varrho_{\mathcal{Q}}^{\left( p-1 \right)}, \varrho_{\mathcal{Q}}^{\left( p \right)} \right)\right)\) regulates the stability of the learning trajectory; by penalizing abrupt variations in fidelity between subsequent steps \(\left(F_{p} \left( \varrho_{\mathcal{Q}}^{\left( p-1 \right)}, \varrho_{\mathcal{Q}}^{\left( p \right)} \right)\approx 1\right)\), we guarantee that the evolution of the quantum state occurs continuously and stably within the Hilbert space. In this way, by summing expressions (\ref{eq:39}) and (\ref{eq:62}), we obtain the total loss function of our hybrid model, given by the expression:
\begin{equation}\label{eq:68}
    \mathcal{L}_{TOT} = \lambda_{IC} \tilde{\mathcal{L}}_{IC} + \lambda_{BC} \tilde{\mathcal{L}}_{BC} + \lambda_{PDE} \tilde{\mathcal{L}}_{PDE} + \lambda_{SOL} \hat{\mathcal{L}}_{SOL} + \mathcal{L}_{INF},
\end{equation}
which enables an optimization process guided not only by the minimization of the residuals of the initial conditions, boundary conditions, and the system of partial differential equations, but also by the preservation of the quantum richness necessary for the hybrid architecture to sustain its universal approximation capability.

\section{Conclusions and future perspectives}\label{sec:5}

In this work, we introduced QPI-DeepONet-MAC, a hybrid
quantum-classical architecture that integrates parameterized quantum
circuits into Physics-Informed Deep Operator Networks through
multiplicative and additive couplings. We established a theoretical
framework for analyzing the expressive power and trainability of the
proposed architecture. In particular, under the stated assumptions, we
demonstrated that QPI-DeepONet-MAC retains universal approximation
capabilities for the class of operators considered. We also derived
upper and lower bounds on the gradients with respect to the quantum
parameters, yielding explicit scaling conditions for maintaining
trainable gradients and providing theoretical conditions for mitigating
barren-plateau effects.

Furthermore, we derived a bound on the minimum number of sensors
required to achieve a prescribed approximation accuracy for a class of
well-posed parametric PDEs. This result relates the required sampling
density to the regularity and dimensionality of the underlying problem
and to the dynamical properties of the corresponding evolution
equation. In addition, we introduced a quantum-informational
regularization framework based on the dynamics of quantum coherence
and state fidelity. This framework provides a mechanism for controlling
the evolution of the quantum component during training and for
preventing excessive loss of quantum-state structure.

Taken together, these analytical results provide a theoretical
foundation for incorporating parameterized quantum circuits into
physics-informed operator learning. Importantly, they characterize
conditions under which the proposed architecture can remain expressive
and trainable as its quantum components increase in size, rather than
establishing performance advantages over classical operator-learning
methods.

Several directions remain for future work. First, numerical
experiments are needed to assess the practical consequences of the
theoretical results, including approximation accuracy, optimization
stability, sensor requirements, and the behavior of the proposed
quantum-informational regularization. We are currently developing
classical simulations and numerical benchmarks for this purpose.
Second, the implementation of QPI-DeepONet-MAC on quantum hardware
will require optimized circuit constructions, parameterizations, and
training procedures adapted to the noise and resource constraints of
current Noisy Intermediate-Scale Quantum (NISQ) devices. Such
experiments will provide an opportunity to test whether the theoretical
trainability conditions derived here translate into practical
advantages on quantum hardware.

\begin{appendices}
\section{Ismailov's Theorem and their consequences}\label{sec:Ap-A}
This appendix is dedicated to demonstrating a generalization of the Universal Approximation Theorem for operators introduced in \cite{Lu202100}. To this end, our intention is to provide a detailed proof of the theorem in question to approximate continuous operators \(\boldsymbol{\mathcal{G}}: \mathscr{V} \to  C_{\mathbb{R}^{M}}\left( \mathcal{K} \right)\), where \(\mathscr{V}\) is a compact subset of continuous functions defined on some compact set, not necessarily Euclidean, and \(\mathcal{K}\) is a compact subset of some \(\mathbb{R}^{n}\). This proof arises in an elegant and natural way from a much more general version proved by Vugar E. Ismailov \cite{Ismailov202600}.

However, even though our proof may not be entirely innovative, we find it necessary to include it in order to highlight its utility in the proof of Theorem \ref{theorem:Teo-3.1}. Furthermore, our proof helps to identify the objects involved in the classical part of the QPI-DeepONet-MAC Hybrid Model.

First of all, we would like to emphasize that in what follows in this appendix, \(\mathscr{X}\) will always denote a real vector space, and when \(\left\Vert  \cdot \right\Vert_{\mathscr{X}}\) is a norm defined on it, we will say that \(\left(\mathscr{X},\left\Vert  \cdot \right\Vert_{\mathscr{X}}\right)\) is a normed space. In this context, we define
\begin{equation}\label{eq:A.1}
    \mathscr{B}_{\mathbb{R}} \left( \mathscr{X} \right) \coloneq \left\lbrace \varphi:\mathscr{X}\to\mathbb{R} \, | \, \varphi \mbox{ is linear and continuous} \right\rbrace,
\end{equation}
as the topological dual space \cite{Brezis201100} of \(\mathscr{X}\). Now, we would like to make the following observation.
\begin{remark}\label{remark:Obs-A.1}
    Let \(\left(\mathbb{B},\left\Vert \cdot \right\Vert_{\mathbb{B}}\right)\) be a Banach space and \(\mathbb{K} \subset \mathbb{B}\) a compact subset. It is well known that
    \begin{equation}\label{eq:A.2}
        \mathcal{M}_{\mathbb{R}} \left(\mathbb{K}\right) \coloneq \left\lbrace \mu : \mathcal{B}_{\mathbb{K}} \to \mathbb{R} \, | \, \mu \mbox{ is a finite signed regular Borel measure} \right\rbrace,
    \end{equation}
    where \(\mathcal{B}_{\mathbb{K}}\) denotes the Borel \(\sigma\)-algebra on \(\mathbb{K}\). Thus, the Riesz-Markov-Kakutani Representation Theorem \cite{Rudin198700} guarantees that
    \begin{equation}\label{eq:A.3}
        \mathscr{B}_{\mathbb{R}} \left( C_{\mathbb{R}}\left(\mathbb{K} \right) \right) = \mathcal{M}_{\mathbb{R}} \left(\mathbb{K}\right)
    \end{equation}
    and for each \(\mu \in \mathscr{B}_{\mathbb{R}} \left( C_{\mathbb{R}}\left(\mathbb{K} \right) \right)\),
    \begin{equation}\label{eq:A.4}
        \mu \left( f \right) \coloneq \int_{\mathbb{K}} f \, \operatorname{d} \mu, \, \forall \, f \in C_{\mathbb{R}} \left( \mathbb{K} \right).
    \end{equation}
    In addition to the above, using the same notation and as a consequence of it, we have that
    \begin{equation}\label{eq:A.5}
        \mathscr{B}_{\mathbb{R}} \left( C_{\mathbb{R}^{M}}\left(\mathbb{K} \right) \right) = \mathcal{M}_{\mathbb{R}^{M}} \left( \mathbb{K} \right) \coloneq \left\lbrace \boldsymbol{\mu} = \left(\mu_{1},\ldots,\mu_{M}\right)\, : \, \mu_{i} \in \mathcal{M}_{\mathbb{R}} \left(\mathbb{K}\right)\right\rbrace
    \end{equation}
    and for each \(\boldsymbol{\mu} = \left(\mu_{1}, \dots, \mu_{M}\right) \in \mathscr{B}_{\mathbb{R}} \left( C_{\mathbb{R}^{M}}\left(\mathbb{K} \right) \right)\),
    \begin{equation}\label{eq:A.6}
        \boldsymbol{\mu} \left(\textbf{f} \right) \coloneq \int_{\mathbb{K}} \textbf{f} \cdot \operatorname{d} \boldsymbol{\mu} \equiv \sum_{i \in \mathcal{I}_{M}} \int_{\mathbb{K}} f_{i} \, \operatorname{d} \mu_{i},
    \end{equation}
    for all \(\textbf{f} = \left(f_{1}, \dots, f_{M}\right) \in C_{\mathbb{R}^{M}} \left( \mathbb{K} \right)\).
\end{remark}
Now, we need to define topological neural networks in a normed space.
\begin{definition}\label{definition:Def-A.1}
    Let \(\left(\mathscr{X},\left\Vert \cdot \right\Vert_{\mathscr{X}}\right)\) be a normed space, let \( \sigma: \mathbb{R} \to \mathbb{R} \) be an activation function, and let \(M \ge 1\). We define a vector-valued topological feedforward neural network in \(\mathscr{X}\) with a single hidden layer as any mapping \( H : \mathscr{X} \to \mathbb{R}^{M} \) of the form
    \begin{equation}\label{eq:A.7}
        H \left(x\right) = A \sigma \left(T\left(x\right)\right),
    \end{equation}
    where
    \begin{equation}\label{eq:A.8}
        T \left(x\right) = \left(\varphi_{1} \left(x\right) + \vartheta_{1}, \dots, \varphi_{R} \left(x\right) + \vartheta_{R}\right),
    \end{equation}
    with \(\varphi_{r} \in \mathscr{B}_{\mathbb{R}}\left(\mathscr{X} \right)\), \(\vartheta_{r} \in \mathbb{R}\), and \(A \in \mathbb{R}^{M \times R}\). Furthermore, we would like to emphasize that here \(\sigma\) acts componentwise on \(\mathbb{R}^{R}\). Each hidden neuron evaluates a continuous linear function on the input element and then applies the activation function.
    
    Now, let \(\left(\mathscr{X},\left\Vert \cdot \right\Vert_{\mathscr{X}}\right)\) be as defined above, with the additional condition of being complete, that is, \(\left(\mathscr{X},\left\Vert \cdot \right\Vert_{\mathscr{X}}\right)\) is a Banach space. Given compact subsets \(\mathscr{V} \subset \mathscr{X}\) and \(\mathcal{K} \subset \mathbb{R}^{n}\), respectively, we consider continuous operators
    \begin{equation}\label{eq:A.9}
        G : \mathscr{V} \longrightarrow C_{\mathbb{R}^{M}} \left(\mathcal{K}\right),
    \end{equation} 
    where \(C_{\mathbb{R}^{M}}\left(\mathcal{K}\right)\) is equipped with the uniform norm \(\left\Vert \cdot \right\Vert_{\infty}\). In the context of DeepONet, one seeks to approximate the mapping
    \begin{equation}\label{eq:A.10}
        \left(u,\mathbf{y}\right) \longmapsto G \left(u\right) \left(\mathbf{y}\right), \, \left(u,\mathbf{y}\right) \in \mathscr{V} \times \mathcal{K},
    \end{equation}
    by a structured network that separates the dependence on the input \(u\) from the dependence on the coordinate \(\mathbf{y}\).
\end{definition}

In this sense, we need to define what topological DeepONets in a normed space are.

\begin{definition}\label{definition:Def-A.2}
    Let \(P \in \mathbb{N}^{\ast}\) and let \(M \in \mathbb{N}^{\ast}\). A topological DeepONet consists of a branch subnetwork defined by:
    \begin{equation}\label{eq:A.11}
        B : \mathscr{X} \longrightarrow \mathbb{R}^{M \times P},
    \end{equation}
    which receives an input \(x \in \mathscr{X}\) and yields a matrix as output
    \begin{equation}\label{A.12}
        B \left(x\right) = \left[\textbf{b}_{1} \left(x\right), \dots, \textbf{b}_{P} \left(x\right)\right],
    \end{equation}
    where each column \(\textbf{b}_{p}: \mathscr{X} \to \mathbb{R}^{M}\) is a topological neural network in \(\mathscr{X}\) constructed via (\ref{eq:A.7})-(\ref{eq:A.8}), and a trunk subnetwork defined by:
    \begin{equation}\label{eq:A.13}
        \boldsymbol{\tau} :\mathbb{R}^{n} \longrightarrow \mathbb{R}^{P}.
    \end{equation}
    This trunk subnetwork receives a point \(\textbf{y} \in \mathbb{R}^{n}\) as input and returns the vector
    \begin{equation}\label{eq:A.14}
        \boldsymbol{\tau} \left(\textbf{y}\right) = \left(\tau_{1} \left(\textbf{y}\right), \dots, \tau_{P} \left(\textbf{y}\right)\right),
    \end{equation}
    represented by a Euclidean neural network (for example, a network with a single hidden layer), whose component functions \(\tau_{p}: \mathbb{R}^{n} \to \mathbb{R}\), \(p \in \mathcal{I}_{P}\), denote the \(P\) outputs of the network.
\end{definition}

Thus, using Definitions \ref{definition:Def-A.1} and \ref{definition:Def-A.2}, the resulting DeepONet defines an operator
\begin{equation}\label{eq:A.15}
    \widehat{G} : \mathscr{X} \longrightarrow C_{\mathbb{R}^{M}} \left(\mathbb{R}^{n}\right),
\end{equation}
given by:
\begin{equation}\label{eq:A.16}
    \widehat{G} \left(x\right) \left(\textbf{y}\right) \coloneq B\left(x\right)\boldsymbol{\tau}\left(\textbf{y}\right) = \sum_{p \in \mathcal{I}_{P}} \tau_{p} \left(\textbf{y}\right) \boldsymbol{b}_{p} \left(x\right).
\end{equation}
Furthermore, according to Definitions (\ref{definition:Def-A.1})-(\ref{definition:Def-A.2}), each \(\boldsymbol{b}_{p}\) is of the form
\begin{equation}\label{eq:A.17}
    \begin{split}
        \boldsymbol{b}_{p} \left(x\right) &= A^{(p)} \sigma \left(\varphi_{1}^{p} \left(x\right) + \vartheta_{1}^{p}, \dots, \varphi_{R}^{p} \left(x\right) + \vartheta_{R}^{p}\right) \\
        &= \left(\sum_{r \in \mathcal{I}_{R}} a^{(p)}_{1r} \sigma \left(\varphi_{r}^{p} \left(x\right) + \vartheta_{r}^{p}\right) , \dots, \sum_{r \in \mathcal{I}_{R}} a^{(p)}_{Mr} \sigma \left(\varphi_{r}^{p} \left(x\right) + \vartheta_{r}^{p}\right)\right),
    \end{split}
\end{equation}
where \(A^{(p)} = \left[a^{(p)}_{jr}\right] \in \mathbb{R}^{M \times R}\), \(\varphi_{r}^{p} \in \mathscr{B}_{\mathbb{R}} \left( \mathscr{X} \right)\), and \(\vartheta_{r}^{p} \in \mathbb{R}\). Consequently, we can write
\begin{equation}\label{eq:A.18}
    b_{jp} \left(x\right) = \sum_{r \in \mathcal{I}_{R}} a^{(p)}_{jr} \sigma \left(\varphi_{r}^{p} \left(x\right) + \vartheta_{r}^{p}\right),
\end{equation}
where \(B \left(x\right) = \left[b_{jp}\left(x\right)\right]\in \mathbb{R}^{M \times P}\). Note also that with this notation we have
\begin{equation}\label{eq:A.19}
    \boldsymbol{b}_{p} \left(x\right) = \left( b_{1p}\left(x\right), \dots, b_{Mp}\left(x\right) \right).
\end{equation}

Next, we introduce the following theorem, whose proof is detailed in \cite{Ismailov202600}. As mentioned at the beginning of this appendix, this result is of vital importance in proving both the main result of this appendix and Theorem \ref{theorem:Teo-3.1}.

\begin{theorem}[Ismailov's theorem \cite{Ismailov202600}]\label{theorem:Teo-A.1}
    Suppose that \(\left(\mathscr{X}, \left\Vert \cdot \right\Vert_{\mathscr{X}}\right)\) is a Banach space and that \(\sigma : \mathbb{R} \to \mathbb{R}\) is a Tauber-Wiener activation function. Let \(\mathscr{V} \subset \mathscr{X}\) and \(\mathcal{K} \subset \mathbb{R}^{n}\) be compact subsets, and let \(G : \mathscr{V} \to C_{\mathbb{R}^{M}} \left(\mathcal{K}\right)\) be continuous. Then, for any \(\epsilon > 0\), there exist \(P \in \mathbb{N}^{\ast}\), functions \(\tau_{p} : \mathcal{K} \to \mathbb{R}\) of the form \(\tau_{p} \left(\textbf{y}\right) = \sigma \left(\textbf{w}_{p} \cdot \textbf{y} + \zeta_{p}\right)\), \(p \in \mathcal{I}_{P}\), and a function \(B : \mathscr{X} \to \mathbb{R}^{M \times P}\) whose columns are topological neural networks in \(\mathscr{X}\) constructed as in Definition \ref{definition:Def-A.2}, such that:
    \begin{equation}\label{eq:A.20}
        \sup_{x \in \mathscr{V}} \sup_{\textbf{y} \in \mathcal{K}} \left\Vert G\left(x\right)\left(\textbf{y}\right) -B\left(x\right)\boldsymbol{\tau}\left(\textbf{y}\right) \right\Vert_{2,M} < \epsilon,
    \end{equation}
    where \(\boldsymbol{\tau}=\left(\tau_{1}, \dots, \tau_{P}\right)\). That is, the operator \(G\) admits an \(\epsilon\)-approximation on \(\mathscr{V}\times \mathcal{K}\) by a topological DeepONet with branch network \(B\) and trunk network \(\boldsymbol{\tau}\).
\end{theorem}
Now, we apply the above theorem to our problem of interest, starting by introducing some results that will be of great utility, in which \(\mathbb{K} \subset \mathbb{B}\) is as defined in Observation \ref{remark:Obs-A.1}.
\begin{lemma}\label{lemma:Lem-A.1}
    If \(\mathbb{V}\subset C_{\mathbb{R}}\left(\mathbb{K}\right)\) is a compact subset, then \(\mathbb{V}\) is uniformly bounded and equicontinuous. Consequently, \(\mathbb{V}^{M}\), for \(M \ge 1\), is uniformly bounded and equicontinuous.
\end{lemma}
\begin{proof}
    We begin the proof by verifying that \(\mathbb{V}\) is uniformly bounded. Since every compact set is closed and bounded, and since we are considering the uniform norm, it follows that there exists \(S > 0\) such that:
    \begin{equation}\label{A.21}
        \sup_{x \in \mathbb{K}} \left\Vert f\left(x\right)\right\Vert = \left\Vert f\right\Vert_{\infty} \le S,\, \forall \, f \in \mathbb{V}.
    \end{equation}
    That is, \(\mathbb{V}\) is uniformly bounded. Furthermore, to verify that \(\mathbb{V}\) is equicontinuous, since \(\mathbb{V}\) is compact, it is totally bounded; that is, for any \(\epsilon > 0\), there exist \( f_{1}, \dots, f_{N} \in C_{\mathbb{R}}\left(\mathbb{K}\right)\), with \( N = N \left(\epsilon\right)\), such that
    \begin{equation}\label{eq:A.22}
        \mathbb{V} \subset \bigcup_{j \in \mathcal{I}_{N}} B \left(f_{j} ; \frac{\epsilon}{3}\right),
    \end{equation}
    with \(\mathcal{I}_{N} = \{1, \dots, N\}\). Thus, from the inclusion above, we have that for each \( f \in \mathbb{V} \) there exists a \(1 \le j^{\prime} \le N\) such that:
    \begin{equation}\label{eq:A.23}
        \left\Vert f_{j^{\prime}} - f\right\Vert_{\infty} < \frac{\epsilon}{3}.
    \end{equation}
    Now, since each \(f_{j}\) is continuous on the compact set \(\mathbb{K}\), \(f_{j}\) is uniformly continuous; consequently, there exists \(\delta_{j} > 0\) such that
    \begin{equation}\label{eq:A.24}
        \left\Vert x - u \right\Vert_{\mathbb{B}} < \delta_{j} \, \Longrightarrow \, \left| f_{j} \left(x\right) - f_{j} \left(u\right) \right| < \frac{\epsilon}{3}.
    \end{equation}
    Let \(\delta = \min \left\lbrace \delta_{1}, \dots, \delta_{N} \right\rbrace > 0\). Then, from (\ref{eq:A.24}), we have that
    \begin{equation}\label{eq:A.25}
        \left\Vert x - u \right\Vert_{\mathbb{B}} < \delta \, \Longrightarrow \, \left| f_{j} \left(x\right) - f_{j} \left(u\right) \right| < \frac{\epsilon}{3}, \, \forall \, j \in \mathcal{I}_{N}.
    \end{equation}
    From the above, we have that if \(\left\Vert x - u \right\Vert_{\mathbb{B}} < \delta\), then
    \begin{equation}\label{eq:A.26}
        \begin{split}
            \left| f\left(x\right) - f\left(u\right) \right| &\le \left| f \left(x\right) - f_{j^{\prime}} \left(x\right) \right| + \left| f_{j^{\prime}} \left(x\right) - f_{j^{\prime}} \left(u\right) \right| + \left| f_{j^{\prime}} \left(u\right) - f \left(u\right) \right| \\
            &\overset{(\ref{eq:A.23})}{<} \frac{2\epsilon}{3} + \left| f_{j^{\prime}} \left(x\right) -f_{j^{\prime}} \left(u\right) \right| \\
            &\overset{(\ref{eq:A.25})}{<} \, \epsilon,
        \end{split}
    \end{equation}
    for all \(f \in \mathbb{V}\). That is, \(\mathbb{V}\) is a family of equicontinuous functions.
\end{proof} However, it should be pointed out that:
\begin{remark}\label{remark:Obs-A.2}
    Lemma \ref{lemma:Lem-A.1} can be easily generalized to the case where we consider any compact subset of \(C_{\mathbb{R}^{M}} \left(\mathbb{K}\right)\), \(M \ge 1\), taking the uniform norm into account.
\end{remark}

\begin{lemma}\label{lemma:Lem-A.2}
    Under the same conditions and notations defined in Observation \ref{remark:Obs-A.1}, let \(\mathbb{V}\) be a compact subset of \(C_{\mathbb{R}}\left(\mathbb{K}\right)\). Given \(\mu \in \mathscr{B}_{\mathbb{R}} \left( C_{\mathbb{R}}\left(\mathbb{K} \right) \right)\) and \(\epsilon > 0\), there exist \(x_{1}, \dots, x_{L} \in \mathbb{K}\) and scalars \(\xi_{1}, \dots, \xi_{L} \in \mathbb{R}\), with \(L = L \left(\epsilon\right) \in \mathbb{N}^{\ast}\) and each \(\xi_{j}\) depending on \(\mu\) and \(\epsilon\), such that:
    \begin{equation}\label{eq:A.27}
        \sup_{f \in \mathbb{V}} \left| \mu \left(f\right) - \sum_{j \in \mathcal{I}_{L}} \xi_{j} f \left(x_{j}\right) \right| < \epsilon.
    \end{equation}
\end{lemma}
\begin{proof}
    The proof begins by considering \(\epsilon > 0\) and \(\mu \in \mathscr{B}_{\mathbb{R}} \left( C_{\mathbb{R}}\left(\mathbb{K} \right) \right)\), chosen arbitrarily. Now, recalling that \(\mathbb{V}\) is compact, it follows from Lemma \ref{lemma:Lem-A.1} that \(\mathbb{V}\) is equicontinuous, that is, there exists \(\delta = \delta \left(\epsilon\right) > 0\) such that:
    \begin{equation}\label{eq:A.28}
        \left\Vert x - y \right\Vert_{\mathbb{B}} < \delta \, \Longrightarrow \, \left| f \left(x\right) - f \left(y\right) \right| < \epsilon, 
    \end{equation} 
    for all \(f \in \mathbb{V}\). Also, since every compact set is totally bounded, we have that \(\mathbb{K}\) is totally bounded; that is, there exist \(\mathcal{A}_{1}, \dots, \mathcal{A}_{L}\) open balls of radius \(\delta/2\) that form a cover of \(\mathbb{K}\). These allow us to construct the following disjoint Borel sets:
    \begin{equation}\label{eq:A.29}
        \mathcal{V}_{1} = \mathcal{A}_{1} \quad \mbox{and} \quad \mathcal{V}_{j} = \mathcal{A}_{j} \setminus \left(\bigcup_{k=1}^{j-1} \mathcal{A}_{k} \right),
    \end{equation}
    for \(j = 2, \dots, L\). Note that:
    \begin{equation}\label{eq:A.30}
        \operatorname{diam} \left(\mathcal{V}_j\right) \le \operatorname{diam} \left(\mathcal{A}_j\right)<\delta,
    \end{equation}
    for each \(j\) and
    \begin{equation}\label{eq:A.31}
        \mathbb{K} \subset \bigcup_{j \in \mathcal{I}_{L}} \mathcal{A}_{j} = \bigsqcup_{j \in \mathcal{I}_{L}} \mathcal{V}_{j}.
    \end{equation}
    Thus, for each \(\mathcal{V}_{j}\), we choose an \(x_{j} \in  \mathcal{V}_{j}\). Hence, for any \( f \in \mathbb{V}\),
    \begin{equation}\label{eq:A.32}
        \begin{split}
            \left| \mu \left(f\right) - \sum_{j \in \mathcal{I}_{L}} \mu \left(\mathcal{V}_{j} \cap \mathbb{K}\right) f \left(x_{j}\right) \right| &= \left| \int_{\mathbb{K}} f \, \operatorname{d} \mu - \sum_{j \in \mathcal{I}_{L}} \int_{\mathbb{K} \cap \mathcal{V}_{j}} f \left(x_{j}\right) \operatorname{d} \mu \right| \\
            &= \left| \sum_{j \in \mathcal{I}_{L}} \int_{\mathbb{K} \cap \mathcal{V}_{j}} f \, \operatorname{d} \mu -\sum_{j \in \mathcal{I}_{L}} \int_{\mathbb{K} \cap \mathcal{V}_{j}} f \left(x_{j}\right) \operatorname{d} \mu \right| \\
            &\le \sum_{j \in \mathcal{I}_{L}} \int_{\mathbb{K} \cap \mathcal{V}_{j}} \left| f \left(x\right) - f \left(x_{j}\right) \right| \operatorname{d} \left| \mu \right| \left(x\right) \\
            &\overset{(\ref{eq:A.28})}{\le} \epsilon \sum_{j \in \mathcal{I}_{L}} \int_{\mathbb{K} \cap \mathcal{V}_{j}} \operatorname{d} \left| \mu \right| \left(x\right) \\
            &= \, \epsilon \left| \mu \right| \left(\mathbb{K}\right) \\
            &< \infty.
        \end{split}
    \end{equation}
    Now, let us consider two cases. In the first case, if \(\left| \mu \right| \left(\mathbb{K}\right) = 0\), which is equivalent to stating that \(\mathbb{K}\) is \(\mu\)-null, then from (\ref{eq:A.32}) we obtain:
    \begin{equation}\label{eq:A.33}
        \mu \left(f\right) = \sum_{j \in \mathcal{I}_{N}} \mu \left(\mathcal{V}_{j} \cap \mathbb{K}\right) f \left(x_{j}\right) = 0,
    \end{equation}
    for all \(f \in \mathbb{V}\). In the second case, if \(\left| \mu \right| \left(\mathbb{K}\right) > 0\), we can redefine \(\epsilon\) as \(\epsilon/\left[ \left| \mu \right| \left(\mathbb{K}\right)\right]\) in (\ref{eq:A.28}). That is, in either case, from (\ref{eq:A.32}) we can conclude that
    \begin{equation}\label{eq:A.34}
        \left| \mu \left(f\right) - \sum_{j \in \mathcal{I}_{L}} \mu \left(\mathcal{V}_{j} \cap \mathbb{K}\right) f \left(x_{j}\right) \right| < \epsilon,
    \end{equation}
    for all \(f \in \mathbb{V}\). Therefore, it suffices to take \(\xi_{j} = \mu \left(\mathcal{V}_{j} \cap \mathbb{K}\right)\) for each \(j\), and the lemma is proven.
\end{proof}
\begin{corollary}\label{corollary:Col-A.1}
    Under the same conditions and notations of Definition \ref{definition:Def-A.1}, let \(\mathscr{F}\) now be a compact subset of \(C_{\mathbb{R}^{M}} \left(\mathbb{K}\right)\). Given \(\boldsymbol{\mu} = \left(\mu_{1}, \dots, \mu_{M}\right) \in \mathscr{B}_{\mathbb{R}} \left( C_{\mathbb{R}^{M}}\left(\mathbb{K} \right) \right)\) and \(\epsilon > 0\), there exist \(x_{jl} \in \mathbb{K}\) and scalars \(\xi_{jl} \in \mathbb{R}\), for \(j \in \mathcal{I}_{M}= \{1, \dots, M\}\) and \(l \in \mathcal{I}_{L} = \{1, \dots, L\}\), with \(M, L = L \left(\epsilon\right) \in \mathbb{N}^{\ast}\) and each \(\xi_{jl}\) depending on \(\mu_{j}\) and \(\epsilon\), such that:
    \begin{equation}\label{eq:A.35}
        \left|\boldsymbol{\mu} \left(\textbf{f}\right) - \sum_{\substack{j \in \mathcal{I}_{M} \\ l \in \mathcal{I}_{L}}} \xi_{jl} f_{j} \left(x_{jl}\right) \right| < \epsilon,
    \end{equation}
    for all \(\textbf{f} = \left(f_{1}, \dots, f_{M}\right) \in \mathscr{F}\).
\end{corollary}
\begin{proof}
    Consider \( \epsilon > 0\) and \(\boldsymbol{\mu} = \left(\mu_{1}, \dots, \mu_{M}\right) \in \mathscr{B}_{\mathbb{R}} \left( C_{\mathbb{R}^{M}}\left(\mathbb{K} \right) \right)\), chosen arbitrarily. Thus, for any \(\textbf{f} = \left(f_{1}, \dots, f_{M}\right) \in \mathscr{F}\), each \(f_{j}\) belongs to some compact subset \(\mathscr{F}_{j}\) of \(C_{\mathbb{R}}\left(\mathbb{K}\right)\). Therefore, thanks to Lemma \ref{lemma:Lem-A.2}, for any \(j \in \mathcal{I}_{M}\), there exists \(L=L\left(\epsilon\right) \in \mathbb{N}^{\ast}\) for which there exist scalars \(\xi_{jl} \in  \mathbb{R}\), depending on \(\mu_{j}\) and \(\epsilon\), where \(l \in \mathcal{I}_{L}\), such that:
    \begin{equation}\label{eq:A.36}
        \left| \mu_{j} \left(f_{j}\right) - \sum_{l \in \mathcal{I}_{L}} \xi_{jl} f_{j} \left(x_{jl}\right) \right| < \frac{\epsilon}{2M}.
    \end{equation}
    Hence,
    \begin{equation}\label{eq:A.37}
        \begin{split}
            \left| \boldsymbol{\mu} \left(\textbf{f}\right) - \sum_{\substack{j \in \mathcal{I}_{M} \\ l \in \mathcal{I}_{L}}} \xi_{jl} f_{j} \left(x_{jl}\right) \right| &= \left| \int_{\mathbb{K}}\textbf{f} \cdot \boldsymbol{\mu} - \sum_{\substack{j \in \mathcal{I}_{M} \\ l \in \mathcal{I}_{L}}} \xi_{jl} f_{j} \left(x_{jl}\right) \right| \\
            &= \left| \sum_{j \in \mathcal{I}_{M}} \int_{\mathbb{K}} f_{j} \, \operatorname{d} \mu_{i} - \sum_{\substack{j \in \mathcal{I}_{M} \\ l \in \mathcal{I}_{L}}} \xi_{jl} f_{j} \left(x_{jl}\right) \right| \\
            &\le \sum_{j \in \mathcal{I}_{M}} \left| \int_{\mathbb{K}} f_{j} \, \operatorname{d} \mu_{j} - \sum_{l \in \mathcal{I}_{L}} \xi_{jl} f_{j} \left(x_{jl}\right) \right| \\
            &\overset{(\ref{eq:A.36})}{<} \epsilon
        \end{split}
    \end{equation}
    for any \(\textbf{f} = \left(f_{1}, \dots, f_{M}\right) \in \mathscr{F}\).
\end{proof}

Finally, we will contextualize Theorem \ref{theorem:Teo-A.1} for the case where \(\mathscr{X} = C_{\mathbb{R}^{M}} \left(\mathbb{K}\right)\) is endowed with the uniform norm.
\begin{theorem}\label{theorem:Teo-A.2}
    Suppose that \(\sigma : \mathbb{R} \to \mathbb{R}\) is a Tauber-Wiener activation function. Let \(\mathscr{V} \subset C_{\mathbb{R}^{M}} \left(\mathbb{K}\right)\) and \(\mathcal{K} \subset D\) be compact subsets, and let \(\boldsymbol{\mathcal{G}} : \mathscr{V} \to C_{\mathbb{R}^{M}} \left(\mathcal{K}\right)\) be continuous. Then, for any \(\epsilon > 0\), there exist \(P \in \mathbb{N}^{\ast}\), functions \(\tau_{p} : \mathcal{K} \to \mathbb{R}\) of the form \( \tau_{p} \left(\textbf{y}\right) = \sigma \left(\textbf{w}_{p} \cdot \textbf{y} + \zeta_{p}\right)\), \(p \in \mathcal{I}_{P} = \{1, \dots, P\}\), and a function \(B : C_{\mathbb{R}^{M}} \left(\mathbb{K}\right) \to \mathbb{R}^{M \times P}\) of the form
    \begin{equation}\label{eq:A.38}
        b_{jp} \left(\textbf{f}\right) = \sum_{r \in \mathcal{I}_{R}} a^{(p)}_{jr} \sigma \left(\sum_{\substack{\nu \in \mathcal{I}_{M} \\ l \in \mathcal{I}_{L}}} \xi_{\nu l}^{pr}f_{\nu} \left(x_{\nu l}\right) + \vartheta_{r}^{p}\right),
    \end{equation}
    for all \(\textbf{f} = \left(f_{1}, \dots, f_{M}\right) \in C_{\mathbb{R}^{M}} \left(\mathbb{K}\right)\), where \(B \left(\textbf{f}\right) = \left[b_{jp}\left(\textbf{f}\right)\right] \in \mathbb{R}^{M \times P}\), such that:
    \begin{equation}\label{eq:A.39}
        \sup_{\textbf{f} \in \mathscr{V}} \left\lbrace \sup_{\textbf{y} \in \mathcal{K}} \left\Vert \boldsymbol{\mathcal{G}} \left(\textbf{f}\right) \left(\textbf{y}\right) - B \left(\textbf{f}\right) \boldsymbol{\tau} \left(\textbf{y}\right) \right\Vert_{2,M} \right\rbrace < \epsilon,
    \end{equation}
    with \(\boldsymbol{\tau} = \left(\tau_{1}, \dots, \tau_{P}\right)\).
\end{theorem}
\begin{proof}
    Consider \(\epsilon > 0\) chosen arbitrarily. Theorem \ref{theorem:Teo-A.1} guarantees the existence of \(P \in \mathbb{N}^{\ast}\), functions \(\tau_{p} : \mathcal{K} \to \mathbb{R}\) of the form \(\tau_{p} \left(\textbf{y}\right) = \sigma \left(\textbf{w}_{p} \cdot \textbf{y} + \zeta_{p}\right)\), \(p \in \mathcal{I}_{P}\), and a function \( Q : C_{\mathbb{R}^{M}} \left(\mathbb{K}\right) \to \mathbb{R}^{M \times P}\) whose columns are topological neural networks in \(C_{\mathbb{R}^{M}} \left(\mathbb{K}\right)\) constructed as in (\ref{eq:A.16}), such that:
    \begin{equation}\label{eq:A.40}
        \sup_{\textbf{f} \in \mathscr{V}} \left\lbrace \sup_{\textbf{y} \in \mathcal{K}} \left\Vert \boldsymbol{\mathcal{G}} \left(\textbf{f}\right) \left(\textbf{y}\right) - Q \left(\textbf{f}\right) \boldsymbol{\tau} \left(\textbf{y}\right) \right\Vert_{2,M} \right\rbrace < \frac{\epsilon}{2},
    \end{equation}
    where \(\boldsymbol{\tau} = \left(\tau_{1}, \dots, \tau_{P}\right)\). Furthermore, according to (\ref{eq:A.18}), each entry of \(Q \left(\textbf{f}\right)\) is of the form
    \begin{equation}\label{eq:A.41}
        q_{jp} \left(\textbf{f}\right) = \sum_{r \in \mathcal{I}_{R}} a^{(p)}_{jr} \sigma \left(\boldsymbol{\mu}_{r}^{p} \left(\textbf{f}\right) + \vartheta_{r}^{p}\right),
    \end{equation}
    where each \(\boldsymbol{\mu}_{r}^{p}\) belongs to \(\mathscr{B}_{\mathbb{R}} \left( C_{\mathbb{R}^{M}}\left(\mathbb{K} \right) \right)\). On the other hand, thanks to Corollary \ref{corollary:Col-A.1}, given any \(\delta > 0\), we have the following approximation:
    \begin{equation}\label{eq:A.42}
        \left| \boldsymbol{\mu}_{r}^{p} \left(\textbf{f}\right) - \sum_{\substack{\nu \in \mathcal{I}_{M} \\ l \in \mathcal{I}_{L}}} \xi_{\nu l}^{pr} f_{\nu} \left(x_{\nu l}\right)\right| < \delta,
    \end{equation}
    for all \(\textbf{f} = \left(f_{1}, \dots, f_{M}\right) \in \mathscr{V}\). Then, since \(\sigma\) is uniformly continuous on compact subsets of \(\mathbb{R}\), we can choose \(0 < \delta = \delta_{0} = \delta_{0} \left(\epsilon\right) \ll 1\), in (\ref{eq:A.42}), such that
    \begin{equation}\label{eq:A.43}
        \sup_{\textbf{f} \in \mathscr{V}} \left|\sum_{r \in \mathcal{I}_{R}} a^{(p)}_{jp} \sigma \left(\boldsymbol{\mu}_{r}^{p} \left(\textbf{f}\right) + \vartheta_{r}^{p}\right) - \sum_{r \in \mathcal{I}_{R}} a^{(p)}_{jr} \sigma \left(\sum_{\substack{\nu \in \mathcal{I}_{M} \\  l \in \mathcal{I}_{L}}} \xi_{\nu l}^{pr} f_{\nu} \left(x_{\nu l}\right) + \vartheta_{r}^{p}\right)\right| < \frac{\epsilon}{2PM\tilde{S}},
    \end{equation} 
    where
    \begin{equation}\label{eq:A.44}
        \tilde{S} = \max_{p \in \mathcal{I}_{P}} \left\lbrace \sup_{\textbf{y} \in \mathcal{K}} \left|\tau_{p} \left(\textbf{y}\right)\right|\right\rbrace.
    \end{equation}
    Finally, if \(B : C_{\mathbb{R}^{M}} \left(\mathbb{K}\right) \to \mathbb{R}^{M \times P}\) is defined by
    \begin{equation}\label{eq:A.45}
        b_{jp} \left(\textbf{f}\right) = \sum_{r \in \mathcal{I}_{R}} a^{(p)}_{jr} \sigma \left(\sum_{\substack{\nu \in \mathcal{I}_{M} \\ l \in \mathcal{I}_{L}}} \xi_{\nu l}^{pr} f_{\nu} \left(x_{\nu l}\right) + \vartheta_{r}^{p}\right),
    \end{equation}
    for all \(\textbf{f} = \left(f_{1}, \dots, f_{M}\right) \in C_{\mathbb{R}^{M}} \left(\mathbb{K}\right)\), where \(B \left(\textbf{f}\right) = \left[b_{jp}\left(\textbf{f}\right)\right] \in \mathbb{R}^{M \times P}\), we have
    \begin{equation}\label{eq:A.46}
        \begin{split}
            \left\Vert \boldsymbol{\mathcal{G}}\left(\textbf{f}\right) \left(\textbf{y}\right) - B\left(\textbf{f}\right)\boldsymbol{\tau}\left(\textbf{y}\right) \right\Vert_{2,M} &\le \left\Vert \boldsymbol{\mathcal{G}}\left(\textbf{f}\right)\left(\textbf{y}\right)-Q\left(\textbf{f}\right)\boldsymbol{\tau}\left(\textbf{y}\right)\right\Vert_{2,M} \\
            &+ \left\Vert Q\left(\textbf{f}\right)\boldsymbol{\tau}\left(\textbf{y}\right)-B\left(\textbf{f}\right)\boldsymbol{\tau}\left(\textbf{y}\right)\right\Vert_{2,M} \\
            &\overset{(\ref{eq:A.40})}{<} \frac{\epsilon}{2} + \left\Vert Q\left(\textbf{f}\right)\boldsymbol{\tau}\left(\textbf{y}\right)-B\left(\textbf{f}\right)\boldsymbol{\tau}\left(\textbf{y}\right) \right\Vert_{2,M} \\
            &\overset{(\ref{eq:A.43})}{<} \frac{\epsilon}{2} + \frac{\epsilon}{2} \\
            &= \epsilon,
        \end{split}
    \end{equation}
    for all \(\textbf{f} \in \mathscr{V}\) and \(\textbf{y} \in \mathcal{K}\). Thus, the theorem is proven.
\end{proof}
\begin{corollary}\label{corollary:Col-A.2}
    Under the same conditions of Theorem \ref{theorem:Teo-A.2}, we write
    \begin{equation}\label{eq:A.47}
        \boldsymbol{\mathcal{G}} = \left(\mathcal{G}_{1}, \dots, \mathcal{G}_{M}\right)
    \end{equation}
    that is, \(\mathcal{G}_{j} : \mathscr{V} \to C_{\mathbb{R}}\left(\mathcal{K}\right)\) is continuous for each \(j \in \mathcal{I}_{M} = \{1, \dots, M\}\). Then, for any \(\epsilon > 0\), there exist \(P \in \mathbb{N}^{\ast}\), functions \(\tau_{p} : \mathcal{K} \to \mathbb{R}\) of the form \(\tau_{p} \left(\textbf{y}\right) = \sigma \left(\textbf{w}_{p} \cdot \textbf{y} + \zeta_{p}\right)\), \(p \in \mathcal{I}_{P} = \{1, \dots, P\}\), and functions \(\boldsymbol{\hat{\beta}}_{j} : C_{\mathbb{R}^{M}} \left(\mathbb{K}\right) \to \mathbb{R}^{P}\) of the form
    \begin{equation}\label{eq:A.48}
        \begin{split}
            &\boldsymbol{\hat{\beta}}_{j} \left(\textbf{f}\right) = \left( \sum_{r \in \mathcal{I}_{R}} a^{(1)}_{jr} \sigma \left(\sum_{\substack{\nu \in \mathcal{I}_{M} \\ l \in \mathcal{I}_{L}}} \xi_{\nu l}^{1r} f_{\nu} \left(x_{\nu l}\right) + \vartheta_{r}^{1}\right), \dots \right. \\
            &\left. \hspace{5.0cm} \dots, \sum_{r \in \mathcal{I}_{R}} a^{(p)}_{jr} \sigma \left(\sum_{\substack{\nu \in \mathcal{I}_{M} \\ l \in \mathcal{I}_{L}}} \xi_{\nu l}^{Pr} f_{\nu} \left(x_{\nu l}\right) + \vartheta_{r}^{P}\right)\right),
        \end{split}
    \end{equation}
    for \(\textbf{f} = \left(f_{1}, \dots, f_{M}\right) \in C_{\mathbb{R}^{M}} \left(\mathbb{K}\right)\), such that for each \(j\) we have
    \begin{equation}\label{eq:A.49}
        \left|\boldsymbol{\hat{\beta}}_{j} \left(\textbf{f}\right) \cdot \boldsymbol{\tau} \left(\textbf{y}\right) - \mathcal{G}_{j} \left(\textbf{f}\right) \left(\textbf{y}\right)\right| < \epsilon,
    \end{equation}
    for all \(\textbf{f} \in \mathscr{V}\), \(\textbf{y} \in \mathcal{K}\), and where \(\boldsymbol{\tau} = \left(\tau_{1}, \dots, \tau_{P}\right)\).
\end{corollary}
\begin{proof}
    It is a straightforward reformulation of Theorem \ref{theorem:Teo-A.2}; indeed, \(\boldsymbol{\hat{\beta}}_{j}\left(\textbf{f}\right)\) is the \(j\)-th row of the matrix \(B \left(\textbf{f}\right) \in \mathbb{R}^{M \times P}\) from the conclusion of that theorem.
\end{proof}

Finally, the proof for the case of deep architectures is derived from the function composition property. Since each hidden layer refines the representation space, the approximation result established for shallow networks (Theorem \ref{theorem:Teo-A.2} and Corollary \ref{corollary:Col-A.2}) is extended and amplified, thereby ensuring the effectiveness of more complex architectures.
\section{Universal approximation lemma via Stone-Weierstrass}\label{sec:Ap-B}
In this appendix, we will prove that the functions defined in (\ref{eq:31}) are effectively universal approximators. To achieve this, it is necessary to state and prove the following lemma.
\begin{lemma}\label{lemma:Lem-B.1}
    Let \(\mbox{OBS} \coloneq \{ \hat{O}^{\nu} \in \mathbb{O}^{\nu} \}_{\nu \in \{\mathcal{B},\mathcal{T}\}}\) be the set of Hermitian quantum observables acting on the Hilbert spaces \(\mathcal{H}_{\mathcal{Q}}^{\nu \in \{\mathcal{B},\mathcal{T}\}}\). Let \(\mathcal{K} \subset D\) be a compact set. Let \(\mathbb{V} \subset C_{\mathbb{R}} \left( \mathbb{K} \right)\) be a compact subset as in theorem \ref{theorem:Teo-3.1}, such that \(\mathbb{V}^{M}\) is a compact subset of \(C_{\mathbb{R}^{M}}^{k} \left( \mathbb{K} \right)\). Let \(\sigma : \mathbb{R} \to \mathbb{R}\) be a bijective Tauber-Wiener activation function, and let \(\mbox{MAC}\) be the set formed by all functions defined by (\ref{eq:31}), whose domain is the compact set \(\mathcal{K}\). Then, \(\mbox{MAC}\) satisfies the following properties:
    \begin{enumerate}
        \item{\(\mbox{MAC}\) is closed under scalar multiplication, addition, and function multiplication.}
        \item{\(\mbox{MAC}\) separates points on \(\mathcal{K}\).}
        \item{\(\mbox{MAC}\) contains all constant functions.}
    \end{enumerate}
    Consequently, by the Stone-Weierstrass theorem, \(\mbox{MAC}\) is dense in \(L_{\mathbb{R}}^{\infty} \left(\mathcal{K}\right)\). Therefore, for any function \(\mathfrak{g} \in L_{\mathbb{R}}^{\infty} \left(\mathcal{K}\right)\) and \(\epsilon_{j} > 0\), \( \boldsymbol{\mathfrak{v}} \in \mathbb{V}^{M} \), \( \boldsymbol{\tilde{\mathfrak{v}}}_{\flat} \in \mathbb{R}^{m M \times 1} \), \( \mathbf{y} \in \mathcal{K} \), there exist \(\hat{O}^{\nu} \in \mbox{OBS}\), \(\boldsymbol{\mathcal{W}}_{\flat}^{\mathcal{B} \ast} \), \( \boldsymbol{\mathcal{W}}^{\mathcal{T} \ast} \), and \( \boldsymbol{\Theta}_{p}^{\nu \ast \in \{\mathcal{B},\mathcal{T}\}} \), such that:
    \begin{equation}\label{eq:B.50}
        \left\Vert \mathfrak{G}_{j} \left( \boldsymbol{\tilde{\mathfrak{v}}}_{\flat}, \hat{O}^{\mathcal{B}}, \boldsymbol{\mathcal{W}}_{\flat}^{\mathcal{B} \ast}, \boldsymbol{\Theta}^{\mathcal{B} \ast}_{p} \right) \left( \cdot, \hat{O}^{\mathcal{T}}, \boldsymbol{\mathcal{W}}^{\mathcal{T} \ast}, \boldsymbol{\Theta}^{\mathcal{T} \ast}_{p} \right) - \mathfrak{g} \left( \cdot \right) \right\Vert_{\infty} \leq \epsilon_{j}.
    \end{equation}
\end{lemma}
\begin{proof}
    The proof of lemma \ref{lemma:Lem-B.1}, although not difficult, will be divided into five parts. In the first part, we verify that \(\mbox{MAC}\) is closed under multiplication by any scalar \(a \in \mathbb{R}\), which emerges directly from definitions (\ref{eq:31})--(\ref{eq:34}) by considering \(\hat{O}^{\prime \nu} \coloneq a \hat{O}^{\nu} \in \mbox{OBS}\), allowing us to obtain the expressions
    \begin{equation}\label{eq:B.51}
        \mathfrak{B}_{jp^{\prime}}^{\prime} \left( \cdot, \hat{O}^{\prime\mathcal{B}}, \boldsymbol{\mathcal{W}}_{\flat}^{\mathcal{B}}, \boldsymbol{\Theta}^{\mathcal{B}}_{p} \right) =        a \mathfrak{B}_{jp^{\prime}} \left( \cdot, \hat{O}^{\mathcal{B}}, \boldsymbol{\mathcal{W}}_{\flat}^{\mathcal{B}}, \boldsymbol{\Theta}^{\mathcal{B}}_{p} \right),
    \end{equation}
    \begin{equation}\label{eq:B.52}
        \mathfrak{T}_{p^{\prime}}^{\prime \prime} \left( \cdot, \hat{O}^{\prime \mathcal{T}}, \boldsymbol{\mathcal{W}}^{\mathcal{T}}, \boldsymbol{\Theta}^{\mathcal{T}}_{p} \right) =        a \mathfrak{T}_{p^{\prime}}\left( \cdot, \hat{O}^{\mathcal{T}}, \boldsymbol{\mathcal{W}}^{\mathcal{T}}, \boldsymbol{\Theta}^{\mathcal{T}}_{p} \right),
    \end{equation}
    \begin{equation}\label{eq:B.53}
        \mathfrak{b}_{j}^{\prime} \left( \hat{O}^{\prime \mathcal{B}}, \hat{O}^{\mathcal{T}}, \boldsymbol{\Theta}^{\mathcal{B}}_{p}, \boldsymbol{\Theta}^{\mathcal{T}}_{p} \right) =        a \mathfrak{b}_{j} \left( \hat{O}^{\mathcal{B}}, \hat{O}^{\mathcal{T}}, \boldsymbol{\Theta}^{\mathcal{B}}_{p}, \boldsymbol{\Theta}^{\mathcal{T}}_{p} \right),
    \end{equation}
    \begin{equation}\label{eq:B.54}
        \mathfrak{b}_{j}^{\prime \prime} \left( \hat{O}^{\mathcal{B}}, \hat{O}^{\prime \mathcal{T}}, \boldsymbol{\Theta}^{\mathcal{B}}_{p}, \boldsymbol{\Theta}^{\mathcal{T}}_{p} \right) =        a \mathfrak{b}_{j} \left( \hat{O}^{\mathcal{B}}, \hat{O}^{\mathcal{T}}, \boldsymbol{\Theta}^{\mathcal{B}}_{p}, \boldsymbol{\Theta}^{\mathcal{T}}_{p} \right).
    \end{equation}
    Therefore, it suffices to write
    \begin{equation}\label{eq:B.55}
        \begin{split}
            &a \mathfrak{G}_{j} \left( \cdot, \hat{O}^{\mathcal{B}}, \boldsymbol{\mathcal{W}}_{\flat}^{\mathcal{B}}, \boldsymbol{\Theta}^{\mathcal{B}}_{p} \right) \left( \cdot, \hat{O}^{\mathcal{T}}, \boldsymbol{\mathcal{W}}^{\mathcal{T}}, \boldsymbol{\Theta}^{\mathcal{T}}_{p} \right) = \left[a \mathfrak{B}_{jp^{\prime}} \left( \cdot, \hat{O}^{\mathcal{B}}, \boldsymbol{\mathcal{W}}_{\flat}^{\mathcal{B}}, \boldsymbol{\Theta}^{\mathcal{B}}_{p} \right)\right] \\
            &\hspace{4.0cm} \times \mathfrak{T}_{p^{\prime}}\left( \cdot, \hat{O}^{\mathcal{T}}, \boldsymbol{\mathcal{W}}^{\mathcal{T}}, \boldsymbol{\Theta}^{\mathcal{T}}_{p} \right) + a \mathfrak{b}_{j} \left( \hat{O}^{\mathcal{B}}, \hat{O}^{\mathcal{T}}, \boldsymbol{\Theta}^{\mathcal{B}}_{p}, \boldsymbol{\Theta}^{\mathcal{T}}_{p} \right) \\
            &= \mathfrak{B}_{jp^{\prime}}^{\prime} \left( \cdot, \hat{O}^{\prime \mathcal{B}}, \boldsymbol{\mathcal{W}}_{\flat}^{\mathcal{B}}, \boldsymbol{\Theta}^{\mathcal{B}}_{p} \right) \\
            &\hspace{4.0cm} \times \mathfrak{T}_{p^{\prime}}\left( \cdot, \hat{O}^{\mathcal{T}}, \boldsymbol{\mathcal{W}}^{\mathcal{T}}, \boldsymbol{\Theta}^{\mathcal{T}}_{p} \right) + \mathfrak{b}_{j}^{\prime} \left( \hat{O}^{\prime \mathcal{B}}, \hat{O}^{\mathcal{T}}, \boldsymbol{\Theta}^{\mathcal{B}}_{p}, \boldsymbol{\Theta}^{\mathcal{T}}_{p} \right),
        \end{split}
    \end{equation}
    from which it is clear that \( a \mathfrak{G}_{j} \in \mbox{MAC}\). 
    
    In the second part, we verify that \(\mbox{MAC}\) is closed under the addition of functions \(\mathfrak{G}_{j}^{\left( 1 \right)} \in \mbox{MAC}\) and \(\mathfrak{G}_{j}^{\left( 2 \right)} \in \mbox{MAC}\). To achieve this, we define the functions
    \begin{equation}\label{eq:B.56}
        \begin{split}
            &\mathfrak{G}_{j}^{\left( 1 \right)} \left( \cdot, \hat{O}^{\mathcal{B}}, \boldsymbol{\mathcal{W}}_{\flat}^{\mathcal{B}}, \boldsymbol{\Theta}^{\mathcal{B}}_{p} \right) \left( \cdot, \hat{O}^{\mathcal{T}}, \boldsymbol{\mathcal{W}}^{\mathcal{T}}, \boldsymbol{\Theta}^{\mathcal{T}}_{p} \right) = \mathfrak{B}_{jp^{\prime}}^{\left( 1 \right)} \left( \cdot, \hat{O}^{\mathcal{B}}, \boldsymbol{\mathcal{W}}_{\flat}^{\mathcal{B}}, \boldsymbol{\Theta}^{\mathcal{B}}_{p^{\prime}} \right) \\
            &\hspace{4.0cm} \times \mathfrak{T}_{p^{\prime}} \left( \cdot, \hat{O}^{\mathcal{T}}, \boldsymbol{\mathcal{W}}^{\mathcal{T}}, \boldsymbol{\Theta}^{\mathcal{T}}_{p} \right) + \mathfrak{b}_{j}^{\left( 1 \right)} \left( \hat{O}^{\mathcal{B}}, \hat{O}^{\mathcal{T}}, \boldsymbol{\Theta}^{\mathcal{B}}_{p}, \boldsymbol{\Theta}^{\mathcal{T}}_{p} \right)
        \end{split}
    \end{equation}
    and
    \begin{equation}\label{eq:B.57}
        \begin{split}
            &\mathfrak{G}_{j}^{\left( 2 \right)} \left( \cdot, \hat{O}^{\mathcal{B}}, \boldsymbol{\mathcal{W}}_{\flat}^{\mathcal{B}}, \boldsymbol{\Theta}^{\mathcal{B}}_{p} \right) \left( \cdot, \hat{O}^{\mathcal{T}}, \boldsymbol{\mathcal{W}}^{\mathcal{T}}, \boldsymbol{\Theta}^{\mathcal{T}}_{p} \right) = \mathfrak{B}_{jp^{\prime}}^{\left( 2 \right)} \left( \cdot, \hat{O}^{\mathcal{B}}, \boldsymbol{\mathcal{W}}_{\flat}^{\mathcal{B}}, \boldsymbol{\Theta}^{\mathcal{B}}_{p^{\prime}} \right) \\
            &\hspace{4.0cm} \times \mathfrak{T}_{p^{\prime}} \left( \cdot, \hat{O}^{\mathcal{T}}, \boldsymbol{\mathcal{W}}^{\mathcal{T}}, \boldsymbol{\Theta}^{\mathcal{T}}_{p} \right) + \mathfrak{b}_{j}^{\left( 2 \right)} \left( \hat{O}^{\mathcal{B}}, \hat{O}^{\mathcal{T}}, \boldsymbol{\Theta}^{\mathcal{B}}_{p}, \boldsymbol{\Theta}^{\mathcal{T}}_{p} \right).
        \end{split}
    \end{equation}
    Which, upon summing them, yields
    \begin{equation}\label{eq:B.58}
        \begin{split}
            &\mathfrak{G}_{j}^{\left( 1 \right)} \left( \cdot, \hat{O}^{\mathcal{B}}, \boldsymbol{\mathcal{W}}_{\flat}^{\mathcal{B}}, \boldsymbol{\Theta}^{\mathcal{B}}_{p} \right) \left( \cdot, \hat{O}^{\mathcal{T}}, \boldsymbol{\mathcal{W}}^{\mathcal{T}}, \boldsymbol{\Theta}^{\mathcal{T}}_{p} \right) + \mathfrak{G}_{j}^{\left( 2 \right)} \left( \cdot, \hat{O}^{\mathcal{B}}, \boldsymbol{\mathcal{W}}_{\flat}^{\mathcal{B}}, \boldsymbol{\Theta}^{\mathcal{B}}_{p} \right) \left( \cdot, \hat{O}^{\mathcal{T}}, \boldsymbol{\mathcal{W}}^{\mathcal{T}}, \boldsymbol{\Theta}^{\mathcal{T}}_{p} \right) \\
            &= \left[\mathfrak{B}_{jp^{\prime}}^{\left( 1 \right)} \left( \cdot, \hat{O}^{\mathcal{B}}, \boldsymbol{\mathcal{W}}_{\flat}^{\mathcal{B}}, \boldsymbol{\Theta}^{\mathcal{B}}_{p^{\prime}} \right) + \mathfrak{B}_{jp^{\prime}}^{\left( 2 \right)} \left( \cdot, \hat{O}^{\mathcal{B}}, \boldsymbol{\mathcal{W}}_{\flat}^{\mathcal{B}}, \boldsymbol{\Theta}^{\mathcal{B}}_{p} \right)\right]\mathfrak{T}_{p^{\prime}} \left( \cdot, \hat{O}^{\mathcal{T}}, \boldsymbol{\mathcal{W}}^{\mathcal{T}}, \boldsymbol{\Theta}^{\mathcal{T}}_{p} \right) \\
            &\hspace{4.0cm}+ \left[ \mathfrak{b}_{j}^{\left( 1 \right)} \left( \hat{O}^{\mathcal{B}}, \hat{O}^{\mathcal{T}}, \boldsymbol{\Theta}^{\mathcal{B}}_{p}, \boldsymbol{\Theta}^{\mathcal{T}}_{p} \right) + \mathfrak{b}_{j}^{\left( 2 \right)} \left( \hat{O}^{\mathcal{B}}, \hat{O}^{\mathcal{T}}, \boldsymbol{\Theta}^{\mathcal{B}}_{p}, \boldsymbol{\Theta}^{\mathcal{T}}_{p} \right)\right],
        \end{split}
    \end{equation}
    where it is clear that the first term on the right-hand side of (\ref{eq:B.58}) corresponds to the form of a stacked classical-quantum hybrid DeepONet architecture. Furthermore, since the sum of biases is also a bias, we see that (\ref{eq:B.58}) reduces to 
    \begin{equation}\label{eq:B.59}
        \begin{split}
            &\mathfrak{G}_{j}^{\left( 3 \right)} \left( \cdot, \hat{O}^{\mathcal{B}}, \boldsymbol{\mathcal{W}}_{\flat}^{\mathcal{B}}, \boldsymbol{\Theta}^{\mathcal{B}}_{p} \right) \left( \cdot, \hat{O}^{\mathcal{T}}, \boldsymbol{\mathcal{W}}^{\mathcal{T}}, \boldsymbol{\Theta}^{\mathcal{T}}_{p} \right) = \mathfrak{B}_{jp^{\prime}}^{\left( 3 \right)} \left( \cdot, \hat{O}^{\mathcal{B}}, \boldsymbol{\mathcal{W}}_{\flat}^{\mathcal{B}}, \boldsymbol{\Theta}^{\mathcal{B}}_{p} \right) \\
            &\hspace{4.0cm} \times \mathfrak{T}_{p^{\prime}} \left( \cdot, \hat{O}^{\mathcal{T}}, \boldsymbol{\mathcal{W}}^{\mathcal{T}}, \boldsymbol{\Theta}^{\mathcal{T}}_{p} \right) + \mathfrak{b}_{j}^{\left( 3 \right)} \left( \hat{O}^{\mathcal{B}}, \hat{O}^{\mathcal{T}}, \boldsymbol{\Theta}^{\mathcal{B}}_{p}, \boldsymbol{\Theta}^{\mathcal{T}}_{p} \right).
        \end{split}
    \end{equation}
    Therefore, it is clear that \( \mathfrak{G}_{j}^{\left( 3 \right)} \in \mbox{MAC}\). In the third part, we verify that \(\mbox{MAC}\) is closed under the multiplication of functions \(\mathfrak{G}_{j}^{\left( 1 \right)} \in \mbox{MAC}\) and \(\mathfrak{G}_{j}^{\left( 2 \right)} \in \mbox{MAC}\). To achieve this, we write the expression for the product:
    \begin{equation}\label{eq:B.60}
        \begin{split}
            &\mathfrak{G}_{j}^{\left( 1 \right)} \left( \cdot, \hat{O}^{\mathcal{B}}, \boldsymbol{\mathcal{W}}_{\flat}^{\mathcal{B}}, \boldsymbol{\Theta}^{\mathcal{B}}_{p} \right) \left( \cdot, \hat{O}^{\mathcal{T}}, \boldsymbol{\mathcal{W}}^{\mathcal{T}}, \boldsymbol{\Theta}^{\mathcal{T}}_{p} \right) \mathfrak{G}_{j}^{\left( 2 \right)} \left( \cdot, \hat{O}^{\mathcal{B}}, \boldsymbol{\mathcal{W}}_{\flat}^{\mathcal{B}}, \boldsymbol{\Theta}^{\mathcal{B}}_{p} \right) \left( \cdot, \hat{O}^{\mathcal{T}}, \boldsymbol{\mathcal{W}}^{\mathcal{T}}, \boldsymbol{\Theta}^{\mathcal{T}}_{p} \right) \\
            &= \displaystyle{\sum_{\substack{p^{\prime} \in \mathcal{I}_{P^{\prime}} \\ q^{\prime} \in \mathcal{I}_{Q^{\prime}}}}} \left[\mathfrak{B}_{jp^{\prime}}^{\left( 1 \right)} \left( \cdot, \hat{O}^{\mathcal{B}}, \boldsymbol{\mathcal{W}}_{\flat}^{\mathcal{B}}, \boldsymbol{\Theta}^{\mathcal{B}}_{p} \right) \mathfrak{B}_{jq^{\prime}}^{\left( 2 \right)} \left( \cdot, \hat{O}^{\mathcal{B}}, \boldsymbol{\mathcal{W}}_{\flat}^{\mathcal{B}}, \boldsymbol{\Theta}^{\mathcal{B}}_{p} \right)\right] \\
            &\times \left[ \mathfrak{T}_{p^{\prime}} \left( \cdot, \hat{O}^{\mathcal{T}}, \boldsymbol{\mathcal{W}}^{\mathcal{T}}, \boldsymbol{\Theta}^{\mathcal{T}}_{p} \right) \mathfrak{T}_{q^{\prime}} \left( \cdot, \hat{O}^{\mathcal{T}}, \boldsymbol{\mathcal{W}}^{\mathcal{T}}, \boldsymbol{\Theta}^{\mathcal{T}}_{p} \right)\right] \\
            &+ \left[\mathfrak{b}_{j}^{\left( 2 \right)} \left( \hat{O}^{\mathcal{B}}, \hat{O}^{\mathcal{T}}, \boldsymbol{\Theta}^{\mathcal{B}}_{p}, \boldsymbol{\Theta}^{\mathcal{T}}_{p} \right) \displaystyle{\sum_{p^{\prime} \in \mathcal{I}_{P^{\prime}}}} \mathfrak{B}_{jp^{\prime}}^{\left( 1 \right)} \left( \cdot, \hat{O}^{\mathcal{B}}, \boldsymbol{\mathcal{W}}_{\flat}^{\mathcal{B}}, \boldsymbol{\Theta}^{\mathcal{B}}_{p} \right)  \mathfrak{T}_{p^{\prime}} \left( \cdot, \hat{O}^{\mathcal{T}}, \boldsymbol{\mathcal{W}}^{\mathcal{T}}, \boldsymbol{\Theta}^{\mathcal{T}}_{p} \right) \right. \\
            &\left. + \mathfrak{b}_{j}^{\left( 1 \right)} \left( \hat{O}^{\mathcal{B}}, \hat{O}^{\mathcal{T}}, \boldsymbol{\Theta}^{\mathcal{B}}_{p}, \boldsymbol{\Theta}^{\mathcal{T}}_{p} \right) \displaystyle{\sum_{q^{\prime} \in \mathcal{I}_{Q^{\prime}}}} \mathfrak{B}_{jq^{\prime}}^{\left( 2 \right)} \left( \cdot, \hat{O}^{\mathcal{B}}, \boldsymbol{\mathcal{W}}_{\flat}^{\mathcal{B}}, \boldsymbol{\Theta}^{\mathcal{B}}_{p} \right)  \mathfrak{T}_{q^{\prime}} \left( \cdot, \hat{O}^{\mathcal{T}}, \boldsymbol{\mathcal{W}}^{\mathcal{T}}, \boldsymbol{\Theta}^{\mathcal{T}}_{p} \right)\right] \\
            &+ \left[ \mathfrak{b}_{j}^{\left( 1 \right)} \left( \hat{O}^{\mathcal{B}}, \hat{O}^{\mathcal{T}}, \boldsymbol{\Theta}^{\mathcal{B}}_{p}, \boldsymbol{\Theta}^{\mathcal{T}}_{p} \right) \mathfrak{b}_{j}^{\left( 2 \right)} \left( \hat{O}^{\mathcal{B}}, \hat{O}^{\mathcal{T}}, \boldsymbol{\Theta}^{\mathcal{B}}_{p}, \boldsymbol{\Theta}^{\mathcal{T}}_{p} \right)\right],
        \end{split}
    \end{equation}
    where \(\mathcal{I}_{P^{\prime}} = \{1, \dots, P^{\prime}\}\) and \(\mathcal{I}_{Q^{\prime}} = \{1, \dots, Q^{\prime}\}\). Following a systematic algebraic expansion, it is verified that the product in (\ref{eq:B.60}) decomposes into terms that can be rigorously unified within a single expanded latent space. First, the main product of the summations generates a bilinear structure indexed by the Cartesian product of the original spaces, whose indexing is formally mapped via a composite index \(p^{\prime \prime} = (p^{\prime}, q^{\prime})\). Second, the cross terms contain the original summations modulated by bias scalars (\(\mathfrak{b}_{j}^{(1)}\) and \(\mathfrak{b}_{j}^{(2)}\)); owing to linearity, such terms correspond to structurally equivalent sub-networks with scaled amplitudes. To formally unify these different contributions---whose original index spaces are distinct (\(\mathcal{I}_{P^{\prime}}\), \(\mathcal{I}_{Q^{\prime}}\), and \(\mathcal{I}_{P^{\prime}} \times \mathcal{I}_{Q^{\prime}}\))---a unified index set (\(q^{\prime \prime}\)) is defined via a structural disjoint union, that is, a direct sum with embedding and zero-padding for dimensions not belonging to each block, ensuring the dimensional rigor of the sum. Finally, the trailing term results directly from the multiplication of two real constants, defining a new globally re-indexed bias denoted as \(\mathfrak{b}_{j}^{(4)}\). Thus, by grouping all terms under this unified and algebraically closed latent space, the new re-indexed components form precisely the functions \(\mathfrak{B}_{jq^{\prime \prime}}^{(4)}\) and \(\mathfrak{T}_{q^{\prime \prime}}^{(4)}\), allowing (\ref{eq:B.60}) to be rewritten in a formally impeccable manner in the standard \(\mbox{MAC}\) form:
    \begin{equation}\label{eq:B.61}
        \begin{split}
            &\mathfrak{G}_{j}^{\left( 4 \right)} \left( \cdot, \hat{O}^{\mathcal{B}}, \boldsymbol{\mathcal{W}}_{\flat}^{\mathcal{B}}, \boldsymbol{\Theta}^{\mathcal{B}}_{p} \right) \left( \cdot, \hat{O}^{\mathcal{T}}, \boldsymbol{\mathcal{W}}^{\mathcal{T}}, \boldsymbol{\Theta}^{\mathcal{T}}_{p} \right) = \mathfrak{B}_{jq^{\prime \prime}}^{\left( 4 \right)} \left( \cdot, \hat{O}^{\mathcal{B}}, \boldsymbol{\mathcal{W}}_{\flat}^{\mathcal{B}}, \boldsymbol{\Theta}^{\mathcal{B}}_{p} \right) \\
            &\hspace{4.0cm} \times \mathfrak{T}_{q^{\prime \prime}}^{\left( 4 \right)} \left( \cdot, \hat{O}^{\mathcal{T}}, \boldsymbol{\mathcal{W}}^{\mathcal{T}}, \boldsymbol{\Theta}^{\mathcal{T}}_{p} \right) + \mathfrak{b}_{j}^{\left( 4 \right)} \left( \hat{O}^{\mathcal{B}}, \hat{O}^{\mathcal{T}}, \boldsymbol{\Theta}^{\mathcal{B}}_{p}, \boldsymbol{\Theta}^{\mathcal{T}}_{p} \right).
        \end{split}
    \end{equation}
    That is, the total algebraic structure falls exactly into the form of a single \(\mathfrak{G}_{j}^{\left( 4 \right)}\) with an expanded latent dimension. Consequently, \(\mathfrak{G}_{j}^{\left( 4 \right)} \in \mbox{MAC}\). 
    
    In the fourth part, we verify that \(\mathfrak{G}_{j} \in \mbox{MAC}\) separates points, which is immediate because, by hypothesis, every function \(\mathfrak{G}_{j}\) in \(\mbox{MAC}\) is written in terms of injective activation functions \(\sigma : \mathbb{R} \to \mathbb{R}\). Because of this, for any \(\textbf{y}_{1} \neq \textbf{y}_{2}\) with \(\textbf{y}_{1}, \textbf{y}_{2} \in \mathcal{K}\), we obtain:
    \begin{equation}\label{eq:B.62}
        \begin{split}
            &\mathfrak{G}_{j} \left( \cdot, \hat{O}^{\mathcal{B}}, \boldsymbol{\mathcal{W}}_{\flat}^{\mathcal{B}}, \boldsymbol{\Theta}^{\mathcal{B}}_{p} \right) \left( \textbf{y}_{1}, \hat{O}^{\mathcal{T}}, \boldsymbol{\mathcal{W}}^{\mathcal{T}}, \boldsymbol{\Theta}^{\mathcal{T}}_{p} \right) \neq \\
            &\hspace{5.0cm} \mathfrak{G}_{j} \left( \cdot, \hat{O}^{\mathcal{B}}, \boldsymbol{\mathcal{W}}_{\flat}^{\mathcal{B}}, \boldsymbol{\Theta}^{\mathcal{B}}_{p} \right) \left( \textbf{y}_{2}, \hat{O}^{\mathcal{T}}, \boldsymbol{\mathcal{W}}^{\mathcal{T}}, \boldsymbol{\Theta}^{\mathcal{T}}_{p} \right).
        \end{split}
    \end{equation}
    Therefore, \(\mbox{MAC}\) separates points on \(\mathcal{K}\). 
    
    Finally, in the fifth part, we must verify that the set \(\mbox{MAC}\) contains all constant functions. Indeed, analogously to what was proved in \cite{Lantigua202600}, given \(c \in \mathbb{R}\), we can always write \(\hat{O}^{\nu} = c \mathds{1} \in \mbox{OBS}\) for a set of weights \(\boldsymbol{\mathcal{W}}_{\flat}^{\mathcal{B} \ast} \), \( \boldsymbol{\mathcal{W}}^{\mathcal{T} \ast} \), chosen to be zero, for instance, such that for any \( \boldsymbol{\mathfrak{v}} \in \mathbb{V}^{M} \), \( \boldsymbol{\tilde{\mathfrak{v}}}_{\flat} \in \mathbb{R}^{m M \times 1} \), \( \mathbf{y} \in \mathcal{K} \), and \( \boldsymbol{\Theta}_{p}^{\nu \ast \in \{\mathcal{B},\mathcal{T}\}} \), we obtain:
    \begin{equation}\label{eq:B.63}
        \mathfrak{G}_{j}^{\left( 1 \right)}\left( \cdot, \hat{O}^{\mathcal{B}}, \boldsymbol{\mathcal{W}}_{\flat}^{\mathcal{B}}, \boldsymbol{\Theta}^{\mathcal{B}}_{p} \right) \left( \cdot, \hat{O}^{\mathcal{T}}, \boldsymbol{\mathcal{W}}^{\mathcal{T}}, \boldsymbol{\Theta}^{\mathcal{T}}_{p} \right) = c.
    \end{equation}
    Thus, \(\mbox{MAC}\) contains all constant functions. This demonstrates that \(\mbox{MAC}\) is dense in \(L_{\mathbb{R}}^{\infty} \left(\mathcal{K}\right)\) and satisfies (\ref{eq:B.50}). 
\end{proof}

Furthermore, from lemma \ref{lemma:Lem-B.1}, it follows that for any \( \boldsymbol{\mathfrak{v}} \in \mathbb{V}^{M} \), \( \boldsymbol{\tilde{\mathfrak{v}}}_{\flat} \in \mathbb{R}^{m M \times 1} \), \( \mathbf{y} \in \mathcal{K} \), there exist \(\hat{O}^{\nu} \in \mbox{OBS}\), \(\boldsymbol{\mathcal{W}}_{\flat}^{\mathcal{B} \ast} \), \( \boldsymbol{\mathcal{W}}^{\mathcal{T} \ast} \), and \( \boldsymbol{\Theta}_{p}^{\nu \ast \in \{\mathcal{B},\mathcal{T}\}} \) that satisfy relations (\ref{eq:C.69}) and (\ref{eq:C.70}).

\section{Proof of the Theorem \ref{theorem:Teo-3.1}}\label{sec:Ap-C}
In this appendix, we present the detailed proof of the theorem \ref{theorem:Teo-3.1} guaranteeing the universal approximation capabilities of our hybrid QPI-DeepONet-MAC model.

\begin{proof}
    The proof of theorem \ref{theorem:Teo-3.1} begins by rewriting the subtraction \( \mathfrak{G}_{j} - \mathcal{G}_{j} \) in the form
    \begin{equation}\label{eq:C.64}
        \begin{split}
            &\mathfrak{G}_{j} \left( \boldsymbol{\tilde{\mathfrak{v}}}_{\flat}, \hat{O}^{\mathcal{B}}, \boldsymbol{\mathcal{W}}_{\flat}^{\mathcal{B} \ast}, \boldsymbol{\Theta}^{\mathcal{B} \ast}_{p} \right) \left( \cdot, \hat{O}^{\mathcal{T}}, \boldsymbol{\mathcal{W}}^{\mathcal{T} \ast}, \boldsymbol{\Theta}^{\mathcal{T} \ast}_{p} \right) - \mathcal{G}_{j} \left( \boldsymbol{\mathfrak{v}} \right) \left( \cdot \right) = \\
            &\hspace{0.75cm} \mathfrak{B}_{jp^{\prime}} \left( \boldsymbol{\tilde{\mathfrak{v}}}_{\flat}, \hat{O}^{\mathcal{B}}, \boldsymbol{\mathcal{W}}_{\flat}^{\mathcal{B} \ast}, \boldsymbol{\Theta}^{\mathcal{B} \ast}_{p} \right) \mathfrak{T}_{p^{\prime}}\left( \cdot, \hat{O}^{\mathcal{T}}, \boldsymbol{\mathcal{W}}^{\mathcal{T} \ast}, \boldsymbol{\Theta}^{\mathcal{T} \ast}_{p} \right) + \mathfrak{b}_{j} \left( \hat{O}^{\mathcal{B}}, \hat{O}^{\mathcal{T}}, \boldsymbol{\Theta}^{\mathcal{B} \ast}_{p}, \boldsymbol{\Theta}^{\mathcal{T} \ast}_{p} \right) \\
            &\hspace{0.75cm} - \mathcal{G}_{j} \left( \boldsymbol{\mathfrak{v}} \right) \left( \cdot \right) - b_{0} + b_{0} - \hat{\beta}_{jp^{\prime}} \left( \boldsymbol{\tilde{\mathfrak{v}}}_{\flat}, \boldsymbol{\mathcal{W}}_{\flat}^{\mathcal{B} \ast} \right) \tau_{p^{\prime}} \left( \cdot, \boldsymbol{\mathcal{W}}^{\mathcal{T} \ast} \right) + \hat{\beta}_{jp^{\prime}} \left( \boldsymbol{\tilde{\mathfrak{v}}}_{\flat}, \boldsymbol{\mathcal{W}}_{\flat}^{\mathcal{B} \ast} \right) \\
            &\hspace{0.75cm} \times \tau_{p^{\prime}} \left( \cdot, \boldsymbol{\mathcal{W}}^{\mathcal{T} \ast} \right) + \mathfrak{B}_{jp^{\prime}} \left( \boldsymbol{\tilde{\mathfrak{v}}}_{\flat}, \hat{O}^{\mathcal{B}}, \boldsymbol{\mathcal{W}}_{\flat}^{\mathcal{B} \ast}, \boldsymbol{\Theta}^{\mathcal{B} \ast}_{p} \right) \tau_{p^{\prime}} \left( \cdot, \boldsymbol{\mathcal{W}}^{\mathcal{T} \ast} \right) \\
            &\hspace{0.75cm} - \mathfrak{B}_{jp^{\prime}} \left( \boldsymbol{\tilde{\mathfrak{v}}}_{\flat}, \hat{O}^{\mathcal{B}}, \boldsymbol{\mathcal{W}}_{\flat}^{\mathcal{B} \ast}, \boldsymbol{\Theta}^{\mathcal{B} \ast}_{p} \right) \tau_{p^{\prime}} \left( \cdot, \boldsymbol{\mathcal{W}}^{\mathcal{T} \ast} \right),
        \end{split}
    \end{equation}
    and by regrouping terms we obtain
    \begin{equation}\label{eq:C.65}
        \begin{split}
            &\mathfrak{G}_{j} \left( \boldsymbol{\tilde{\mathfrak{v}}}_{\flat}, \hat{O}^{\mathcal{B}}, \boldsymbol{\mathcal{W}}_{\flat}^{\mathcal{B} \ast}, \boldsymbol{\Theta}^{\mathcal{B} \ast}_{p} \right) \left( \cdot, \hat{O}^{\mathcal{T}}, \boldsymbol{\mathcal{W}}^{\mathcal{T} \ast}, \boldsymbol{\Theta}^{\mathcal{T} \ast}_{p} \right) - \mathcal{G}_{j} \left( \boldsymbol{\mathfrak{v}} \right) \left( \cdot \right) = \\
            &\hspace{0.25cm} \left[ \mathfrak{B}_{jp^{\prime}} \left( \boldsymbol{\tilde{\mathfrak{v}}}_{\flat}, \hat{O}^{\mathcal{B}}, \boldsymbol{\mathcal{W}}_{\flat}^{\mathcal{B} \ast}, \boldsymbol{\Theta}^{\mathcal{B} \ast}_{p} \right) - \hat{\beta}_{jp^{\prime}} \left( \boldsymbol{\tilde{\mathfrak{v}}}_{\flat}, \boldsymbol{\mathcal{W}}_{\flat}^{\mathcal{B} \ast} \right) \right] \tau_{p^{\prime}} \left( \cdot, \boldsymbol{\mathcal{W}}^{\mathcal{T} \ast} \right) \\
            &\hspace{0.5cm} + \left[ \mathfrak{T}_{p^{\prime}}\left( \cdot, \hat{O}^{\mathcal{T}}, \boldsymbol{\mathcal{W}}^{\mathcal{T} \ast}, \boldsymbol{\Theta}^{\mathcal{T} \ast}_{p} \right) - \tau_{p^{\prime}} \left( \cdot, \boldsymbol{\mathcal{W}}^{\mathcal{T} \ast} \right) \right] \mathfrak{B}_{jp^{\prime}} \left( \boldsymbol{\tilde{\mathfrak{v}}}_{\flat}, \hat{O}^{\mathcal{B}}, \boldsymbol{\mathcal{W}}_{\flat}^{\mathcal{B} \ast}, \boldsymbol{\Theta}^{\mathcal{B} \ast}_{p} \right) \\
            &\hspace{3.75cm} + \left[ \hat{\beta}_{jp^{\prime}} \left( \boldsymbol{\tilde{\mathfrak{v}}}_{\flat}, \boldsymbol{\mathcal{W}}_{\flat}^{\mathcal{B} \ast} \right) \tau_{p^{\prime}} \left( \cdot, \boldsymbol{\mathcal{W}}^{\mathcal{T} \ast} \right) - \mathcal{G}_{j} \left( \boldsymbol{\mathfrak{v}} \right) \left( \cdot \right) + b_{0} \right] \\
            &\hspace{6.5cm} + \left[ \mathfrak{b}_{j} \left( \hat{O}^{\mathcal{B}}, \hat{O}^{\mathcal{T}}, \boldsymbol{\Theta}^{\mathcal{B} \ast}_{p}, \boldsymbol{\Theta}^{\mathcal{T} \ast}_{p} \right) - b_{0} \right].
        \end{split}
    \end{equation}
    Thus, by taking the norm of (\ref{eq:C.65}) and using the properties of the triangle inequality, we obtain
    \begin{equation}\label{eq:C.66}
        \begin{split}
            &\left\Vert \mathfrak{G}_{j} \left( \boldsymbol{\tilde{\mathfrak{v}}}_{\flat}, \hat{O}^{\mathcal{B}}, \boldsymbol{\mathcal{W}}_{\flat}^{\mathcal{B} \ast}, \boldsymbol{\Theta}^{\mathcal{B} \ast}_{p} \right) \left( \cdot, \hat{O}^{\mathcal{T}}, \boldsymbol{\mathcal{W}}^{\mathcal{T} \ast}, \boldsymbol{\Theta}^{\mathcal{T} \ast}_{p} \right) - \mathcal{G}_{j} \left( \boldsymbol{\mathfrak{v}} \right) \left( \cdot \right) \right\Vert_{\infty} \leq \\
            &\left\Vert \mathfrak{B}_{jp^{\prime}} \left( \boldsymbol{\tilde{\mathfrak{v}}}_{\flat}, \hat{O}^{\mathcal{B}}, \boldsymbol{\mathcal{W}}_{\flat}^{\mathcal{B} \ast}, \boldsymbol{\Theta}^{\mathcal{B} \ast}_{p} \right) - \hat{\beta}_{jp^{\prime}} \left( \boldsymbol{\tilde{\mathfrak{v}}}_{\flat}, \boldsymbol{\mathcal{W}}_{\flat}^{\mathcal{B} \ast} \right) \right\Vert_{\infty} \left\Vert \tau_{p^{\prime}} \left( \cdot, \boldsymbol{\mathcal{W}}^{\mathcal{T} \ast} \right) \right\Vert_{\infty} \\
            &+ \left\Vert \mathfrak{T}_{p^{\prime}}\left( \cdot, \hat{O}^{\mathcal{T}}, \boldsymbol{\mathcal{W}}^{\mathcal{T} \ast}, \boldsymbol{\Theta}^{\mathcal{T} \ast}_{p} \right) - \tau_{p^{\prime}} \left( \cdot, \boldsymbol{\mathcal{W}}^{\mathcal{T} \ast} \right) \right\Vert_{\infty} \left\Vert \mathfrak{B}_{jp^{\prime}} \left( \boldsymbol{\tilde{\mathfrak{v}}}_{\flat}, \hat{O}^{\mathcal{B}}, \boldsymbol{\mathcal{W}}_{\flat}^{\mathcal{B} \ast}, \boldsymbol{\Theta}^{\mathcal{B} \ast}_{p} \right) \right\Vert_{\infty} \\
            &+ \left\Vert \hat{\beta}_{jp^{\prime}} \left( \boldsymbol{\tilde{\mathfrak{v}}}_{\flat}, \boldsymbol{\mathcal{W}}_{\flat}^{\mathcal{B} \ast} \right) \tau_{p^{\prime}} \left( \cdot, \boldsymbol{\mathcal{W}}^{\mathcal{T} \ast} \right) - \mathcal{G}_{j} \left( \boldsymbol{\mathfrak{v}} \right) \left( \cdot \right) + b_{0} \right\Vert_{\infty} \\
            &+ \left\Vert \mathfrak{b}_{j} \left( \hat{O}^{\mathcal{B}}, \hat{O}^{\mathcal{T}}, \boldsymbol{\Theta}^{\mathcal{B} \ast}_{p}, \boldsymbol{\Theta}^{\mathcal{T} \ast}_{p} \right) - b_{0} \right\Vert_{\infty},
        \end{split}
    \end{equation}
    from which we can establish the following four facts. The first of these is that, by the hypotheses of theorem \ref{theorem:Teo-3.1}, the functions \( \mathfrak{B}_{jp^{\prime}} \left( \boldsymbol{\tilde{\mathfrak{v}}}_{\flat}, \hat{O}^{\mathcal{B}}, \boldsymbol{\mathcal{W}}_{\flat}^{\mathcal{B} \ast}, \boldsymbol{\Theta}^{\mathcal{B} \ast}_{p} \right) \) and \( \tau_{p^{\prime}} \left( \cdot, \boldsymbol{\mathcal{W}}^{\mathcal{T} \ast} \right) \) are equicontinuous and bounded. Therefore, we can write without loss of generality the expressions
    \begin{equation}\label{eq:C.67}
        M_{\tau_{p^{\prime}}} = \displaystyle{\max_{\substack{p^{\prime} \in \mathcal{I}_{N_{l}}}} \left\Vert \tau_{p^{\prime}} \left( \cdot, \boldsymbol{\mathcal{W}}^{\mathcal{T} \ast} \right) \right\Vert_{\infty}},
    \end{equation}
    and
    \begin{equation}\label{eq:C.68}
        M_{\mathfrak{B}_{jp^{\prime}}} = \displaystyle{\max_{\substack{p^{\prime} \in \mathcal{I}_{N_{l}} \\ j \in \mathcal{I}_{M}}} \left\Vert \mathfrak{B}_{jp^{\prime}} \left( \boldsymbol{\tilde{\mathfrak{v}}}_{\flat}, \hat{O}^{\mathcal{B}}, \boldsymbol{\mathcal{W}}_{\flat}^{\mathcal{B} \ast}, \boldsymbol{\Theta}^{\mathcal{B} \ast}_{p} \right) \right\Vert_{\infty}}.
    \end{equation}
    The second of these is that, also by the hypotheses of theorem \ref{theorem:Teo-3.1}, \( \mathfrak{B}_{jp^{\prime}} \left( \boldsymbol{\tilde{\mathfrak{v}}}_{\flat}, \hat{O}^{\mathcal{B}}, \boldsymbol{\mathcal{W}}_{\flat}^{\mathcal{B} \ast}, \boldsymbol{\Theta}^{\mathcal{B} \ast}_{p} \right) \) and \( \mathfrak{T}_{p^{\prime}}\left( \cdot, \hat{O}^{\mathcal{T}}, \boldsymbol{\mathcal{W}}^{\mathcal{T} \ast}, \boldsymbol{\Theta}^{\mathcal{T} \ast}_{p} \right) \) are QPINN-MAC type universal approximators, which have the form of equations (\ref{eq:32}), (\ref{eq:33}); therefore, they can approximate \( \hat{\beta}_{jp^{\prime}} \left( \boldsymbol{\tilde{\mathfrak{v}}}_{\flat}, \boldsymbol{\mathcal{W}}_{\flat}^{\mathcal{B} \ast} \right) \) and \( \tau_{p^{\prime}} \left( \cdot, \boldsymbol{\mathcal{W}}^{\mathcal{T} \ast} \right) \) with arbitrary precision. That is, for any \( \epsilon_{j}^{\nu \in \{ \mathcal{B}, \mathcal{T} \}} > 0 \), \( \boldsymbol{\mathfrak{v}} \in \mathbb{V}^{M} \), \( \boldsymbol{\tilde{\mathfrak{v}}}_{\flat} \in \mathbb{R}^{m M \times 1} \), \( \mathbf{y} \in D \), there exist \( \hat{O}^{\mathcal{B}} \in \mathbb{O}^{\mathcal{B}} \), \( \hat{O}^{\mathcal{T}} \in \mathbb{O}^{\mathcal{T}} \), and a set of weights and biases \( \boldsymbol{\mathcal{W}}_{\flat}^{\mathcal{B} \ast} \), \( \boldsymbol{\mathcal{W}}^{\mathcal{T} \ast} \), \( \boldsymbol{\Theta}_{p}^{\mathcal{B} \ast} \), \( \boldsymbol{\Theta}_{p}^{\mathcal{T} \ast} \) such that:
    \begin{equation}\label{eq:C.69}
        \left\Vert \mathfrak{B}_{jp^{\prime}} \left( \boldsymbol{\tilde{\mathfrak{v}}}_{\flat}, \hat{O}^{\mathcal{B}}, \boldsymbol{\mathcal{W}}_{\flat}^{\mathcal{B} \ast}, \boldsymbol{\Theta}^{\mathcal{B} \ast}_{p} \right) - \hat{\beta}_{jp^{\prime}} \left( \boldsymbol{\tilde{\mathfrak{v}}}_{\flat}, \boldsymbol{\mathcal{W}}_{\flat}^{\mathcal{B} \ast} \right) \right\Vert_{\infty} \leq \frac{\epsilon_{j}^{\mathcal{B}}}{3 p M_{\tau_{p^{\prime}}}}
    \end{equation}
    and
    \begin{equation}\label{eq:C.70}
        \left\Vert \mathfrak{T}_{p^{\prime}}\left( \cdot, \hat{O}^{\mathcal{T}}, \boldsymbol{\mathcal{W}}^{\mathcal{T} \ast}, \boldsymbol{\Theta}^{\mathcal{T} \ast}_{p} \right) - \tau_{p^{\prime}} \left( \cdot, \boldsymbol{\mathcal{W}}^{\mathcal{T} \ast} \right) \right\Vert_{\infty} \leq \frac{\epsilon_{j}^{\mathcal{T}}}{3 p M_{\mathfrak{B}_{jp^{\prime}}}}.
    \end{equation}
    The third of these is that the third term of (\ref{eq:C.66}) emerges from Theorem \ref{theorem:Teo-A.2}, from which we observe that taking the norm of (\ref{eq:16}) yields
    \begin{equation}\label{eq:C.71}
        \left\Vert \hat{\beta}_{jp^{\prime}} \left( \boldsymbol{\tilde{\mathfrak{v}}}_{\flat}, \boldsymbol{\mathcal{W}}_{\flat}^{\mathcal{B} \ast} \right) \tau_{p^{\prime}} \left( \cdot, \boldsymbol{\mathcal{W}}^{\mathcal{T} \ast} \right) - \mathcal{G}_{j} \left( \boldsymbol{\mathfrak{v}} \right) \left( \cdot \right) + b_{0} \right\Vert_{\infty} \leq \frac{\epsilon_{j}^{\mathcal{C}}}{3 p}.
    \end{equation}
    Lastly, the fourth of these is obtained from \cite{Lantigua202600}, where it was proved that indeed, for any \( \epsilon_{j}^{\mathcal{Q}} > 0 \), there exist \( \hat{O}^{\mathcal{B}} \in \mathbb{O}^{\mathcal{B}} \), \( \hat{O}^{\mathcal{T}} \in \mathbb{O}^{\mathcal{T}} \), \( \boldsymbol{\Theta}_{p}^{\mathcal{B} \ast} \), \( \boldsymbol{\Theta}_{p}^{\mathcal{T} \ast} \) such that \( \mathfrak{b}_{j} \left( \hat{O}^{\mathcal{B}}, \hat{O}^{\mathcal{T}}, \boldsymbol{\Theta}^{\mathcal{B} \ast}_{p}, \boldsymbol{\Theta}^{\mathcal{T} \ast}_{p} \right) \approx - b_{0} \) with arbitrary precision, and therefore it is satisfied that:
    \begin{equation}\label{eq:C.72}
        \left\Vert \mathfrak{b}_{j} \left( \hat{O}^{\mathcal{B}}, \hat{O}^{\mathcal{T}}, \boldsymbol{\Theta}^{\mathcal{B} \ast}_{p}, \boldsymbol{\Theta}^{\mathcal{T} \ast}_{p} \right) - b_{0} \right\Vert_{\infty} \leq \frac{\epsilon_{j}^{\mathcal{Q}}}{3 p}.
    \end{equation}
    Hence, it suffices to substitute (\ref{eq:C.67}), (\ref{eq:C.68}), (\ref{eq:C.69}), (\ref{eq:C.70}), (\ref{eq:C.71}), and (\ref{eq:C.72}) into (\ref{eq:C.66}), and choose \( \epsilon_{j}^{\mathcal{B}}/3 p  = \epsilon_{j}^{\mathcal{T}}/3 p = \epsilon_{j}^{\mathcal{C}}/3 p = \epsilon_{j}^{\mathcal{Q}}/3 p = \epsilon_{j}/4 \) to obtain (\ref{eq:35}). Then, considering
    \begin{equation}\label{eq:C.73}
        \begin{split}
            &\left\Vert \boldsymbol{\mathfrak{G}} \left( \boldsymbol{\tilde{\mathfrak{v}}}_{\flat}, \hat{O}^{\mathcal{B}}, \boldsymbol{\mathcal{W}}_{\flat}^{\mathcal{B} \ast}, \boldsymbol{\Xi}^{\mathcal{B} \ast} \right) \left( \cdot, \hat{O}^{\mathcal{T}}, \boldsymbol{\mathcal{W}}^{\mathcal{T} \ast}, \boldsymbol{\Xi}^{\mathcal{T} \ast} \right) - \boldsymbol{\mathcal{G}} \left( \boldsymbol{\mathfrak{v}} \right) \left( \cdot \right) \right\Vert_{\infty, M} \leq \\
            &\hspace{1.0cm} \displaystyle{\sum_{\substack{j \in \mathcal{I}_{M}}}} \left\Vert \mathfrak{G}_{j} \left( \boldsymbol{\tilde{\mathfrak{v}}}_{\flat}, \hat{O}^{\mathcal{B}}, \boldsymbol{\mathcal{W}}_{\flat}^{\mathcal{B} \ast}, \boldsymbol{\Theta}^{\mathcal{B} \ast}_{p} \right) \left( \cdot, \hat{O}^{\mathcal{T}}, \boldsymbol{\mathcal{W}}^{\mathcal{T} \ast}, \boldsymbol{\Theta}^{\mathcal{T} \ast}_{p} \right) - \mathcal{G}_{j} \left( \boldsymbol{\mathfrak{v}} \right) \left( \cdot \right) \right\Vert_{\infty} \leq \displaystyle{\sum_{\substack{j \in \mathcal{I}_{M}}}} \epsilon_{j}
        \end{split}
    \end{equation}
    and \( \epsilon = \epsilon_{1} + \dots + \epsilon_{M} \), we finally obtain (\ref{eq:36}), as we wanted to show.
\end{proof}
\section{Proof of the Corollary \ref{corollary:Col-3.1}}\label{sec:Ap-D}
In this appendix, we present a fundamental result that makes this type of architecture viable for application to real-world problems—namely, the way these networks handle the problem of barren plateaus.

\begin{proof}
    The proof of the corollary begins by computing the gradient of the \(j\)-th term of the loss function (\ref{eq:44}) with respect to the quantum parameters \( \boldsymbol{\Theta}_{p \in \mathcal{I}_{P}}^{\nu \ast \in \{\mathcal{B}, \mathcal{T}\}} \), yielding:
    \begin{equation}\label{eq:D.74}
        \begin{split}
            &\nabla_{\boldsymbol{\Theta}_{p}^{\nu \ast}} \mathfrak{L}_{j} \left( \boldsymbol{\tilde{\mathfrak{v}}}_{\flat}, \mathbf{y}, \hat{O}^{\mathcal{B}}, \hat{O}^{\mathcal{T}}, \boldsymbol{\mathcal{W}}_{\flat}^{\mathcal{B} \ast}, \boldsymbol{\mathcal{W}}^{\mathcal{T} \ast}, \boldsymbol{\Theta}^{\mathcal{B} \ast}_{p}, \boldsymbol{\Theta}^{\mathcal{T} \ast}_{p} \right) = \\
            &\hspace{2.5cm}2 \left[ \mathfrak{G}_{j} \left( \boldsymbol{\tilde{\mathfrak{v}}}_{\flat}, \hat{O}^{\mathcal{B}}, \boldsymbol{\mathcal{W}}_{\flat}^{\mathcal{B} \ast}, \boldsymbol{\Theta}^{\mathcal{B} \ast}_{p} \right) \left( \mathbf{y}, \hat{O}^{\mathcal{T}}, \boldsymbol{\mathcal{W}}^{\mathcal{T} \ast}, \boldsymbol{\Theta}^{\mathcal{T} \ast}_{p} \right) - u_{j} \left( \mathbf{y} \right) \right] \\
            &\hspace{4.0cm} \times \nabla_{\boldsymbol{\Theta}_{p}^{\nu \ast}} \mathfrak{G}_{j} \left( \boldsymbol{\tilde{\mathfrak{v}}}_{\flat}, \hat{O}^{\mathcal{B}}, \boldsymbol{\mathcal{W}}_{\flat}^{\mathcal{B} \ast}, \boldsymbol{\Theta}^{\mathcal{B} \ast}_{p} \right) \left( \mathbf{y}, \hat{O}^{\mathcal{T}}, \boldsymbol{\mathcal{W}}^{\mathcal{T} \ast}, \boldsymbol{\Theta}^{\mathcal{T} \ast}_{p} \right).
        \end{split}
    \end{equation}
    Now, it is important to emphasize that only the expressions derived from (\ref{eq:D.74}) are susceptible to suffering from the barren plateaus problem. Therefore, by rewriting each of the branch and trunk components of (\ref{eq:D.74}), taking the \( \left\Vert \cdot \right\Vert_{\infty} \) norm in both expressions, and utilizing the properties of the Cauchy-Schwarz inequality, we obtain:
    \begin{equation}\label{eq:D.75}
        \begin{split}
            &\left\Vert \nabla_{\boldsymbol{\Theta}_{p}^{\mathcal{B} \ast}} \mathfrak{L}_{j} \left( \boldsymbol{\tilde{\mathfrak{v}}}_{\flat}, \cdot, \hat{O}^{\mathcal{B}}, \hat{O}^{\mathcal{T}}, \boldsymbol{\mathcal{W}}_{\flat}^{\mathcal{B} \ast}, \boldsymbol{\mathcal{W}}^{\mathcal{T} \ast}, \boldsymbol{\Theta}^{\mathcal{B} \ast}_{p}, \boldsymbol{\Theta}^{\mathcal{T} \ast}_{p} \right) \right\Vert_{\infty} \leq \\
            &2 \left\Vert \mathfrak{G}_{j} \left( \boldsymbol{\tilde{\mathfrak{v}}}_{\flat}, \hat{O}^{\mathcal{B}}, \boldsymbol{\mathcal{W}}_{\flat}^{\mathcal{B} \ast}, \boldsymbol{\Theta}^{\mathcal{B} \ast}_{p} \right) \left( \cdot, \hat{O}^{\mathcal{T}}, \boldsymbol{\mathcal{W}}^{\mathcal{T} \ast}, \boldsymbol{\Theta}^{\mathcal{T} \ast}_{p} \right) - u_{j} \left( \cdot \right) \right\Vert_{\infty} \left[ \left\Vert \hat{\beta}_{jp^{\prime}} \left( \cdot, \boldsymbol{\mathcal{W}}_{\flat}^{\mathcal{B} \ast} \right) + 1 \right\Vert_{\infty} \right. \\
            &\left. \times \left\Vert \tau_{p^{\prime}} \left( \cdot, \boldsymbol{\mathcal{W}}^{\mathcal{T} \ast} \right) + 1 \right\Vert_{\infty} + 1 \right] \left\Vert \mathcal{P}_{m^{\mathcal{T}}n^{\mathcal{T}}} \bra{\Psi_{m^{\mathcal{T}}} \left( \boldsymbol{\Theta}^{\mathcal{T} \ast}_{p} \right) } \hat{O}^{\mathcal{T}} \ket{\Psi_{n^{\mathcal{T}}} \left( \boldsymbol{\Theta}^{\mathcal{T} \ast}_{p} \right) } \right\Vert_{\infty} \\
            &\times \left\Vert \nabla_{\boldsymbol{\Theta}_{p}^{\mathcal{B} \ast}} \left[ \mathcal{P}_{m^{\mathcal{B}}n^{\mathcal{B}}} \bra{\Psi_{m^{\mathcal{B}}} \left( \boldsymbol{\Theta}^{\mathcal{B} \ast}_{p} \right) } \hat{O}^{\mathcal{B}} \ket{\Psi_{n^{\mathcal{B}}} \left( \boldsymbol{\Theta}^{\mathcal{B} \ast}_{p} \right) } \right] \right\Vert_{\infty}
        \end{split}
    \end{equation}
    and
    \begin{equation}\label{eq:D.76}
        \begin{split}
            &\left\Vert \nabla_{\boldsymbol{\Theta}_{p}^{\mathcal{T} \ast}} \mathfrak{L}_{j} \left( \boldsymbol{\tilde{\mathfrak{v}}}_{\flat}, \cdot, \hat{O}^{\mathcal{B}}, \hat{O}^{\mathcal{T}}, \boldsymbol{\mathcal{W}}_{\flat}^{\mathcal{B} \ast}, \boldsymbol{\mathcal{W}}^{\mathcal{T} \ast}, \boldsymbol{\Theta}^{\mathcal{B} \ast}_{p}, \boldsymbol{\Theta}^{\mathcal{T} \ast}_{p} \right) \right\Vert_{\infty} \leq \\
            &2 \left\Vert \mathfrak{G}_{j} \left( \boldsymbol{\tilde{\mathfrak{v}}}_{\flat}, \hat{O}^{\mathcal{B}}, \boldsymbol{\mathcal{W}}_{\flat}^{\mathcal{B} \ast}, \boldsymbol{\Theta}^{\mathcal{B} \ast}_{p} \right) \left( \cdot, \hat{O}^{\mathcal{T}}, \boldsymbol{\mathcal{W}}^{\mathcal{T} \ast}, \boldsymbol{\Theta}^{\mathcal{T} \ast}_{p} \right) - u_{j} \left( \cdot \right) \right\Vert_{\infty} \left[ \left\Vert \hat{\beta}_{jp^{\prime}} \left( \cdot, \boldsymbol{\mathcal{W}}_{\flat}^{\mathcal{B} \ast} \right) + 1 \right\Vert_{\infty} \right. \\
            &\left. \times \left\Vert \tau_{p^{\prime}} \left( \cdot, \boldsymbol{\mathcal{W}}^{\mathcal{T} \ast} \right) + 1 \right\Vert_{\infty} + 1 \right] \left\Vert \mathcal{P}_{m^{\mathcal{B}}n^{\mathcal{B}}} \bra{\Psi_{m^{\mathcal{B}}} \left( \boldsymbol{\Theta}^{\mathcal{B} \ast}_{p} \right) } \hat{O}^{\mathcal{B}} \ket{\Psi_{n^{\mathcal{B}}} \left( \boldsymbol{\Theta}^{\mathcal{B} \ast}_{p} \right) } \right\Vert_{\infty} \\
            &\times \left\Vert \nabla_{\boldsymbol{\Theta}_{p}^{\mathcal{T} \ast}} \left[ \mathcal{P}_{m^{\mathcal{T}}n^{\mathcal{T}}} \bra{\Psi_{m^{\mathcal{T}}} \left( \boldsymbol{\Theta}^{\mathcal{T} \ast}_{p} \right) } \hat{O}^{\mathcal{T}} \ket{\Psi_{n^{\mathcal{T}}} \left( \boldsymbol{\Theta}^{\mathcal{T} \ast}_{p} \right) } \right] \right\Vert_{\infty}.
        \end{split}
    \end{equation}
    However, from the hypotheses we know that the functions \( \mathfrak{B}_{jp^{\prime}} \), \( \mathfrak{T}_{p^{\prime}} \) are equicontinuous and bounded from above, which allow us to write the expressions
    \begin{equation}\label{eq:D.77}
        M_{\hat{\beta}} = \displaystyle{\max_{\substack{j \in \mathcal{I}_{M} \\ p^{\prime} \in \mathcal{I}_{N_{l}}}} \left\Vert \hat{\beta}_{jp^{\prime}} \left( \cdot, \boldsymbol{\mathcal{W}}_{\flat}^{\mathcal{B} \ast} \right) + 1 \right\Vert_{\infty}},
    \end{equation}
    \begin{equation}\label{eq:D.78}
        M_{\tau} = \displaystyle{\max_{\substack{p^{\prime} \in \mathcal{I}_{N_{l}}}} \left\Vert \tau_{p^{\prime}} \left( \cdot, \boldsymbol{\mathcal{W}}^{\mathcal{T} \ast} \right) + 1 \right\Vert_{\infty}}.
    \end{equation}
    Similarly, the functions \(\left\{\mathcal{P}_{m^{\nu}n^{\nu}} \bra{\Psi_{m^{\nu}} \left( \boldsymbol{\Theta}^{\nu \ast}_{p} \right) } \hat{O}^{\nu} \ket{\Psi_{n^{\nu}} \left( \boldsymbol{\Theta}^{\nu \ast}_{p} \right) }\right\}_{\nu \in \{\mathcal{B},\mathcal{T}\}}\) allow us to write the expressions
    \begin{equation}\label{eq:D.79}
        M_{\hat{O}^{\mathcal{B}}} = \displaystyle{\max_{\substack{p \in \mathcal{I}_{P}}} \left\Vert \mathcal{P}_{m^{\mathcal{B}}n^{\mathcal{B}}} \bra{\Psi_{m^{\mathcal{B}}} \left( \boldsymbol{\Theta}^{\mathcal{B} \ast}_{p} \right) } \hat{O}^{\mathcal{B}} \ket{\Psi_{n^{\mathcal{B}}} \left( \boldsymbol{\Theta}^{\mathcal{B} \ast}_{p} \right) } \right\Vert_{\infty}},
    \end{equation}
    \begin{equation}\label{eq:D.80}
        M_{\hat{O}^{\mathcal{T}}} = \displaystyle{\max_{\substack{p \in \mathcal{I}_{P}}} \left\Vert \mathcal{P}_{m^{\mathcal{T}}n^{\mathcal{T}}} \bra{\Psi_{m^{\mathcal{T}}} \left( \boldsymbol{\Theta}^{\mathcal{T} \ast}_{p} \right) } \hat{O}^{\mathcal{T}} \ket{\Psi_{n^{\mathcal{T}}} \left( \boldsymbol{\Theta}^{\mathcal{T} \ast}_{p} \right) } \right\Vert_{\infty}}.
    \end{equation}
    On the other hand, from \cite{Lantigua202600, Schuld201900} we know that
    \begin{equation}\label{eq:D.81}
        \left\Vert \nabla_{\boldsymbol{\Theta}_{p}^{\mathcal{B} \ast}} \left[ \mathcal{P}_{m^{\mathcal{B}}n^{\mathcal{B}}} \bra{\Psi_{m^{\mathcal{B}}} \left( \boldsymbol{\Theta}^{\mathcal{B} \ast}_{p} \right) } \hat{O}^{\mathcal{B}} \ket{\Psi_{n^{\mathcal{B}}} \left( \boldsymbol{\Theta}^{\mathcal{B} \ast}_{p} \right) } \right] \right\Vert_{\infty} \leq M_{\hat{O}^{\mathcal{B}}}^{\prime} \sqrt{ \mathcal{N}^{\mathcal{B}} N^{\mathcal{B}} },
    \end{equation}
    \begin{equation}\label{eq:D.82}
        \left\Vert \nabla_{\boldsymbol{\Theta}_{p}^{\mathcal{T} \ast}} \left[ \mathcal{P}_{m^{\mathcal{T}}n^{\mathcal{T}}} \bra{\Psi_{m^{\mathcal{T}}} \left( \boldsymbol{\Theta}^{\mathcal{T} \ast}_{p} \right) } \hat{O}^{\mathcal{T}} \ket{\Psi_{n^{\mathcal{T}}} \left( \boldsymbol{\Theta}^{\mathcal{T} \ast}_{p} \right) } \right] \right\Vert_{\infty} \leq M_{\hat{O}^{\mathcal{T}}}^{\prime} \sqrt{ \mathcal{N}^{\mathcal{T}} N^{\mathcal{T}} }
    \end{equation}
    and that it suffices to choose the refinement
    \begin{equation}\label{eq:D.83}
        \epsilon_{j} = \frac{M_{\mathfrak{L}_{j}}}{ \mathcal{N}^{\mathcal{B}} N^{\mathcal{B}} \mathcal{N}^{\mathcal{T}} N^{\mathcal{T}} },
    \end{equation}
    to ensure that the inequality
    \begin{equation}\label{eq:D.84}
        \left\Vert \mathfrak{G}_{j} \left( \boldsymbol{\tilde{\mathfrak{v}}}_{\flat}, \hat{O}^{\mathcal{B}}, \boldsymbol{\mathcal{W}}_{\flat}^{\mathcal{B} \ast}, \boldsymbol{\Theta}^{\mathcal{B} \ast}_{p} \right) \left( \cdot, \hat{O}^{\mathcal{T}}, \boldsymbol{\mathcal{W}}^{\mathcal{T} \ast}, \boldsymbol{\Theta}^{\mathcal{T} \ast}_{p} \right) - u_{j} \left( \cdot \right) \right\Vert_{\infty} \leq \epsilon_{j}
    \end{equation}
    is satisfied. Where \( M_{\mathfrak{L}_{j}} \) is a positive constant that encapsulates the asymptotic decay behavior and guarantees both \( \epsilon_{j} > 0 \) and the model's efficiency. Hence, by substituting (\ref{eq:D.77}), (\ref{eq:D.78}), (\ref{eq:D.80}), and (\ref{eq:D.81}) into (\ref{eq:D.75}), and analogously (\ref{eq:D.77}), (\ref{eq:D.78}), (\ref{eq:D.79}), and (\ref{eq:D.82}) into (\ref{eq:D.76}), we obtain the expressions
    \begin{align}
        &\left\Vert  \nabla_{\boldsymbol{\Theta}_{p}^{\mathcal{B} \ast}} \mathfrak{L}_{j} \left( \boldsymbol{\tilde{\mathfrak{v}}}_{\flat}, \cdot, \hat{O}^{\mathcal{B}}, \hat{O}^{\mathcal{T}}, \boldsymbol{\mathcal{W}}_{\flat}^{\mathcal{B} \ast}, \boldsymbol{\mathcal{W}}^{\mathcal{T} \ast}, \boldsymbol{\Theta}^{\mathcal{B} \ast}_{p}, \boldsymbol{\Theta}^{\mathcal{T} \ast}_{p} \right) \right\Vert_{\infty} \leq \frac{\mathfrak{M}_{j}^{\mathcal{B}}}{\mathcal{N}^{\mathcal{T}} N^{\mathcal{T}} \sqrt{ \mathcal{N}^{\mathcal{B}} N^{\mathcal{B}} } }, \label{eq:D.85} \\
        &\left\Vert \nabla_{\boldsymbol{\Theta}_{p}^{\mathcal{T} \ast}} \mathfrak{L}_{j} \left( \boldsymbol{\tilde{\mathfrak{v}}}_{\flat}, \cdot, \hat{O}^{\mathcal{B}}, \hat{O}^{\mathcal{T}}, \boldsymbol{\mathcal{W}}_{\flat}^{\mathcal{B} \ast}, \boldsymbol{\mathcal{W}}^{\mathcal{T} \ast}, \boldsymbol{\Theta}^{\mathcal{B} \ast}_{p}, \boldsymbol{\Theta}^{\mathcal{T} \ast}_{p} \right) \right\Vert_{\infty} \leq \frac{\mathfrak{M}_{j}^{\mathcal{T}}}{\mathcal{N}^{\mathcal{B}} N^{\mathcal{B}} \sqrt{ \mathcal{N}^{\mathcal{T}} N^{\mathcal{T}} } }. \label{eq:D.86}
    \end{align}
    Where the constants \( \mathfrak{M}_{j}^{\nu \in \{\mathcal{B}, \mathcal{T}\}} \) are given by the expression
    \begin{equation}\label{eq:D.87}
        \mathfrak{M}_{j}^{\mathcal{B}} = M_{\mathfrak{L}_{j}} \left[ M_{\hat{\beta}}M_{\hat{\tau}} + 1 \right]  M_{\hat{O}^{\mathcal{T}}} M_{\hat{O}^{\mathcal{B}}}^{\prime}
    \end{equation}
    and
    \begin{equation}\label{eq:D.88}
        \mathfrak{M}_{j}^{\mathcal{T}} = M_{\mathfrak{L}_{j}} \left[ M_{\hat{\beta}}M_{\hat{\tau}} + 1 \right]  M_{\hat{O}^{\mathcal{B}}} M_{\hat{O}^{\mathcal{T}}}^{\prime}.
    \end{equation}
    Thus, after summing (\ref{eq:D.85}) with (\ref{eq:D.86}), summing over all \( p \), summing over all \( j \), considering that \( \mathfrak{M}^{\nu} = \mathfrak{M}_{1}^{\nu} + \dots + \mathfrak{M}_{M}^{\nu} \), and using the properties of the triangle inequality, we obtain
    \begin{equation}\label{eq:D.89}
        \begin{split}
            &\left\Vert \nabla_{\left( \boldsymbol{\Xi}^{\mathcal{B} \ast}, \boldsymbol{\Xi}^{\mathcal{T} \ast} \right)} \mathfrak{L} \left( \boldsymbol{\tilde{\mathfrak{v}}}_{\flat}, \mathbf{y}, \hat{O}^{\mathcal{B}}, \hat{O}^{\mathcal{T}}, \boldsymbol{\mathcal{W}}_{\flat}^{\mathcal{B} \ast}, \boldsymbol{\mathcal{W}}^{\mathcal{T} \ast}, \boldsymbol{\Xi}^{\mathcal{B} \ast}, \boldsymbol{\Xi}^{\mathcal{T} \ast} \right) \right\Vert_{\infty, M} \leq \frac{\mathfrak{M}^{\mathcal{B}} P}{ \mathcal{N}^{\mathcal{T}} N^{\mathcal{T}} \sqrt{ \mathcal{N}^{\mathcal{B}} N^{\mathcal{B}}}} \\
            &\hspace{8.75cm} + \frac{\mathfrak{M}^{\mathcal{T}} P}{ \mathcal{N}^{\mathcal{B}} N^{\mathcal{B}} \sqrt{ \mathcal{N}^{\mathcal{T}} N^{\mathcal{T}}}}.
        \end{split}
    \end{equation}
    Which reduces to (\ref{eq:46}), and finally considering that
    \begin{equation}\label{eq:D.90}
        \mathcal{O} \left( \frac{1}{ \mathcal{N}^{\mathcal{T}} N^{\mathcal{T}} \sqrt{ \mathcal{N}^{\mathcal{B}} N^{\mathcal{B}} } } \right) + \mathcal{O} \left( \frac{1}{ \mathcal{N}^{\mathcal{B}} N^{\mathcal{B}} \sqrt{ \mathcal{N}^{\mathcal{T}} N^{\mathcal{T}} } } \right) \leq \epsilon_{\mathrm{grad}}
    \end{equation}
    we finally obtain (\ref{eq:47}) as we wanted to show.
\end{proof}
\section{Proof of the Corollary \ref{corollary:Col-3.2}}\label{sec:Ap-E}
In this appendix, we present a fundamental result showing the lower bound on the gradient for quantum circuits constructed using local operators whose decay avoids the barren plateau problem.

\begin{proof}
    The proof of this corollary is quite straightforward, since (\ref{eq:D.74}) allows us to write each of the components of the gradient of \(\mathfrak{L}_{j}\) associated with the branch and trunk sub-networks, whose square of the norm \(\left\Vert \cdot \right\Vert_{2}\) leads to the expressions
    \begin{equation}\label{eq:E.91}
        \begin{split}
            &\left\Vert \nabla_{\boldsymbol{\Theta}_{p}^{\mathcal{B} \ast}} \mathfrak{L}_{j} \left( \boldsymbol{\tilde{\mathfrak{v}}}_{\flat}, \mathbf{y}, \hat{O}^{\mathcal{B}}, \hat{O}^{\mathcal{T}}, \boldsymbol{\mathcal{W}}_{\flat}^{\mathcal{B} \ast}, \boldsymbol{\mathcal{W}}^{\mathcal{T} \ast}, \boldsymbol{\Xi}^{\mathcal{B} \ast}, \boldsymbol{\Xi}^{\mathcal{T} \ast} \right) \right\Vert_{2}^{2} \geq \left(2 \epsilon_{\mbox{hib}} \epsilon_{\mbox{cla}}\right)^{2} \\
            &\times \left\Vert \mathcal{P}_{m^{\mathcal{T}}n^{\mathcal{T}}} \bra{\Psi_{m^{\mathcal{T}}} \left( \boldsymbol{\Theta}^{\mathcal{T} \ast}_{p} \right) } \hat{O}^{\mathcal{T}} \ket{\Psi_{n^{\mathcal{T}}} \left( \boldsymbol{\Theta}^{\mathcal{T} \ast}_{p} \right) } \right\Vert_{2}^{2} \\
            &\times \left\Vert \nabla_{\boldsymbol{\Theta}_{p}^{\mathcal{B} \ast}} \left[ \mathcal{P}_{m^{\mathcal{B}}n^{\mathcal{B}}} \bra{\Psi_{m^{\mathcal{B}}} \left( \boldsymbol{\Theta}^{\mathcal{B} \ast}_{p} \right) } \hat{O}^{\mathcal{B}} \ket{\Psi_{n^{\mathcal{B}}} \left( \boldsymbol{\Theta}^{\mathcal{B} \ast}_{p} \right) } \right] \right\Vert_{2}^{2}
        \end{split}
    \end{equation}
    and
    \begin{equation}\label{eq:E.92}
        \begin{split}
            &\left\Vert \nabla_{\boldsymbol{\Theta}_{p}^{\mathcal{T} \ast}} \mathfrak{L}_{j} \left( \boldsymbol{\tilde{\mathfrak{v}}}_{\flat}, \mathbf{y}, \hat{O}^{\mathcal{B}}, \hat{O}^{\mathcal{T}}, \boldsymbol{\mathcal{W}}_{\flat}^{\mathcal{B} \ast}, \boldsymbol{\mathcal{W}}^{\mathcal{T} \ast}, \boldsymbol{\Xi}^{\mathcal{B} \ast}, \boldsymbol{\Xi}^{\mathcal{T} \ast} \right) \right\Vert_{2}^{2} \geq \left(2 \epsilon_{\mbox{hib}} \epsilon_{\mbox{cla}}\right)^{2} \\
            &\times \left\Vert \mathcal{P}_{m^{\mathcal{B}}n^{\mathcal{B}}} \bra{\Psi_{m^{\mathcal{B}}} \left( \boldsymbol{\Theta}^{\mathcal{B} \ast}_{p} \right) } \hat{O}^{\mathcal{B}} \ket{\Psi_{n^{\mathcal{B}}} \left( \boldsymbol{\Theta}^{\mathcal{B} \ast}_{p} \right) } \right\Vert_{2}^{2} \\
            &\times \left\Vert \nabla_{\boldsymbol{\Theta}_{p}^{\mathcal{T} \ast}} \left[ \mathcal{P}_{m^{\mathcal{T}}n^{\mathcal{T}}} \bra{\Psi_{m^{\mathcal{T}}} \left( \boldsymbol{\Theta}^{\mathcal{T} \ast}_{p} \right) } \hat{O}^{\mathcal{T}} \ket{\Psi_{n^{\mathcal{T}}} \left( \boldsymbol{\Theta}^{\mathcal{T} \ast}_{p} \right) } \right] \right\Vert_{2}^{2}.
        \end{split}
    \end{equation}
    Where \(\epsilon_{\mbox{hib}}\), \(\epsilon_{\mbox{cla}}\) are positive constants, which, by Theorem \ref{theorem:Teo-3.1} we can verify that
    \begin{equation}\label{eq:E.93}
        0 < \inf_{\mathbf{y} \in \mathcal{K}} \left\lbrace \inf_{\boldsymbol{\tilde{\mathfrak{v}}}_{\flat} \in \mathbb{R}^{m M \times 1}} \left\lbrace \left\Vert \mathfrak{G}_{j} \left( \boldsymbol{\tilde{\mathfrak{v}}}_{\flat}, \hat{O}^{\mathcal{B}}, \boldsymbol{\mathcal{W}}_{\flat}^{\mathcal{B} \ast}, \boldsymbol{\Theta}^{\mathcal{B} \ast}_{p} \right) \left( \mathbf{y}, \hat{O}^{\mathcal{T}}, \boldsymbol{\mathcal{W}}^{\mathcal{T} \ast}, \boldsymbol{\Theta}^{\mathcal{T} \ast}_{p} \right) - u_{j} \left( \mathbf{y} \right) \right\Vert_{2}\right\rbrace \right\rbrace \leq \epsilon_{\mbox{hib}}
    \end{equation}
    and
    \begin{equation}\label{eq:E.94}
        0 < \inf_{\mathbf{y} \in \mathcal{K}} \left\lbrace \inf_{\boldsymbol{\tilde{\mathfrak{v}}}_{\flat} \in \mathbb{R}^{m M \times 1}} \left\lbrace \left\Vert \left[ \hat{\beta}_{jp^{\prime}} \left( \boldsymbol{\tilde{\mathfrak{v}}}_{\flat}, \boldsymbol{\mathcal{W}}_{\flat}^{\mathcal{B} \ast} \right) + 1\right] \left[\tau_{p^{\prime}} \left( \mathbf{y}, \boldsymbol{\mathcal{W}}^{\mathcal{T} \ast} \right) + 1\right] + 1\right\Vert_{2} \right\rbrace \right\rbrace \leq \epsilon_{\mbox{cla}},
    \end{equation}
    since (\ref{eq:E.94}) is nothing more than a classical DeepONet, which, by Theorem \ref{theorem:Teo-A.2}, we know is capable of approximating the constant \(-1\) with arbitrary precision. Therefore, it suffices to compute the expectation of expression (\ref{eq:E.91}) to obtain:
    \begin{equation}\label{eq:E.95}
        \begin{split}
            &\mathbb{E}_{\left( \boldsymbol{\Xi}^{\mathcal{B} \ast}, \boldsymbol{\Xi}^{\mathcal{T} \ast} \right)} \left[\left\Vert \nabla_{\boldsymbol{\Theta}_{p}^{\mathcal{B} \ast}} \mathfrak{L}_{j} \left( \boldsymbol{\tilde{\mathfrak{v}}}_{\flat}, \mathbf{y}, \hat{O}^{\mathcal{B}}, \hat{O}^{\mathcal{T}}, \boldsymbol{\mathcal{W}}_{\flat}^{\mathcal{B} \ast}, \boldsymbol{\mathcal{W}}^{\mathcal{T} \ast}, \boldsymbol{\Xi}^{\mathcal{B} \ast}, \boldsymbol{\Xi}^{\mathcal{T} \ast} \right) \right\Vert_{2}^{2}\right] \geq \left(2 \epsilon_{\mbox{hib}} \epsilon_{\mbox{cla}}\right)^{2} \\
            &\times \mathbb{E}_{\left( \boldsymbol{\Xi}^{\mathcal{B} \ast}, \boldsymbol{\Xi}^{\mathcal{T} \ast} \right)} \left[\left\Vert \mathcal{P}_{m^{\mathcal{T}}n^{\mathcal{T}}} \bra{\Psi_{m^{\mathcal{T}}} \left( \boldsymbol{\Theta}^{\mathcal{T} \ast}_{p} \right) } \hat{O}^{\mathcal{T}} \ket{\Psi_{n^{\mathcal{T}}} \left( \boldsymbol{\Theta}^{\mathcal{T} \ast}_{p} \right) } \right\Vert_{2}^{2} \right. \\
            &\times \left. \left\Vert \nabla_{\boldsymbol{\Theta}_{p}^{\mathcal{B} \ast}} \left[ \mathcal{P}_{m^{\mathcal{B}}n^{\mathcal{B}}} \bra{\Psi_{m^{\mathcal{B}}} \left( \boldsymbol{\Theta}^{\mathcal{B} \ast}_{p} \right) } \hat{O}^{\mathcal{B}} \ket{\Psi_{n^{\mathcal{B}}} \left( \boldsymbol{\Theta}^{\mathcal{B} \ast}_{p} \right) } \right] \right\Vert_{2}^{2}\right],
        \end{split}
    \end{equation}
    which, in the initialization regime, the statistical independence of the parameters of the Branch (\(\Xi^{\mathcal{B} \ast}\)) and the Trunk (\(\Xi^{\mathcal{T} \ast}\)) allows us to factorize the cross expectations on the right-hand side of (\ref{eq:E.95}), that is,
    \begin{equation}\label{eq:E.96}
        \begin{split}
            &\mathbb{E}_{\left( \boldsymbol{\Xi}^{\mathcal{B} \ast}, \boldsymbol{\Xi}^{\mathcal{T} \ast} \right)} \left[\left\Vert \mathcal{P}_{m^{\mathcal{T}}n^{\mathcal{T}}} \bra{\Psi_{m^{\mathcal{T}}} \left( \boldsymbol{\Theta}^{\mathcal{T} \ast}_{p} \right) } \hat{O}^{\mathcal{T}} \ket{\Psi_{n^{\mathcal{T}}} \left( \boldsymbol{\Theta}^{\mathcal{T} \ast}_{p} \right) } \right\Vert_{2}^{2} \right.\\
            &\times \left. \left\Vert \nabla_{\boldsymbol{\Theta}_{p}^{\mathcal{B} \ast}} \left[ \mathcal{P}_{m^{\mathcal{B}}n^{\mathcal{B}}} \bra{\Psi_{m^{\mathcal{B}}} \left( \boldsymbol{\Theta}^{\mathcal{B} \ast}_{p} \right) } \hat{O}^{\mathcal{B}} \ket{\Psi_{n^{\mathcal{B}}} \left( \boldsymbol{\Theta}^{\mathcal{B} \ast}_{p} \right) } \right] \right\Vert_{2}^{2}\right] = \\
            &\mathbb{E}_{\boldsymbol{\Xi}^{\mathcal{T} \ast}} \left[\left\Vert \mathcal{P}_{m^{\mathcal{T}}n^{\mathcal{T}}} \bra{\Psi_{m^{\mathcal{T}}} \left( \boldsymbol{\Theta}^{\mathcal{T} \ast}_{p} \right) } \hat{O}^{\mathcal{T}} \ket{\Psi_{n^{\mathcal{T}}} \left( \boldsymbol{\Theta}^{\mathcal{T} \ast}_{p} \right) } \right\Vert_{2}^{2}\right] \\
            &\times \mathbb{E}_{\boldsymbol{\Xi}^{\mathcal{B} \ast}} \left[ \left\Vert \nabla_{\boldsymbol{\Theta}_{p}^{\mathcal{B} \ast}} \left[ \mathcal{P}_{m^{\mathcal{B}}n^{\mathcal{B}}} \bra{\Psi_{m^{\mathcal{B}}} \left( \boldsymbol{\Theta}^{\mathcal{B} \ast}_{p} \right) } \hat{O}^{\mathcal{B}} \ket{\Psi_{n^{\mathcal{B}}} \left( \boldsymbol{\Theta}^{\mathcal{B} \ast}_{p} \right) } \right] \right\Vert_{2}^{2}\right].
        \end{split}
    \end{equation}
    Therefore, it is evident that:
    \begin{equation}\label{eq:E.97}
        0 \leq \mathbb{E}_{\boldsymbol{\Xi}^{\mathcal{T} \ast}} \left[\left\Vert \mathcal{P}_{m^{\mathcal{T}}n^{\mathcal{T}}} \bra{\Psi_{m^{\mathcal{T}}} \left( \boldsymbol{\Theta}^{\mathcal{T} \ast}_{p} \right) } \hat{O}^{\mathcal{T}} \ket{\Psi_{n^{\mathcal{T}}} \left( \boldsymbol{\Theta}^{\mathcal{T} \ast}_{p} \right) } \right\Vert_{2}^{2}\right] \equiv \mathfrak{c}^{\mathcal{T}}.
    \end{equation}
    Furthermore, it is reported in \cite{McClean201800, Larocca202500} that local Dynamical Lie Algebras for shallow circuits imply that their depth is of the order \(\mathcal{N}^{\mathcal{B}} = \mathcal{O}\left(\log{N^{\mathcal{B}}}\right)\), which leads to:
    \begin{equation}\label{eq:E.98}
        \mathbb{E}_{\boldsymbol{\Xi}^{\mathcal{B} \ast}} \left[\left\Vert \nabla_{\boldsymbol{\Theta}_{p}^{\mathcal{B} \ast}} \left[\mathcal{P}_{m^{\mathcal{B}}n^{\mathcal{B}}} \bra{\Psi_{m^{\mathcal{B}}} \left( \boldsymbol{\Theta}^{\mathcal{B} \ast}_{p} \right) } \hat{O}^{\mathcal{B}} \ket{\Psi_{n^{\mathcal{B}}} \left( \boldsymbol{\Theta}^{\mathcal{B} \ast}_{p} \right) } \right] \right\Vert_{2}^{2}\right] \geq \mathcal{O} \left(\frac{1}{\mbox{poly} \left(N^{\mathcal{B}}\right)}\right).
    \end{equation}
    Thus, we can use (\ref{eq:E.96})-(\ref{eq:E.98}) to obtain:
    \begin{equation}\label{eq:E.99}
        \begin{split}
            &\mathbb{E}_{\left( \boldsymbol{\Xi}^{\mathcal{B} \ast}, \boldsymbol{\Xi}^{\mathcal{T} \ast} \right)} \left[\left\Vert \nabla_{\boldsymbol{\Theta}_{p}^{\mathcal{B} \ast}} \mathfrak{L}_{j} \left( \boldsymbol{\tilde{\mathfrak{v}}}_{\flat}, \mathbf{y}, \hat{O}^{\mathcal{B}}, \hat{O}^{\mathcal{T}}, \boldsymbol{\mathcal{W}}_{\flat}^{\mathcal{B} \ast}, \boldsymbol{\mathcal{W}}^{\mathcal{T} \ast}, \boldsymbol{\Xi}^{\mathcal{B} \ast}, \boldsymbol{\Xi}^{\mathcal{T} \ast} \right) \right\Vert_{2}^{2}\right] \geq \mathcal{O} \left(\frac{1}{\mbox{poly} \left(N^{\mathcal{B}}\right)}\right),
        \end{split}
    \end{equation}
    and through an analogous procedure
    \begin{equation}\label{eq:E.100}
        \begin{split}
            &\mathbb{E}_{\left( \boldsymbol{\Xi}^{\mathcal{B} \ast}, \boldsymbol{\Xi}^{\mathcal{T} \ast} \right)} \left[\left\Vert \nabla_{\boldsymbol{\Theta}_{p}^{\mathcal{T} \ast}} \mathfrak{L}_{j} \left( \boldsymbol{\tilde{\mathfrak{v}}}_{\flat}, \mathbf{y}, \hat{O}^{\mathcal{B}}, \hat{O}^{\mathcal{T}}, \boldsymbol{\mathcal{W}}_{\flat}^{\mathcal{B} \ast}, \boldsymbol{\mathcal{W}}^{\mathcal{T} \ast}, \boldsymbol{\Xi}^{\mathcal{B} \ast}, \boldsymbol{\Xi}^{\mathcal{T} \ast} \right) \right\Vert_{2}^{2}\right] \geq \mathcal{O} \left(\frac{1}{\mbox{poly} \left(N^{\mathcal{T}}\right)}\right).
        \end{split}
    \end{equation}
    Thus, combining the gradient bounds of the branch (\ref{eq:E.99}) and trunk (\ref{eq:E.100}) for the \(j\)-th component, and using the definition of the supremum norm over the output space, the validity of inequality (\ref{eq:48}) is directly concluded.
\end{proof}
\section{Proof of the Theorem \ref{theorem:Teo-3.2}}\label{sec:Ap-F}
In this appendix, we determine the minimum number of sensors required for our hybrid QPI-DeepONet-MAC architecture to approximate the operator \( \boldsymbol{\mathcal{G}} \) with arbitrary precision—an operator that maps a function \( \mathbf{u} \in C_{\mathbb{R}^{M}}^{\infty} \left( \mathcal{K} \right) \) serving as a solution to the problem defined by equations (\ref{eq:04})–(\ref{eq:07}).

To achieve this, we begin by assuming that \(\mathbb{B}\) is Euclidean of dimension \(d\). Now, let us see that given \(x_{1}, \ldots, x_{m}\in \mathbb{K}\), all distinct, and \(\mathbf{f} \in C_{\mathbb{R}^{M}}\left(\mathbb{K}\right)\), there exists a subspace \(\mathbb{X}_{m}\) of \(C_{\mathbb{R}^{M}}\left(\mathbb{K}\right)\) of dimension less than or equal to \(m\) such that:
\begin{equation}\label{eq:F.101}
    \mathbf{f} \left(x_{i}\right) = \mathbf{g} \left(x_{i}\right),
\end{equation}
for all \(i \in \mathcal{I}_{m} = \{1, \dots, m\}\) and some \(\mathbf{g} \in \mathbb{X}_{m}\). Indeed, we construct \(m\) continuous scalar functions \(\xi_{1}, \dots, \xi_{m}: \mathbb{K} \to \mathbb{R}\) satisfying the following condition:
\begin{equation}\label{eq:F.102}
    \xi_{i} \left(x_{j}\right) = \delta_{ij}.
\end{equation}
It is clear that \(\dim{\left( \operatorname{span}\left\lbrace \xi_{1}, \dots,\xi_{m}\right\rbrace\right)} = m\). Next, we define \(\mathbf{g} : \mathbb{K}\to \mathbb{R}^{M}\) by
\begin{equation}\label{eq:F.103}
	\mathbf{g} \left(x\right) \coloneq \sum_{i \in \mathcal{I}_{m}} \xi_{i} \left(x\right)\mathbf{f} \left(x_{i}\right),
\end{equation}
consequently,
\begin{equation}\label{eq:F.104}
    \mathbf{g} \left(x_{j}\right) = \sum_{i \in \mathcal{I}_{m}} \delta_{ij}\mathbf{f}\left(x_{i}\right) = \mathbf{f} \left(x_{j}\right)
\end{equation}
and \(\mathbf{g} \in C_{\mathbb{R}^{M}}\left(\mathbb{K}\right)\). Finally, to obtain our objective, we take
\begin{equation}\label{eq:F.105}
	\mathbb{X}_{m} = \operatorname{span}\left\lbrace \xi_{1}\mathbf{f} \left(x_{1}\right), \dots,\xi_{m} \mathbf{f} \left(x_{m}\right) \right\rbrace, 
\end{equation} 
which evidently satisfies the desired conditions. For the sake of notational simplicity, in what follows we shall write \(\mathbf{g} = \mathbf{f}_{\operatorname{int},m}\). A detail worth highlighting is that if we consider the discretization of \(\mathbf{f}\) by the points \(x_{1}, \dots, x_{m}\), these very same points define the same discretization of \(\mathbf{f}_{\operatorname{int},m}\) due to (\ref{eq:F.101}). This allows us to state the theorem \ref{theorem:Teo-3.2}.
\begin{proof}
    The proof begins by considering
    \begin{equation}\label{eq:F.106}
    	\sup_{\boldsymbol{\mathfrak{v}} \in \mathbb{V}^{M}} \left\lbrace \inf_{\boldsymbol{\tilde{\mathfrak{v}}}_{\flat}} \left\Vert \boldsymbol{\mathfrak{G}} \left( \boldsymbol{\tilde{\mathfrak{v}}}_{\flat}, \hat{O}^{\mathcal{B}}, \boldsymbol{\mathcal{W}}_{\flat}^{\mathcal{B} \ast}, \boldsymbol{\Xi}^{\mathcal{B} \ast} \right) \left( \cdot, \hat{O}^{\mathcal{T}}, \boldsymbol{\mathcal{W}}^{\mathcal{T} \ast}, \boldsymbol{\Xi}^{\mathcal{T} \ast} \right) - \boldsymbol{\mathcal{G}} \left( \boldsymbol{\mathfrak{v}} \right) \left( \cdot \right) \right\Vert_{\infty,M}\right\rbrace,
    \end{equation}
    where the infimum is taken over the set of all discretizations of \(\mathfrak{v}\) by \(m\) strictly distinct sensors chosen in an arbitrary manner. Recalling that the discretization of each \(\mathfrak{v}\) by \(m\) sensors is also the discretization, by the same sensors, of \(\mathfrak{v}_{\operatorname{int},m}\), we can conveniently apply the triangle inequality to the expression inside the infimum of (\ref{eq:F.106}) to obtain:
    \begin{equation}\label{eq:F.107}
        \begin{split}
            &\left\Vert \boldsymbol{\mathfrak{G}} \left( \boldsymbol{\tilde{\mathfrak{v}}}_{\flat}, \hat{O}^{\mathcal{B}}, \boldsymbol{\mathcal{W}}_{\flat}^{\mathcal{B} \ast}, \boldsymbol{\Xi}^{\mathcal{B} \ast} \right) \left( \cdot, \hat{O}^{\mathcal{T}}, \boldsymbol{\mathcal{W}}^{\mathcal{T} \ast}, \boldsymbol{\Xi}^{\mathcal{T} \ast} \right) - \boldsymbol{\mathcal{G}} \left( \boldsymbol{\mathfrak{v}} \right) \left( \cdot \right) \right\Vert_{\infty,M}\le \\
            &\hspace{2.0cm}\left\Vert \boldsymbol{\mathfrak{G}} \left( \boldsymbol{\tilde{\mathfrak{v}}}_{\flat}, \hat{O}^{\mathcal{B}}, \boldsymbol{\mathcal{W}}_{\flat}^{\mathcal{B} \ast}, \boldsymbol{\Xi}^{\mathcal{B} \ast} \right) \left( \cdot, \hat{O}^{\mathcal{T}}, \boldsymbol{\mathcal{W}}^{\mathcal{T} \ast}, \boldsymbol{\Xi}^{\mathcal{T} \ast} \right) - \boldsymbol{\mathcal{G}} \left( \boldsymbol{\mathfrak{v}}_{\operatorname{int},m} \right) \left( \cdot \right) \right\Vert_{\infty,M} \\
            &\hspace{6.70cm} + \left\Vert \boldsymbol{\mathcal{G}} \left( \boldsymbol{\mathfrak{v}}\right) \left( \cdot \right) - \boldsymbol{\mathcal{G}} \left( \boldsymbol{\mathfrak{v}}_{\operatorname{int},m} \right) \left( \cdot \right) \right\Vert_{\infty,M}.
        \end{split}
    \end{equation}
    Now, since \(\epsilon > 0\) is the approximation error of the hybrid model, \(\boldsymbol{\mathcal{G}}\) is Lipschitz continuous with respect to the second variable, and thanks to the Hille-Yosida theorem \cite{Brezis201100}, there exist constants \(\mathfrak{C}_{L} > 0\) and \(\omega \in \mathbb{R}\) such that from
    \begin{equation}\label{eq:F.108}
        \begin{split}
            &\sup_{\boldsymbol{\mathfrak{v}} \in \mathbb{V}^{M}} \left\lbrace \inf_{\boldsymbol{\mathfrak{u}}\in \tilde{\mathbb{X}}_{m}} \left\Vert \boldsymbol{\mathfrak{G}} \left( \boldsymbol{\tilde{\mathfrak{u}}}_{\flat}, \hat{O}^{\mathcal{B}}, \boldsymbol{\mathcal{W}}_{\flat}^{\mathcal{B} \ast}, \boldsymbol{\Xi}^{\mathcal{B} \ast} \right) \left( \cdot, \hat{O}^{\mathcal{T}}, \boldsymbol{\mathcal{W}}^{\mathcal{T} \ast}, \boldsymbol{\Xi}^{\mathcal{T} \ast} \right) - \boldsymbol{\mathcal{G}} \left( \boldsymbol{\mathfrak{v}} \right) \left( \cdot \right) \right\Vert_{\infty,M}\right\rbrace \le \\
            &\hspace{5.0cm} \epsilon + \mathfrak{C}_{L} \exp{\omega t} \sup_{\boldsymbol{\mathfrak{v}} \in \mathbb{V}^{M}} \left\lbrace \inf_{\boldsymbol{\mathfrak{u}} \in \tilde{\mathbb{X}}_{m}} \left\Vert \boldsymbol{\mathfrak{u}} - \boldsymbol{\mathfrak{v}} \right\Vert_{\infty,M}\right\rbrace,
        \end{split}
    \end{equation}
    where \(\boldsymbol{\tilde{\mathfrak{u}}}_{\flat}\) denotes the discretization of \(\boldsymbol{\mathfrak{u}}\) by the \(m\) sensors and \(\tilde{\mathbb{X}}_{m}\) is a vector subspace of \(C_{\mathbb{R}^{M}}\left(\mathbb{K}\right)\) of dimension less than or equal to \(m\), the existence of at least one such subspace is guaranteed by the construction in (\ref{eq:F.105}).
    Now, in order to utilize the properties of the Kolmogorov \(n\)-width \cite{Pinkus201200}, we take the infimum in (\ref{eq:F.108}) over all subspaces \(\tilde{\mathbb{X}}_{m}\) of \(C_{\mathbb{R}^{M}}\left(\mathbb{K}\right)\) with \(\dim{\tilde{\mathbb{X}}_{m}} \le m\), to conclude that:
    \begin{equation}\label{eq:F.109}
        \begin{split}
            &\displaystyle\inf_{\tilde{\mathbb{X}}_{m}} \left\lbrace \sup_{\boldsymbol{\mathfrak{v}} \in \mathbb{V}^{M}} \left\lbrace \inf_{\boldsymbol{\mathfrak{u}}\in \tilde{\mathbb{X}}_{m}} \left\Vert \boldsymbol{\mathfrak{G}} \left( \boldsymbol{\tilde{\mathfrak{u}}}_{\flat}, \hat{O}^{\mathcal{B}}, \boldsymbol{\mathcal{W}}_{\flat}^{\mathcal{B} \ast}, \boldsymbol{\Xi}^{\mathcal{B} \ast} \right) \left( \cdot, \hat{O}^{\mathcal{T}}, \boldsymbol{\mathcal{W}}^{\mathcal{T} \ast}, \boldsymbol{\Xi}^{\mathcal{T} \ast} \right) - \boldsymbol{\mathcal{G}} \left( \boldsymbol{\mathfrak{v}} \right) \left( \cdot \right) \right\Vert_{\infty,M}\right\rbrace \right\rbrace \le \\
            &\hspace{5.0cm} \epsilon + \mathfrak{C}_{L} \exp{\omega t} \displaystyle\inf_{\tilde{\mathbb{X}}_{m}} \left \lbrace\sup_{\boldsymbol{\mathfrak{v}} \in \mathbb{V}^{M}} \left\lbrace \inf_{\boldsymbol{\mathfrak{u}} \in \tilde{\mathbb{X}}_{m}} \left\Vert \boldsymbol{\mathfrak{u}} - \boldsymbol{\mathfrak{v}} \right\Vert_{\infty,M}\right\rbrace \right\rbrace.
        \end{split}
    \end{equation}
    Finally, from the inequality above, thanks to the Kolmogorov \(n\)-width property
    \begin{equation}\label{eq:F.110}
        \inf_{\tilde{\mathbb{X}}_{m}} \left\lbrace \sup_{\boldsymbol{\mathfrak{v}}\in \mathbb{V}^{M}}\left\lbrace \inf_{\boldsymbol{\mathfrak{u}}\in \tilde{\mathbb{X}}_{m}} \left\Vert \boldsymbol{\mathfrak{u}}-\boldsymbol{\mathfrak{v}}\right\Vert_{\infty,M} \right\rbrace \right\rbrace = \mathcal{O}\left(m^{-k/d}\right),
    \end{equation}
    and by the definition of the infimum, we have that there exists a subspace \(\mathbb{Y}_{m}\) of \(C_{\mathbb{R}^{M}}\left(\mathbb{K}\right)\) with \(\dim{\mathbb{Y}_{m}} \le m\) such that:
    \begin{equation}\label{eq:F.111}
        \begin{split}
            &\sup_{\boldsymbol{\mathfrak{v}} \in \mathbb{V}^{M}} \left\lbrace \inf_{\boldsymbol{\mathfrak{u}}\in \tilde{\mathbb{X}}_{m}} \left\Vert \boldsymbol{\mathfrak{G}} \left( \boldsymbol{\tilde{\mathfrak{u}}}_{\flat}, \hat{O}^{\mathcal{B}}, \boldsymbol{\mathcal{W}}_{\flat}^{\mathcal{B} \ast}, \boldsymbol{\Xi}^{\mathcal{B} \ast} \right) \left( \cdot, \hat{O}^{\mathcal{T}}, \boldsymbol{\mathcal{W}}^{\mathcal{T} \ast}, \boldsymbol{\Xi}^{\mathcal{T} \ast} \right) - \boldsymbol{\mathcal{G}} \left( \boldsymbol{\mathfrak{v}} \right) \left( \cdot \right) \right\Vert_{\infty,M}\right\rbrace \le \\
            &\hspace{7.25cm} \epsilon + \mathfrak{C}_{L} \exp{\omega t} \mathcal{O}\left(m^{-k/d}\right).
        \end{split}
    \end{equation}
    Thus, by imposing that (\ref{eq:F.106}) is bounded from above by \(\tilde{\epsilon}\), we can rewrite (\ref{eq:F.111}) in the form
    \begin{equation}\label{eq:F.112}
        \epsilon + C_{\mbox{Int}} \mathfrak{C}_{L} \exp{\omega t} m^{-k/d} \leq \tilde{\epsilon},
    \end{equation}
    where \(C_{\mbox{Int}}\) is a constant that absorbs the geometric characteristics of the compact domain \(\mathbb{K}\), the function regularity \(k\), and the spatial dimension \(d\), ensuring that the inequality is strictly quantitative for any mesh or sensor choice. Finally, via (\ref{eq:F.112}) we obtain expression (\ref{eq:50}), which mathematically represents the minimum number of sensors \(m\) required to guarantee system observability. This proves the theorem, establishing that for unstable dynamical systems (\(\omega > 0\)) the requirement on the number of sensors scales exponentially with the evolution time.
\end{proof}
\section{Nomenclature}\label{sec:Ap-G}
To facilitate readability and ensure clarity regarding the mathematical notation employed throughout this work, a comprehensive summary of the symbols is presented in Table \ref{tab:nomenclature}.
\begin{xltabular}{\textwidth}{>{\raggedright\arraybackslash}p{0.32\textwidth} >{\raggedright\arraybackslash}X}
    \caption{Summary of the nomenclature and symbols for the QPI-DeepONet-MAC architecture.} \label{tab:nomenclature} \\
    \toprule
    \textbf{Symbol} & \textbf{Description} \\
    \midrule
    \endhead 

    \bottomrule
    \endlastfoot 

    $D, \mathcal{K}, \mathbb{V}$ & Spatio-temporal domain $D = \Omega \times \mathbb{R}^{+}$, its compact subset $\mathcal{K}$, and function space compact subset $\mathbb{V}$. \\
    $\mathcal{I}_{X}$ & General index set of the form $\{1, \dots, X\}$ with \(X \in \mathbb{N}^{\ast}\), e.g.: $\mathcal{I}_{M}, \mathcal{I}_{P}, \mathcal{I}_{N_1}, \mathcal{I}_{N_2}, \mathcal{I}_{N_{\nu}}$. \\
    $\nu \in \{\mathcal{B}, \mathcal{T}\}$ & Index for the Branch ($\mathcal{B}$) or Trunk ($\mathcal{T}$) subnetworks. \\
    \midrule
    \addlinespace
    \multicolumn{2}{l}{\textit{Classical Variables and Operators}} \\
    \addlinespace
    \midrule
    $M, m$ & Output dimension of the target PDE solution, and number of input sensors. \\
    $\mathbf{u}, \boldsymbol{\mathfrak{v}}$ & Target PDE solution function, and input parameter function. \\
    $\tilde{\mathfrak{v}}_{\flat}, \mathbf{y}$ & Discretized/vectorized branch input, and trunk continuous spatio-temporal coordinates. \\
    $\boldsymbol{\mathcal{G}}, \mathbf{DeepONet}$ & Target continuous operator, and the classical DeepONet operator. \\
    $\boldsymbol{\mathcal{W}}_{\flat}^{\mathcal{B}}, \boldsymbol{\mathcal{W}}^{\mathcal{T}}$ & Sets of classical weights and biases for the Branch and Trunk networks. \\
    $\hat{\boldsymbol{\beta}}, \boldsymbol{\tau}$ & Branch latent coefficients and Trunk basis functions. \\
    \midrule
    \addlinespace
    \multicolumn{2}{l}{\textit{Quantum and Hybrid MAC Components}} \\
    \addlinespace
    \midrule
    $N^{\nu}, \mathcal{N}^{\nu}$ & Number of qubits and circuit depth for the quantum module. \\
    $\boldsymbol{\Theta}_{p}^{\nu}, \boldsymbol{\Xi}^{\nu}$ & Local parameters of the $p$-th quantum node, and the global parameter tensor. \\
    $|\Psi(\boldsymbol{\Theta}_{p}^{\nu})\rangle, \widehat{U}(\boldsymbol{\Theta}_{p}^{\nu})$ & Parameterized quantum state and its corresponding unitary operator. \\
    $\hat{O}^{\nu}, \mathbb{O}^{\nu}$ & Hermitian quantum observable acting on $\mathcal{H}_{\mathcal{Q}}^{\nu}$, and its corresponding set. \\
    $\mathfrak{B}_{jp^{\prime}}, \mathfrak{T}_{p^{\prime}}, \mathfrak{b}_j$ & Multiplicative Branch component, Trunk component, and Additive bias of the MAC coupling. \\
    $\boldsymbol{\mathfrak{G}}$ & Output operator of the hybrid QPI-DeepONet-MAC model. \\
    \midrule
    \addlinespace
    \multicolumn{2}{l}{\textit{Quantum Informational Dynamics}} \\
    \addlinespace
    \midrule
    $\rho_{\mathcal{C}}^{\nu}$ & Density matrices of the classical subsystems. \\
    $\rho_{\mathcal{C}}^{\nu}, \varrho_{\mathcal{Q}}^{\nu}$ & Density matrix of the quantum subsystem, and its diagonalized form. \\
    $\mathcal{C}_{p}, \mathcal{D}_{p}$ & Quantum Coherence of the reduced state, and the Decoherence rate. \\
    $S_{p}, F_{p}$ & Von Neumann Entropy, and Fidelity between subsequent steps. \\
    \midrule
    \addlinespace
    \multicolumn{2}{l}{\textit{Loss Functions and Regularization}} \\
    \addlinespace
    \midrule
    $\mathcal{L}, \tilde{\mathcal{L}}, \mathfrak{L}$ & Classical Physics-Informed loss, Hybrid Physics-Informed loss, and theoretical loss for gradients. \\
    $\mathcal{L}_{IC}, \mathcal{L}_{BC}$ & Loss function components for Initial Conditions (IC) and Boundary Conditions (BC). \\
    $\mathcal{L}_{PDE}, \mathcal{L}_{SOL}$ & Loss function components for PDE residuals and known analytical/empirical Solutions. \\
    $\mathcal{L}_{INF}, \mathcal{L}_{TOT}$ & Informational regularization term, and the Total hybrid optimization loss function. \\
    $\lambda_{IC}, \lambda_{BC}, \lambda_{PDE}, \lambda_{SOL}$ & Hyperparameters weighting the physical constraints and \\
    $\lambda_{COE}, \lambda_{FID}$ & informational regularization (Coherence, Fidelity). \\
\end{xltabular}
\end{appendices}
\backmatter
\bmhead{Acknowledgments}
First, I thank God the Father, infinite goodness, God the Son, redeemer of the world, and God the Holy Spirit, who—together with the Virgin Mary and Saint Joseph—are my protection, my source of inspiration, and my strength on this path of knowledge. I also thank my family, for they have made me who I am today.

\section*{Declarations}
\begin{itemize}
    \item \textbf{Funding:} I would like to thank the National Council for Scientific and Technological Development (CNPq), which initially funded this work through the Institutional Capacity Building Program (PCI), process no. 301066/2025-6, granted by the Ministry of Science, Technology and Innovation (MCTI), and subsequently through the Junior Postdoctoral Fellowship (PDJ) program under process no. 150351/2026-7.
    \item \textbf{Conflict of interest/Competing interests:} The authors declare no competing interests.
    \item \textbf{Ethics approval and consent to participate:} Not applicable.
    \item \textbf{Consent for publication:} Not applicable.
    \item \textbf{Data and code availability:} The code and data supporting the findings of this study are available from the corresponding author upon reasonable request.
    \item \textbf{Author contributions:} Author Said Lantigua (SL) conceived the project; subsequently, authors José Valencia (JV), Gilson Giraldi (GG), Renato Portugal (RP), and Jonas Maziero (JM) contributed to the development of the research. Specifically, SL and JV carried out the formal derivations that provide the mathematical foundation for the project, under the supervision of GG, RP, and JM. Furthermore, SL and JV actively worked on drafting the initial version of the manuscript as well as subsequent versions, which were improved and revised by GG, RP, and JM. Throughout the research development process, all authors discussed and reviewed the manuscript.
\end{itemize}
\bibliography{sn-bibliography}

\end{document}